\documentclass[11pt,letterpaper]{article}
\usepackage[lmargin=1.0in,rmargin=1.0in,bottom=1.0in,top=1.0in,twoside=False]{geometry}

\usepackage[utf8]{inputenc}
\usepackage[T1]{fontenc}
\usepackage{lmodern}

\usepackage{fullpage,amssymb,amsmath}
\usepackage{graphicx}
\usepackage{enumerate}
\usepackage{tikz}
\usetikzlibrary{shapes}
\usepackage{subcaption}

\usepackage{xcolor}
\usepackage{mathtools}
\usepackage{microtype}
\usepackage{amsfonts}
\usepackage{comment}
\usepackage[english]{babel}
\usepackage{mathrsfs}
\usepackage[ruled,vlined]{algorithm2e}
\usepackage{blindtext}
\usepackage{thmtools}
\usepackage{thm-restate}
\usepackage[numbers,sort&compress]{natbib}

\usepackage{trimspaces}
\usepackage{nccfoots}
\usepackage{setspace}
\usepackage{inconsolata}
\usepackage{libertine}
\usepackage[absolute]{textpos}

\usepackage[backgroundcolor = blue!50]{todonotes}
\usepackage{longtable}

\definecolor{blue}{rgb}{0.1,0.2,0.5}
\definecolor{brown}{rgb}{0.6,0.6,0.2}
\usepackage[ocgcolorlinks, linkcolor={blue}, citecolor={brown}]{hyperref}
\usepackage[amsmath,amsthm,thmmarks,hyperref]{ntheorem}
\usepackage{enumerate}
\usepackage{latexsym}
\usepackage{thm-restate}

\usepackage[nameinlink]{cleveref}

\crefformat{page}{#2page~#1#3}%
\Crefformat{page}{#2Page~#1#3}%
\crefformat{equation}{#2(#1)#3}%
\Crefformat{equation}{#2(#1)#3}%
\crefformat{figure}{#2Figure~#1#3}%
\Crefformat{figure}{#2Figure~#1#3}%
\crefformat{section}{#2Section~#1#3}
\Crefformat{section}{#2Section~#1#3}
\crefformat{chapter}{#2Chapter~#1#3}
\Crefformat{chapter}{#2Chapter~#1#3}
\crefformat{chapter*}{#2Chapter~#1#3}
\Crefformat{chapter*}{#2Chapter~#1#3}
\crefformat{part}{#2Part~#1#3}
\Crefformat{part}{#2Part~#1#3}
\crefformat{enumi}{#2(#1)#3}
\Crefformat{enumi}{#2(#1)#3}

\usepackage[mathlines]{lineno}

\newcommand*\patchAmsMathEnvironmentForLineno[1]{%
  \expandafter\let\csname old#1\expandafter\endcsname\csname #1\endcsname
  \expandafter\let\csname oldend#1\expandafter\endcsname\csname end#1\endcsname
  \renewenvironment{#1}%
     {\linenomath\csname old#1\endcsname}%
     {\csname oldend#1\endcsname\endlinenomath}}%
\newcommand*\patchBothAmsMathEnvironmentsForLineno[1]{%
  \patchAmsMathEnvironmentForLineno{#1}%
  \patchAmsMathEnvironmentForLineno{#1*}}%
\AtBeginDocument{%
  \patchBothAmsMathEnvironmentsForLineno{equation}%
  \patchBothAmsMathEnvironmentsForLineno{align}%
  \patchBothAmsMathEnvironmentsForLineno{flalign}%
  \patchBothAmsMathEnvironmentsForLineno{alignat}%
  \patchBothAmsMathEnvironmentsForLineno{gather}%
  \patchAmsMathEnvironmentForLineno{split}
  \patchBothAmsMathEnvironmentsForLineno{multline}}

\theoremnumbering{arabic}
\theoremstyle{plain}
\theoremsymbol{}
\theorembodyfont{\itshape}
\theoremheaderfont{\normalfont\bfseries}
\theoremseparator{.}

\newtheorem{theorem}{Theorem}
\crefformat{theorem}{#2Theorem~#1#3}
\Crefformat{theorem}{#2Theorem~#1#3}
\crefformat{cthmin}{#2Theorem~#1#3}
\Crefformat{cthmin}{#2Theorem~#1#3}

\crefformat{definition}{#2Definition~#1#3}
\Crefformat{definition}{#2Definition~#1#3}

\newcommand{\newtheoremwithcrefformat}[2]{%
  \newtheorem{#1}[theorem]{#2}%
  \crefformat{#1}{##2\MakeUppercase#1~##1##3}%
  \Crefformat{#1}{##2\MakeUppercase#1~##1##3}%
}
\newcommand{\newseptheoremwithcrefformat}[2]{%
  \newtheorem{#1}{#2}%
  \crefformat{#1}{##2\MakeUppercase#1~##1##3}%
  \Crefformat{#1}{##2\MakeUppercase#1~##1##3}%
}
\newcommand{\newclaimwithcrefformat}[2]{%
  \newtheorem{#1}{#2}[theorem]%
  \crefformat{#1}{##2\MakeUppercase#1~##1##3}%
  \Crefformat{#1}{##2\MakeUppercase#1~##1##3}%
}

\newtheoremwithcrefformat{lemma}{Lemma}

\newtheoremwithcrefformat{proposition}{Proposition}
\newtheoremwithcrefformat{observation}{Observation}
\newseptheoremwithcrefformat{conjecture}{Conjecture}
\newtheoremwithcrefformat{corollary}{Corollary}

\newtheoremwithcrefformat{remark}{Remark}
\newseptheoremwithcrefformat{question}{Question}
\newclaimwithcrefformat{claim}{Claim}

\newtheorem{cthmin}{Theorem}%
\newenvironment{cthm}[1]
  {\renewcommand\thecthmin{#1}\cthmin}
  {\endcthmin}

\theoremstyle{definition}
\theorembodyfont{\upshape}
\newtheoremwithcrefformat{definition}{Definition}

\theoremstyle{nonumberplain}
\theoremheaderfont{\scshape}
\theorembodyfont{\normalfont}
\theoremsymbol{\ensuremath{\square}}

\theoremsymbol{\ensuremath{\lrcorner}}

\newcommand{\cD}{\mathcal{D}}

\newcommand{\cH}{\mathcal{H}}

\newcommand{\cP}{\mathcal{P}}
\newcommand{\cQ}{\mathcal{Q}}

\newcommand{\cT}{\mathcal{T}}

\newcommand{\cX}{\mathcal{X}}
\newcommand{\cY}{\mathcal{Y}}
\newcommand{\cZ}{\mathcal{Z}}

\newcommand{\tG}{\widetilde{G}}
\newcommand{\tH}{\widetilde{H}}
\newcommand{\tL}{\widetilde{L}}

\newcommand{\tF}{\widetilde{F}}
\newcommand{\tpi}{\widetilde{\pi}}
\newcommand{\Cg}{\mathscr{C}_g}

\newcommand{\ork}[1]{\mathrm{OR}_{#1}}
\newcommand{\nand}[1]{\mathrm{NAND}_{#1}}
\DeclareMathOperator{\dist}{dist}

\newcommand{\N}{\mathbb{N}}

\newcommand{\R}{\mathbb{R}}
\newcommand{\X}{\mathbb{X}}
\newcommand{\Y}{\mathbb{Y}}
\newcommand{\vphi}{\varphi}
\newcommand{\tphi}{\widetilde{\vphi}}
\newcommand{\eps}{\varepsilon}
\renewcommand{\epsilon}{\varepsilon}
\newcommand{\Oh}{\mathcal{O}}
\newcommand{\DD}{\mathbb{D}}

\newcommand{\PP}{\mathbb{P}}
\newcommand{\oX}{\overline{\cX}}
\newcommand{\oY}{\overline{\cY}}
\newcommand{\oZ}{\overline{\cZ}}
\newcommand{\oalpha}{\overline{\alpha}}
\newcommand{\obeta}{\overline{\beta}}
\newcommand{\ogamma}{\overline{\gamma}}
\newcommand{\dd}{\textbf{d}}
\newcommand{\od}{\overline{\textbf{d}}}
\newcommand{\x}{\textbf{x}}
\newcommand{\y}{\textbf{y}}
\newcommand{\z}{\textbf{z}}

\renewcommand{\leq}{\leqslant}
\renewcommand{\geq}{\geqslant}
\renewcommand{\le}{\leqslant}
\renewcommand{\ge}{\geqslant}

\newcommand{\homo}[1]{\textsc{Hom}(\ensuremath{#1})\xspace}
\newcommand{\lhomo}[1]{\textsc{LHom}(\ensuremath{#1})\xspace}

\newcommand{\coloring}[1]{\ensuremath{#1}-\textsc{Coloring}\xspace}

\newcommand{\tw}[1]{{\operatorname{tw}(#1)}}
\newcommand{\pw}[1]{{\operatorname{pw}(#1)}}
\newcommand{\fvs}[1]{{\operatorname{fvs}(#1)}}
\newcommand{\ctw}[1]{{\operatorname{ctw}(#1)}}

\newcommand{\mim}[1]{{\operatorname{mim}(#1)}}

\newcommand{\girth}[1]{{\operatorname{girth}(#1)}}
\newenvironment{claimproof}{\noindent {\emph{Proof of Claim.}}}{\hfill$\blacksquare$\smallskip}

\declaretheorem[sibling=theorem]{lemma}
\declaretheorem[sibling=theorem]{corollary}
\declaretheorem[sibling=theorem]{definition}

\usepackage[normalem]{ulem}
\newcommand\rsout{\bgroup\markoverwith{\textcolor{red}{\rule[0.5ex]{2pt}{0.4pt}}}\ULon}

\begin{document}
\title{Fine-grained complexity of the list homomorphism problem: feedback vertex set and cutwidth\thanks{This work is supported by Polish National Science Centre grant no. 2018/31/D/ST6/00062.}}

\author{Marta Piecyk\thanks{Warsaw University of Technology, Faculty of Mathematics and Information Science, \texttt{m.piecyk@mini.pw.edu.pl}} \and Paweł Rzążewski\thanks{Warsaw University of Technology, Faculty of Mathematics and Information Science and University of Warsaw, Institute of Informatics, \texttt{p.rzazewski@mini.pw.edu.pl}}
}

\begin{titlepage}
\def\thepage{}
\thispagestyle{empty}
\maketitle

\begin{abstract}
For graphs $G,H$, a homomorphism from $G$ to $H$ is an edge-preserving mapping from $V(G)$ to $V(H)$.
In the list homomorphism problem, denoted by \textsc{LHom}($H$), we are given a graph $G$,
whose every vertex $v$ is equipped with a list $L(v) \subseteq V(H)$, and we need to determine whether
there exists a homomorphism from $G$ to $H$ which additionally respects the lists $L$.
List homomorphisms are a natural generalization of (list) colorings.

Very recently Okrasa, Piecyk, and Rz\k{a}\.zewski [ESA 2020] studied the fine-grained
complexity of the problem, parameterized by the treewidth of the instance graph $G$.
They defined a new invariant $i^*(H)$, and proved that for every relevant graph $H$,
this invariant is the correct base of the exponent in the running time 
of any algorithm solving the  \textsc{LHom}($H$) problem.

In this paper we continue this direction and study the complexity of the problem under different
parameterizations. As the first result, we show that $i^*(H)$ is also the right complexity base 
if the parameter is the size of a minimum feedback vertex set of $G$, denoted by $\textrm{fvs}(G)$.
In particular, for every relevant graph $H$, the \textsc{LHom}($H$) problem
\begin{itemize}
\item can be solved in time $i^*(H)^{\textrm{fvs}(G)} \cdot |V(G)|^{\mathcal{O}(1)}$, if
a minimum feedback vertex set of $G$ is given,
\item cannot be solved in time $(i^*(H) - \varepsilon)^{\textrm{fvs}(G)} \cdot |V(G)|^{\mathcal{O}(1)}$,
for any $\varepsilon >0$, unless the SETH fails.
\end{itemize}

Then we turn our attention to a parameterization by the cutwidth $\textrm{ctw}(G)$ of $G$.
Jansen and Nederlof~[ESA 2018] showed that \textsc{List $k$-Coloring} (i.e., \textsc{LHom}($K_k$))  can be solved in time 
$\mathcal{O}^*\left (c^{\textrm{ctw}(G)}\right)$ for an absolute constant $c$, i.e., the base
of the exponential function does not depend on the number of colors.
Jansen asked whether this behavior extends to graph homomorphisms.
As the main result of the paper, we answer the question in the negative.
We define a new graph invariant $mim^*(H)$, closely related to the size of a maximum
induced matching in $H$, and prove that for all relevant graphs $H$, the \textsc{LHom}($H$)
problem cannot be solved in time $\mathcal{O}^*\left ((mim^*(H)-\varepsilon)^{\textrm{ctw}(G)}\right)$
for any $\varepsilon >0$, unless the SETH fails.
In particular, this implies that there is no constant $c$,
such that for every odd cycle the non-list version of the problem can be solved in time $\mathcal{O}^*\left (c^{\textrm{ctw}(G)} \right)$.

Finally, we generalize the algorithm of Jansen and Nederlof, so that it can be used to solve
\textsc{LHom}($H$) for every graph $H$; its complexity depends on $\textrm{ctw}(G)$ and  another invariant of $H$, which is constant for cliques.

\end{abstract}
\end{titlepage}

\section{Introduction}
The $k \coloring$ problem, which asks whether an input graph $G$ admits a proper coloring with $k$ colors,
is arguably one of the  best studied computational problems.
The problem is known to be notoriously hard: it is polynomial-time solvable (and, in fact, very simple) only for $k \leq 2$, and \textsf{NP}-complete otherwise, even in very restricted classes of graphs~\cite{GAREY1976237,DBLP:journals/siamcomp/Holyer81a,DBLP:journals/combinatorica/KhannaLS00,DBLP:journals/ejc/Huang16}.

When dealing with such a hard problem, an interesting direction of research is to study its \emph{fine-grained complexity} depending on some parameters of input instances, in order to understand where the boundary of easy and hard cases lies.
Such investigations usually follow two paths in parallel.
On one hand, we extend our algorithmic toolbox in order to solve the problem efficiently in various settings.
On the other hand, we try to show hardness of the problem, using appropriate reductions. This way we can show some lower bounds for the algorithms solving the problem.

In order to obtain meaningful lower bounds, the basic assumption of the classical complexity theory, i.e., $\textsf{P} \neq \textsf{NP}$, is not strong enough.
The usual assumptions used in this context are the Exponential Time Hypothesis (ETH) and the Strong Exponential Time Hypothesis (SETH), both formulated by Impagliazzo and Paturi~\cite{ImpagliazzoPaturi,DBLP:journals/jcss/ImpagliazzoPZ01}. Let us point out that the SETH is indeed stronger than the ETH, i.e, the former implies the latter one~\cite{DBLP:books/sp/CyganFKLMPPS15}.

\begin{conjecture}[ETH]
There exists $\delta > 0$, such that \textsc{3-Sat} with $n$ variables cannot be solved in time $2^{\delta \cdot n} \cdot n^{\Oh(1)}$.
\end{conjecture}

\begin{conjecture}[SETH]
\textsc{CNF-Sat} with $n$ variables and $m$ clauses cannot be solved in time $(2-\eps)^n \cdot (n+m)^{\Oh(1)}$ for any $\eps>0$.
\end{conjecture}

In case of $k \coloring$, the most natural parameter is the number of vertices.
While the brute-force approach to solve the problem on an $n$-vertex instance takes time $k^n \cdot n^{\Oh(1)}$, it is known that this can be improved, so that the base of the exponential function does not depend on $k$. The currently best algorithm is due to Bj\"orklund, Husfeldt, and Koivisto~\cite{DBLP:journals/siamcomp/BjorklundHK09} and has complexity $2^n \cdot n^{\Oh(1)}$. On the other hand, the standard hardness reduction shows that the problem cannot be solved in time $2^{o(n)}$, unless the ETH fails~\cite{DBLP:books/sp/CyganFKLMPPS15}.

Similarly, we can ask how the complexity depends on some parameters, describing the structure of the instance.
The most famous structural parameter is arguably the \emph{treewidth} of the graph, denoted by $\tw{G}$~\cite{DBLP:journals/dam/ArnborgP89,DBLP:journals/jal/RobertsonS86,Bodlaender:2008:COG:2479371.2479375}.
Intuitively, treewidth measures how tree-like the graph is. Thus, on graphs with bounded treewidth, we can mimick the bottom-up dynamic programming algorithms that works very well on trees.
In case of the $k \coloring$ problem, the complexity of such a straightforward approach is $k^{\tw{G}} \cdot n^{\Oh(1)}$, where $n$ is the number of vertices of of an instance graph $G$, provided that $G$ is given along with its tree decomposition of width $\tw{G}$. 
One might wonder whether this could be improved, in particular, if one can design an algorithm with running time $c^{\tw{G}} \cdot n^{\Oh(1)}$, where $c$ is a constant that does not depend on $k$, as it was possible in the case if the parameter is $n$.
Lokshtanov, Marx, and Saurabh~\cite{DBLP:conf/soda/LokshtanovMS11a} proved that this is unlikely, and an algorithm with running time $(k-\epsilon)^{\tw{G}} \cdot n^{\Oh(1)}$, for any $\epsilon >0$, would contradict the SETH. This lower bounds holds even if we replace treewidth with pathwidth $\pw{G}$; the latter result is stronger, as we always have $\tw{G} \leq \pw{G}$.

Another way to measure how close a graph $G$ is to a tree or forest is to analyze the size of a minimum \emph{feedback vertex set}, i.e., the minimum number of vertices that need to be removed from $G$ to break all cycles. This parameter is denoted by $\fvs{G}$.
If $G$ is given with a minimum feedback vertex set $S$, we can easily solve $k \coloring$ by enumerating all possible colorings of $S$, and trying to extend them on the forest $G-S$ using dynamic programming. The running time of such a procedure is $k^{\fvs{G}} \cdot n^{\Oh(1)}$.
This is complemented by a hardness result of Lokshtanov, Marx, and Saurabh~\cite{DBLP:conf/soda/LokshtanovMS11a}, who showed that the problem cannot be solved in time $(k-\epsilon)^{\fvs{G}} \cdot n^{\Oh(1)}$ for any $\epsilon >0$, unless the SETH fails.
Let us point that $\pw{G}$ and $\fvs{G}$ are incomparable parameters, so this result is incomparable with the previously mentioned lower bound.
These two lower bounds were later unified by Jaffke and Jansen~\cite{DBLP:conf/ciac/JaffkeJ17}, who considered the parameterization by the \emph{distance to a linear forest}.

The above examples show a behavior which is typical for many other parameters: the running time of the algorithm depends on the number $k$ of colors and this dependence is necessary under standard complexity assumptions~\cite{DBLP:conf/soda/GolovachL0Z18,DBLP:conf/icalp/Lampis18,DBLP:conf/ciac/JaffkeJ17}.
Thus, it was really surprising that Jansen and Nederlof~\cite{DBLP:journals/tcs/JansenN19} showed that for any $k$, the $k \coloring$ problem can be solved in time $c^{\ctw{G}} \cdot n^{\Oh(1)}$, where $c$ is an absolute constant and $\ctw{G}$ is the \emph{cutwidth} of $G$.
Intuitively, we can imagine $\ctw{G}$ as follows. We fix some permutation of the vertices of $G$ and place them on a horizontal line in this ordering. The edges of $G$ are drawn as arcs above the line; we do not care about intersections. Now, the width of this arrangement is the maximum number of edges that can be cut by a vertical line. The cutwidth is the minimum width over all linear arrangements of vertices of $G$.
The substantial difference between cutwidth and the previously mentioned parameters is that cutwidth corresponds to the number of edges, not the number of vertices, and, in particular, $\ctw{G}$ is not upper-bounded by $|V(G)|$.
Also, it is known that $\pw{G} \leq \ctw{G}$~\cite{PathwidthCutwidth}.
To be more specific, Jansen and Nederlof~\cite{DBLP:journals/tcs/JansenN19} presented two algorithms for $k \coloring$, parameterized by the cutwidth. The first one is deterministic and has running time $2^{\omega \cdot \ctw{G}} \cdot n^{\Oh(1)}$, where $\omega < 2.373$ is the matrix multiplication exponent, see Coppersmith, Winograd~\cite{DBLP:journals/jsc/CoppersmithW90} and Vassilevska-Williams~\cite{DBLP:conf/stoc/Williams12}. The second algorithm is randomized and works in time $2^{\ctw{G}} \cdot n^{\Oh(1)}$. Also, the authors show that the latter complexity is optimal under the SETH, even for $3 \coloring$.

Let us point out that all the algorithms mentioned above work also for the more general \textsc{List $k$-Coloring} problem, where each vertex $v$ of $G$ is equipped with a list $L(v) \subseteq \{1,2,\ldots,k\}$, and we additionally require that the assigned color comes from this list.
The general direction of our work is to investigate how further the techniques developed for $k \coloring$ can be generalized.

\paragraph{Graph homomorphisms.}
A rich family of graph problems that generalize $k \coloring$ comes from considering \emph{graph homomorphisms}.
A homomorphism from a graph $G$ to a graph $H$ (called \emph{target}) is an edge-preserving mapping from $V(G)$ to $V(H)$. In the \homo{H} problem we ask if the input graph $G$ admits a homomorphism to $H$, which is usually treated as a fixed graph.
Observe that if $H$ is $K_k$, i.e., a complete graph on $k$ vertices, then \homo{H} is equivalent to $k \coloring$.
The complexity classification of \homo{H} was provided by the seminal paper by Hell and Ne\v{s}et\v{r}il~\cite{DBLP:journals/jct/HellN90}: the problem is polynomial-time solvable if $H$ is bipartite or has a vertex with a loop, and \textsf{NP}-complete otherwise.
This problem can also be considered in a list setting, where every vertex $v$ of $G$ is equipped with a list $L(v) \subseteq V(H)$, and we ask for a homomorphism from $G$ to $H$, which additionally respects lists $L$. The corresponding computational problem is denoted by \lhomo{H}.

The complexity dichotomy for \lhomo{H} was proven in three steps: first, for reflexive graphs $H$ (i.e., where every vertex has a loops) by Feder and Hell~\cite{FEDER1998236}, then for irreflexive graphs $H$ (i.e., with no loops) by Feder, Hell, and Huang~\cite{DBLP:journals/combinatorica/FederHH99}, and finally, for all graphs $H$, again by  Feder, Hell, and Huang~\cite{DBLP:journals/jgt/FederHH03}.
The problem appears to be polynomial-time solvable if $H$ is a so-called \emph{bi-arc graph}. We will now skip the definition of this class and return to it in \cref{sec:preliminaries}. Let us also mention a special case if $H$ is irreflexive and bipartite: then the \lhomo{H} problem is in \textsf{P} if the complement of $H$ is a circular-arc graphs, and otherwise the problem is \textsf{NP}-complete. This special case will play a prominent role in our paper.

Let us point out that despite the obvious similarity of \homo{H} and \lhomo{H}, the methods used to prove lower bounds are very different.
In case of the \homo{H}, all hardness results use some algebraic tools, which allow us to capture the structure of the whole graph $H$ at once.
On the other hand, hardness proofs for \lhomo{H} are purely combinatorial and are based on the analysis of some small subgraphs of $H$.

The study of the complexity of \homo{H}, \lhomo{H}, and their variants led to many interesting algorithms and lower bounds~\cite{DBLP:journals/algorithmica/ChitnisEM17,DBLP:journals/mst/EgriKLT12,DBLP:conf/wg/ChenCD19,DBLP:conf/soda/HellR11,DBLP:conf/wg/FialaK06,DBLP:conf/lics/DalmauEHLR15,DBLP:conf/soda/EgriHLR14}. Let us mention few of them, that are most relevant to our results. For more information about the combinatorics and complexity of graph homomorphisms, we refer the reader to the comprehensive monograph by Hell and Ne\v{s}et\v{r}il~\cite{hell2004graphs}.

A brute-force approach to solving \homo{H} (and \lhomo{H}) has complexity $|V(H)|^n \cdot n^{\Oh(1)}$.
This can be improved if $H$ has some special structure: several algorithms with running time $f(H)^n \cdot n^{\Oh(1)}$ were obtained, where $f$ is a function of some structural parameter of $H$. Among possible choices of this parameter we can find the maximum degree (folklore), treewidth~\cite{DBLP:journals/mst/FominHK07}, clique-width~\cite{DBLP:journals/mst/Wahlstrom11}, or bandwidth of the complement~\cite{DBLP:journals/ipl/Rzazewski14}. A natural open question was whether one can obtain a $c^{n}$ algorithm, where $c$ is a constant that does not depend on $H$~\cite{DBLP:journals/mst/Wahlstrom11}.
This question was finally answered in the negative by Cygan \emph{et al.}~\cite{DBLP:journals/jacm/CyganFGKMPS17}, who proved that the brute force algorithm is essentially optimal under the ETH.

If we are interested in the complexity, parameterized by the treewidth of $G$, then both \homo{H} and \lhomo{H} can be solved in time $|V(H)|^{\tw{G}} \cdot n^{\Oh(1)}$ by a naive dynamic programming (again, provided that $G$ is given with a tree decomposition). 
The fine-grained complexity of the \homo{H} problem, parameterized by the treewidth of $G$, was studied recently by Okrasa and Rz\k{a}\.zewski~\cite{OkrasaSODA}. Using mostly algebraic tools, they were able obtain tight bounds, conditioned on two conjectures from algebraic graph theory from early 2000s.

The analogous question for the \lhomo{H} problem was first investigated by Egri, Marx, and Rz\k{a}\.zewski~\cite{DBLP:conf/stacs/EgriMR18} for reflexive graphs $H$, and then by Okrasa, Piecyk, and Rz\k{a}\.zewski~\cite{LhomoTreewidth,LhomoTreewidthFull} for the general case.
The authors defined a new graph invariant $i^*(H)$, and proved the following, tight bounds (recall that always $\tw{G} \leq \pw{G}$).

\begin{theorem}[Okrasa, Piecyk, Rzążewski~\cite{LhomoTreewidth,LhomoTreewidthFull}]\label{thm:lhomo-treewidth}
Let $H$ be a connected, non-bi-arc graph.
\begin{enumerate}[a)]
\item Even if $H$ is given in the input, every instance $(G,L)$ of  $\lhomo H$ can be solved in time $i^*(H)^{\tw{G}}\cdot (|V(G)| \cdot |V(H)|)^{\Oh(1)}$, provided that $G$ is given along with a tree decomposition of width $\tw{G}$.
\item Even if $H$ is fixed, there is no algorithm that solves every instance $(G,L)$ of $\lhomo H$ in time $(i^*(H)-\eps)^{\pw{G}}\cdot |V(G)|^{\Oh(1)}$ for any $\eps>0$, unless the SETH fails.
\end{enumerate}
\end{theorem}

To the best of our knowledge, the complexity depending on other structural parameters of $G$ was not investigated.
In this paper, we make some progress to fill this gap. In particular, our main motivation is the following question by Jansen~\cite{BartPersonal}, repeated by Okrasa, Piecyk, Rzążewski~\cite{LhomoTreewidth,LhomoTreewidthFull}.

\begin{question}[Jansen~\cite{BartPersonal}]\label{question}
Is there a universal constant $c$, such that for every $H$, every instance $G$ of the \homo{H} problem can be solved in time $c^{\ctw{G}} \cdot |V(G)|^{\Oh(1)}$?
\end{question}

\paragraph{Our results.}
As our first result, we complement the recent result of Okrasa~\emph{et al.}~\cite{LhomoTreewidth,LhomoTreewidthFull} and show tight complexity bounds, parameterized by the size of a minimum feedback vertex set of the instance graph.

\begin{theorem} \label{thm:main}
Let $H$ be a connected, non-bi-arc graph.
\begin{enumerate}[a)]
\item Even if $H$ is given in the input, every instance $(G,L)$ of  $\lhomo H$ can be solved in time $i^*(H)^{\fvs{G}}\cdot (|V(G)| \cdot |V(H)|)^{\Oh(1)}$, provided that $G$ is given along with a feedback vertex set of size $\fvs{G}$.
\item Even if $H$ is fixed, there is no algorithm that solves every instance $(G,L)$ of $\lhomo H$ in time $(i^*(H)-\eps)^{\fvs{G}}\cdot |V(G)|^{\Oh(1)}$ for any $\eps>0$, unless the SETH fails.
\end{enumerate}
\end{theorem}

Let us point out that the algorithmic part of the theorem, i.e., the statement a), follows directly from \cref{thm:lhomo-treewidth}~a), as given a graph $G$ and its feedback vertex set $S$, we can in polynomial time construct a tree decomposition  of $G$ with width $|S|+1$ (see also \cref{sec:preliminaries}). 
The proof of the lower bound follows the general direction of the hardness proof for $k \coloring$ by Lokshtanov \emph{et al.}~\cite{DBLP:conf/soda/LokshtanovMS11a}. However, as we are showing hardness for all relevant graphs $H$, the gadgets are significantly more complicated. In their construction we use some machinery developed by  Okrasa~\emph{et al.}~\cite{LhomoTreewidth,LhomoTreewidthFull}.
Furthermore, similarly to the proof of \cref{thm:lhomo-treewidth}~b), the proof of \cref{thm:main} is split into two parts: first we prove hardness for the special case if $H$ is bipartite, and then we reduce the general case to the bipartite one.

Then we turn our attention to the setting, where the parameter is the cutwidth of the instance graph.
Recall that $\ctw{G} \geq \pw{G} \geq \tw{G}$. Furthermore, given a linear layout of $G$ with width $w$, we can in polynomial time construct a tree decomposition of $G$ with width at most $w$~\cite{PathwidthCutwidth}. Thus by \cref{thm:lhomo-treewidth}~a) we know that \lhomo{H} can be solved in time $(i^*(H))^{\ctw{G}} \cdot |V(G)|^{\Oh(1)}$. On the other hand, we know that this algorithm cannot be optimal for all $H$, as $i^*(K_k)=k$, while \textsc{List $k$-Coloring}, i.e., $\lhomo{K_k}$, can be solved in deterministic time $2^{\omega \cdot \ctw{G}} \cdot |V(G)|^{\Oh(1)}$ or in randomized time $2^{\ctw{G}} \cdot |V(G)|^{\Oh(1)}$, using the algorithms of Jansen and Nederlof~\cite{DBLP:journals/tcs/JansenN19}.

We introduce another graph parameter, $mim^*(H)$, which is closely related to the size of a maximum induced matching in $H$,
and show two lower bounds, assuming, respectively, the SETH and the ETH.

\begin{restatable}{theorem}{mainctwlhomo}
\label{thm:main-ctw-list-hard}
Let $\mathcal{H}$ be the class of connected non-bi-arc graphs.
For $g \in \N$, let $\mathcal{C}_g$ be the class of subcubic bipartite graphs $G$ with girth at least $g$, such that vertices of degree $3$ in $G$ are at distance at least $g$.
\begin{enumerate}[a)]
\item For every $H \in \mathcal{H}$, there is no algorithm that solves every instance $(G,L)$ of  $\lhomo H$, where $G \in \mathcal{C}_g$, in time $(mim^*(H)-\epsilon)^{\ctw{G}} \cdot |V(G)|^{\Oh(1)}$ for any $\epsilon > 0$, unless the SETH fails,
\item There exists a constant $0<\delta<1$, such that for every $H \in \mathcal{H}$, there is no algorithm that solves every instance $(G,L)$ of $\lhomo H$, where $G \in \mathcal{C}_g$, in time $mim^*(H)^{\delta \cdot \ctw{G}} \cdot |V(G)|^{\Oh(1)}$, unless the ETH fails.
\end{enumerate}
\end{restatable}
As a sanity check, we point out that $mim^*(K_k)=2$, so our lower bounds are consistent with the results of Jansen and Nederlof~\cite{DBLP:journals/tcs/JansenN19}.

Let us highlight that the lower bounds from \cref{thm:main-ctw-list-hard} hold even for very restricted instances, and this statement captures some important graphs classes.
In particular, for a fixed graph $F$, we say that $G$ is $F$-free if it does not contain $F$ as an induced subgraph.
Recently Chudnovsky~\emph{et al.}~\cite{DBLP:conf/esa/ChudnovskyHRSZ19} studied the complexity of $\homo{C_k}$ and $\lhomo{C_k}$ for $F$-free graphs.
Among other results, they proved that if $F$ has a connected component that is not a path nor a subdivided claw, then for any $k \geq 5$, the $\lhomo{C_k}$ problem is \textsf{NP}-complete and cannot be solved in subexponential time in $F$-free graphs, unless the ETH fails.
One can immediately verify that the class $\mathcal{C}_g$ from \cref{thm:main-ctw-list-hard} for $g=|V(F)|+1$ is contained
 in the class of $F$-free graphs.
Thus we obtain  the following corollary from \cref{thm:main-ctw-list-hard}, which  significantly generalizes  the result of  Chudnovsky~\emph{et al.}~\cite{DBLP:conf/esa/ChudnovskyHRSZ19}, as cycles with at least 5 vertices are not bi-arc graphs~\cite{DBLP:journals/combinatorica/FederHH99},
\begin{corollary}
Let $F$ be a fixed graph, which has a connected component that is not a path nor a subdivided claw.
Let $\mathcal{H}$ be the class of connected non-bi-arc graphs.
\begin{enumerate}[a)]
\item For every $H \in \mathcal{H}$, there is no algorithm that solves every instance $(G,L)$ of  $\lhomo H$, where $G$ is $F$-free, in time $(mim^*(H)-\epsilon)^{\ctw{G}} \cdot |V(G)|^{\Oh(1)}$ for any $\epsilon > 0$, unless the SETH fails,
\item There exists a universal constant $0<\delta<1$, such that for every $H \in \mathcal{H}$, there is no algorithm that solves every instance $(G,L)$ of $\lhomo H$, where $G$ is $F$-free, in time $mim^*(H)^{\delta \cdot \ctw{G}} \cdot |V(G)|^{\Oh(1)}$, unless the ETH fails.
\end{enumerate}
\end{corollary}

Next, we focus on the non-list variant of the problem, i.e., \homo{H}.
Note that here we only consider graphs $H$ that are irreflexive and non-bipartite, as otherwise the problem is polynomial-time solvable.
Furthermore, we restrict our attention to graphs $H$ that are \emph{projective cores}.
The definition of projective cores is postponed to \cref{sec:nonlist}, but let us point out that many graphs fall into this class.
We say that a property $\Pi$ is satisfied by \emph{almost all graphs}, if the probability, that a graph chosen uniformly at random from the set of all graphs with $n$ vertices, satisfies $\Pi$, tends to 1 as $n$ grows.
It is well-known that almost all graphs are non-bipartite and connected~\cite{DBLP:books/daglib/0021015}.
Hell and Ne\v{s}et\v{r}il~\cite{hell2004graphs} proved that almost all graphs are cores, while Łuczak and Ne\v{s}et\v{r}il~\cite{luczak2004note} showed that almost all graphs are projective. All these results imply that almost all graphs are connected non-bipartite projective cores, see also~\cite{OkrasaSODA}.
For this class of graphs $H$, we show the following lower bounds, answering \cref{question} in the negative.

\begin{restatable}{theorem}{mainctwhomo}
\label{thm:main-ctw-nonlist-hard}
Let $\mathcal{H}$ be the class of connected non-bipartite projective cores with at least three vertices.
\begin{enumerate}[a)]
\item  For every $H \in \mathcal{H}$, there is no algorithm that solves every instance $G$ of  $\homo H$ in time $(mim^*(H)-\epsilon)^{\ctw{G}} \cdot |V(G)|^{\Oh(1)}$ for any $\epsilon > 0$, unless the SETH fails,

\item There exists a universal constant $0<\delta<1$, such that for every $H \in \mathcal{H}$, there is no algorithm that solves every instance $G$ of $\homo H$ in time $mim^*(H)^{\delta \cdot \ctw{G}} \cdot |V(G)|^{\Oh(1)}$, unless the ETH fails.
\end{enumerate}
\end{restatable}

In particular, it is well known that odd cycles are projective cores~\cite{DBLP:journals/dmgt/Larose02}.
Furthermore, for any odd cycle $C_k$ it holds that  $mim^*(C_k) = \lfloor 2k/3 \rfloor$.
Thus, we obtain the following as a corollary from \cref{thm:main-ctw-nonlist-hard}. Note that here we do not treat $C$ as a fixed graph.

\begin{restatable}{corollary}{corcyclesctw}
Let $\mathcal{C}$ be the family of odd cycles.
Then for $C \in \mathcal{C}$, there is no algorithm that solves $\homo{C}$ for instances $G$ in time $2^{o(\log |V(C)| \cdot \ctw{G})} \cdot |V(G)|^{\Oh(1)}$, unless the ETH fails. \newline
In particular, there is no universal constant $c$, such that $\homo{C}$ can be solved in time $c^{\ctw{G}} \cdot |V(G)|^{\Oh(1)}$ for every $C \in \mathcal{C}$ and every instance $G$.
\end{restatable}

Finally, in ~\cref{sec:algo} we have a closer look at the algorithm by Jansen and Nederlof~\cite{DBLP:journals/tcs/JansenN19}, and try to generalize it to the \lhomo{H} problem for non-complete targets $H$. We define yet another graph invariant $\gamma^*(H)$ and show the following result.

\begin{restatable}{corollary}{algolhomo}\label{thm:alglhomo}
Let $H$ be a graph with possible loops  and let $(G,L)$ be an instance of $\lhomo{H}$, where $G$ is given with a linear layout of width $k$.
Then $(G,L)$ can be solved in time $2^{\gamma^*(H) \cdot \omega \cdot k} \cdot (|V(G)| \cdot |V(H)|)^{\Oh(1)}$.
\end{restatable}

Let us point out that $\gamma^*(K_k)=1$, so our result  covers the result by Jansen and Nederlof~\cite{DBLP:journals/tcs/JansenN19}.
Actually, we define our algorithm for the so-called \emph{Binary Constraint Satisfaction Problem} (BCSP), which is a further generalization of \lhomo{H} (the definition of BCSP is quite technical, so we postpone it to~\cref{sec:algo}).
This approach allows us to obtain an algorithm for another natural problem, called \emph{DP-coloring}~\cite{DBLP:journals/ejc/BernshteynKZ17,DBLP:journals/jct/DvorakP18}, see~\cref{sec:DP} for the definition.
This problem, being a special case of BCSP, was introduced by Dvo\v{r}ak and Postle as a generalization of list coloring.

We conclude the paper with comparison of the discussed parameters, i.e., $i^*(H)$, $mim^*(H)$, and $\gamma^*(H)$, and pointing out some open questions and directions for future investigations.

\newpage
\section{Notation and preliminaries}\label{sec:preliminaries}
For a positive integer $n$, we define $[n]:=\{1,\ldots,n\}$.
For a set $X$, by $2^X$ we denote the set of all subsets of $X$. Unless explicitly stated otherwise,
all logarithms are of base $2$, i.e., $\log x := \log_2x$. 

Let $G$ be a graph.
For a set $S \subsetneq V(G)$, by $G-S$ we denote the graph induced by $V(G) \setminus S$. 
For a vertex $v \in V(G)$, by $N_G(v)$ we denote the set of neighbors of $v$.
If the graph $G$ is clear from the context, we write $N(v)$ instead of $N_G(v)$.
Note that $v \in N(v)$ if and only if $v$ is a vertex with a loop.
We say that two vertices $u,v \in V(G)$ are \emph{incomparable} if $N(u) \not\subseteq N(v)$ and $N(v) \not\subseteq N(u)$.
A set $S\subseteq V(G)$ is \emph{incomparable} if all its vertices are pairwise incomparable.
Equivalently, we can say that for every distinct $u,v \in S$, there is a vertex $u' \in N(u) \setminus N(v)$.
A set $S\subseteq V(G)$ is \emph{strongly incomparable} if for every $u \in S$ there exists $u' \in N(u)$,
such that $u'$ is non-adjacent to every vertex in $S \setminus \{u\}$. Such a vertex $u'$ is called a \emph{private neighbor} of $u$.
Clearly a strongly incomparable set is in particular incomparable.

The \emph{degree} of a vertex $v$, denoted by $\deg(v)$, is the number of vertices in $N(v)$.
For $u,v \in V(G)$, by $\dist(u,v)$ we denote the length (i.e., the number of edges) of a shortest path from $u$ to $v$.
By $\girth{G}$ we denote the length of a shortest cycle in $G$.

Let $G,H$ be graphs.
A mapping $\vphi: V(G) \to V(H)$ is a \emph{homomorphism} from $G$ to $H$ if for every edge $uv\in E(G)$ it holds that $\vphi(u)\vphi(v)\in E(H)$.
By \emph{$H$-lists} we mean an assignment $L: V(G) \to 2^{V(H)}$.
For a graph $G$ with $H$-lists $L$, \emph{list homomorphism} is a homomorphism $\vphi$ from $G$ to $H$, which additionally respects lists $L$, i.e., for every $v \in V(G)$ it holds $\vphi(v) \in L(v)$.
To denote that $\vphi$ is a list homomorphism from $G$ to $H$ we will write $\vphi: (G,L) \to H$.
By $(G,L) \to H$ we denote that some list homomorphism $\vphi: (G,L) \to H$ exists.
In the \homo{H} problem the instance is a graph $G$ and we ask whether $G \to H$.
In the \lhomo{H} problem the instance is a pair $(G,L)$, where $G$ is a graph and $L$ are $H$-lists, and ask whether $(G,L) \to H$.

For a set $S \subseteq V(G)$ we define $L(S):= \bigcup_{v\in S} L(v)$.
If it does not lead to confusion, for a set $V$ such that $V(G)\subseteq V$ and $H$-lists $L: V \to 2^{V(H)}$, we will denote the instance $(G,L|_{V(G)})$ by $(G,L)$, in order to simplify the notation.

Let $H$ be a graph. A \emph{walk} $\cP$ in $H$ is a sequence $p_1,\ldots,p_\ell$ of vertices of $H$ such that $p_ip_{i+1} \in E(H)$ for $i \in [\ell-1]$. We define the \emph{length} of a walk $\cP=p_1,\ldots,p_\ell$ as $\ell-1$ and denote it by $|\cP|$.
We also write $\cP : p_1 \to p_{\ell}$ to emphasize that $\cP$ starts in $p_1$ and ends in $p_\ell$.
Let us define a relation that is crucial for building our gadgets.

\begin{definition}[Avoiding]
For walks $\cP = p_1,\ldots, p_\ell$ and $\cQ= q_1,\ldots, q_\ell$ of equal length, such that $p_1$ is in the same bipartition class as $q_1$, we say $\cP$ \emph{avoids $\cQ$} if $p_1 \neq q_1$ and for every $i \in [\ell-1]$ it holds that $p_iq_{i+1} \not\in E(H)$. 
\end{definition}

By $\overline{\cP}$ we denote walk $\cP$ reversed, i.e., if $\cP=p_1,\ldots,p_\ell$, then $\overline{\cP}=p_\ell,\ldots,p_1$. It is straightforward to observe that if $\cP$ avoids $\cQ$, then $\overline{\cQ}$ avoids $\overline{\cP}$.

\subsection{Graph parameters}

\paragraph{Treewidth and pathwidth.}
Let $G$ be a graph. A \emph{tree decomposition} of $G$ is a pair $(\cT,(X_t)_{t \in V(\cT)})$ such that:
\begin{enumerate}
\item $\cT$ is a tree,
\item $(X_t)_{t \in V(\cT)}$ is a family of subsets of $V(G)$,
\item every $v \in V(G)$ is contained in at least one $X_t$,
\item for every $uv \in E(G)$ there is $X_t$ such that $\{u,v\} \subseteq X_t$,
\item for every $v \in V(G)$ the graph induced by $\{ t \in V(\cT) \ | \ v \in X_t \}$ in $\cT$ is connected.
\end{enumerate}

We call sets $X_t$ for $t \in V(\cT)$ \emph{bags}.
The \emph{width} of a tree decomposition $(\cT,(X_t)_{t\in V(\cT)})$ is $\max_{t \in V(\cT)} |X_t| - 1$.
The minimum width over all tree decompositions of $G$ is called the \emph{treewidth} of $G$ and we denote it by $\tw{G}$.
Similarly, we define a \emph{path decomposition} $(\cP,(X_t)_{t \in V(\cT)})$ of $G$ as a tree decomposition, in which the tree $\cP$ is a path.
The minimum width over all path decompositions of $G$ is called the \emph{pathwidth} of $G$ and we denote it by $\pw{G}$.
Clearly $\tw{G} \leq \pw{G}$.

\paragraph{Feedback vertex set.}
A set $F \subseteq V(G)$, such that $G-F$ does not contain any cycle, is called a \emph{feedback vertex set} of $G$.
We denote the size of a minimum feedback vertex set in $G$ by $\fvs{G}$.

Observe that if $F$ is a feedback vertex set of $G$, we can obtain a tree decomposition of $G$ by taking the tree decomposition of $G - F$ (of width 1), and adding $F$ to every bag. This gives the following.

\begin{proposition}\label{prop:fvs-tw}
Let $G$ be a graph given along with its feedback vertex set $F$ of size $s$.
Then there exists a tree decomposition $(\cT,(X_t)_{t\in V(\cT)})$ of $G$ of width $s+1$.
Moreover, $(\cT,(X_t)_{t\in V(\cT)})$ can be constructed in polynomial time.
\end{proposition}

In particular this implies that for every graph $G$ we have $\tw{G} \leq \fvs{G}+1$.
Let us point out that the parameters $\pw{G}$ and $\fvs{G}$ are incomparable. Indeed, if $G$ is a complete binary tree on $n$ vertices, then $\fvs{G}=0$ and $\pw{G} = \Theta(\log n)$. On the other hand, if $G$ is a collection of $n/3$ disjoint triangles, then $\fvs{G}=n/3$ and $\pw{G}=2$.

\paragraph{Cutwidth.}

Let $\pi=(v_1,\ldots,v_n)$ be a linear ordering of vertices of $G$, we will call it a \emph{linear layout} of $G$ or a \emph{linear arrangement} of $G$.
A \emph{cut} of $\pi$ is a partition of $V(G)$ into two subsets: $\{v_1,\ldots,v_p\}$ and $\{v_{p+1},\ldots,v_n\}$, for some $p \in [n]$.
We say that an edge $v_iv_j$, where $i < j$, \emph{crosses} the cut  $(\{v_1,\ldots,v_p\},\{v_{p+1},\ldots,v_n\})$, if $i \leq p$ and $j > p$.
The \emph{width} of the linear layout $\pi$ is the maximum number of edges that cross any cut of $\pi$. 
Finally, we define the \emph{cutwidth} $\ctw{G}$ of $G$ as the minimum width over all linear layouts of $G$.

It is known that $\pw{G} \leq \ctw{G}$.
Furthermore, given a linear layout of $G$ with width $k$, we can in polynomial time construct a path decomposition of $G$ with width at most $k$~\cite{PathwidthCutwidth}.
On the other hand, for every graph $G$ it holds that $\ctw{G} \leq \pw{G} \cdot \Delta(G)$~\cite{DBLP:journals/dm/ChungS89}.
As we also have that $\ctw{G} \geq \Delta(G)/2$, we can intuitively think that $\ctw{G}$ is bounded if and only if both $\Delta(G)$ and $\pw{G}$ are bounded.

\subsection{Hard instances of $\lhomo{H}$}

In this section we briefly discuss the structure of graphs $H$, for which the \lhomo{H} problem is NP-complete.
Those graphs will be our focus in this paper.

\paragraph{Bipartite graphs $H$.}
First let us discuss the case if $H$ is bipartite.
Recall that Feder, Hell, and Huang~\cite{DBLP:journals/combinatorica/FederHH99} proved that in this case the $\lhomo{H}$ problem is polynomial-time solvable if $H$ is a complement of a circular-arc graph and NP-complete otherwise.
Moreover, they provided a characterization of this class of graphs in terms of forbidden subgraphs:
a bipartite graph $H$ is not the complement of a circular-arc graph if and only if $H$ contains an induced cycle of length at least 6 or a  structure called a \emph{special edge asteroid}.
We omit the definition of a special edge asteroid, as it is quite technical and not relevant to our paper.
Instead, we will rely on the following structural result, obtained by Okrasa \emph{et al.}~\cite{LhomoTreewidth,LhomoTreewidthFull}.

\begin{lemma}[Okrasa \emph{et al.}~\cite{LhomoTreewidth,LhomoTreewidthFull}]\label{obs:walks-between-corners}
Let $H$ be a bipartite graph, whose complement is not a circular-arc graph. Let $X$ be one of bipartition classes in $H$. Then there exists a triple $(\alpha,\beta,\gamma)$ of vertices in $X$ such that: 
\begin{enumerate}
\item there exist $\alpha',\beta' \in V(H)$, such that the edges $\alpha\alpha', \beta\beta'$ induce a matching in $H$,
\item vertices $\alpha,\beta,\gamma$ are pairwise incomparable,
\item there exist walks $\cX,\cX':\alpha \to \beta$ and $\cY,\cY':\beta \to \alpha$, such that $\cX$ avoids $\cY$ and $\cY'$ avoids $\cX'$,
\item at least one of the following holds:
\begin{enumerate}[a)]
\item $H$ contains an induced $C_6$ with consecutive vertices $w_1,\ldots,w_6$ and $\alpha=w_1,\beta=w_5,\gamma=w_3$,
\item $H$ contains an induced $C_8$ with consecutive vertices $w_1,\ldots,w_8$ and $\alpha=w_1,\beta=w_5,\gamma=w_3$,
\item the set $\{\alpha,\beta,\gamma\}$ is strongly incomparable and for any $a,b,c$, such that $\{a,b,c\} = \{\alpha,\beta,\gamma\}$, there exist walks $\cX_c:\alpha \to a \textrm{ and } \cY_c: \alpha \to b,$ and $\cZ_c: \beta \to c$, such that $\cX_c,\cY_c$ avoid $\cZ_c$ and $\cZ_c$ avoids $\cX_c,\cY_c$.
\end{enumerate}
\end{enumerate}
\end{lemma}

Furthermore, Okrasa \emph{et al.}~\cite{LhomoTreewidth,LhomoTreewidthFull} observed that in order to understand the complexity of \lhomo{H}, it
is sufficient to focus on the so-called \emph{consistent instances}.

\begin{definition}[Consistent instance]\label{def:consist-inst}
Let $H$ be a bipartite graph with bipartition classes $X,Y$ and let $(G,L)$ be an instance of $\lhomo H$. We say that $(G,L)$ is \emph{consistent} if the following conditions hold:
\begin{enumerate}
\item $G$ is connected and bipartite with bipartition classes $X_G,Y_G$,
\item $L(X_G)\subseteq X, L(Y_G) \subseteq Y$,
\item for every $v\in V(G)$, the set $L(v)$ is incomparable.
\end{enumerate}
\end{definition}

Indeed, if $G$ is not connected, then we can solve the problem for each connected component of $G$ independently.
If $G$ is not bipartite, then we can immediately report a no-instance.
If some list contains two vertices $x,y$, such that $N(x) \subseteq N(y)$, then we can safely remove $x$ from the list.
Finally, note that in every homomorphism from $G$ to $H$, either $X_G$ is mapped to vertices of $X$, and $Y_G$ is mapped to vertices of $Y$,
or $X_G$ is mapped to vertices of $Y$, and $Y_G$ is mapped to vertices of $X$.
We can consider these two cases separately, reducing the problem to solving two consistent instances, without changing the asymptotic complexity of the algorithm.

\paragraph{General graphs $H$.}
Recall that for general graphs $H$, Feder, Hell, and Huang~\cite{DBLP:journals/jgt/FederHH03} showed that the \lhomo{H} problem is polynomial-time solvable if $H$ is a bi-arc graph, and NP-complete otherwise. They defined the class of bi-arc graphs in terms of some geometric representation, but for us it will be more convenient to show an equivalent definition.

For a graph $H$, the \emph{associated bipartite graph} $H^*$ is the graph with vertex set $V(H^*)=\{ v',v'' \ | \ v  \in V(H)\}$, whose edge set contains  those pairs $u'v''$, for which $uv \in E(H)$.
Note that if $H$ is bipartite, then $H^*$ consists of two disjoint copies of $H$.

Feder, Hell, and Huang~\cite{DBLP:journals/jgt/FederHH03} observed that $H$ is a bi-arc graph if and only if $H^*$ is the complement of a circular-arc graph. Furthermore, an irreflexive graph is bi-arc if and only if it is bipartite and its complement is a circular-arc graph.
Thus ``hard'' cases of $\lhomo H$ correspond to the ``hard'' cases of $\lhomo {H^*}$. 
This observation yields the following, useful proposition.

\begin{proposition}[\cite{LhomoTreewidth,LhomoTreewidthFull}] \label{prop:bipartite-associted} 
Let $H$ be a graph and let $(G,L)$ be a consistent instance of \lhomo{H^*}.
Define $L' \colon V(G) \to 2^{V(H)}$ as $L'(x) := \{u \colon \{u',u''\} \cap L(x) \neq \emptyset\}$.
Then $(G,L) \to H^*$ if and only if $(G,L') \to H$.
\end{proposition}

\subsection{Incomparable sets, decompositions, and main invariants}

In this section we introduce the main invariants, $i^*(H)$ and $mim^*(H)$. First, let us define parameters $i(H)$ and $\mim{H}$.

\begin{definition}[$i(H)$ and $mim(H)$]
Let $H$ be a bipartite graph. 
By $i(H)$ (resp. $mim(H)$) we denote the maximum size of an incomparable set (resp. strongly incomparable set) in $H$,
which is fully contained in one bipartition class.
\end{definition}

Let $S$ be a strongly incomparable set, contained in one bipartition class, and let $S'$ be the set of private neighbors of vertices of $S$.
We observe that the set $S \cup S'$ induces a matching in $H$ of size $|S|$.
On the other hand, if $M$ is an induced matching, then the endpoints of edges from $M$ contained in one bipartition class form a strongly incomparable set of size $|M|$. Thus $mim(H)$ can be equivalently defined as the size of a maximum induced matching in $H$.

Before we define $i^*(H)$ and $mim^*(H)$, we need one more definition.

\begin{figure}
\centering{\begin{tikzpicture}[every node/.style={draw,circle,fill=white,inner sep=0pt,minimum size=30pt},every loop/.style={}]
\node (dx) at (0,0) {};
\node (dy) at (2,0) {};
\node (nx) at (0,-1.5) {};
\node (ny) at (2,-1.5) {};
\node (rx) at (0,-3) {};
\node (ry) at (2,-3) {};
\draw[very thick] (dx)--(ny)--(nx)--(dy);
\draw[color=orange] (nx)--(ry)--(rx)--(ny);
\draw[color=orange] (dx)--(dy);
\draw[dashed] (-1,0.7)--++(4,0)--++(0,-1.4)--++(-4,0)--++(0,1.4);
\draw[dashed] (-1,-0.8)--++(4,0)--++(0,-1.4)--++(-4,0)--++(0,1.4);
\draw[dashed] (-1,-2.3)--++(4,0)--++(0,-1.4)--++(-4,0)--++(0,1.4);
\node[draw=none, fill=none, label=left:{$D$}] (d) at (-1,0) {};
\node[draw=none, fill=none, label=left:{$N$}] (n) at (-1,-1.5) {};
\node[draw=none, fill=none, label=left:{$R$}] (r) at (-1,-3) {};
\end{tikzpicture}}
\caption{Bipartite decomposition $(D,N,R)$. Circles denote indepenent sets. A black line denotes that there are all possible edges between sets, an orange one that there might be some edges, and the lack of a line denotes that there are no edges between sets.}\label{fig:decomp}
\end{figure}
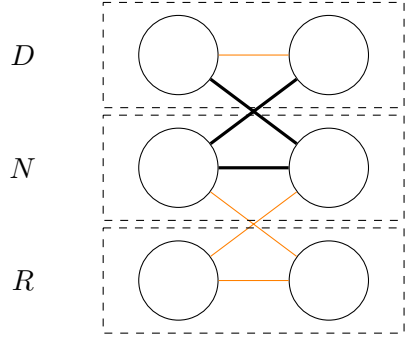

\begin{definition}[Bipartite decomposition]\label{def:bipartite-decomposition}
Let $H$ be a bipartite graph with bipartition classes $X,Y$. A partition of $V(H)$ into an ordered triple of sets $(D,N,R)$ is a \emph{bipartite decomposition} if the following conditions are satisfied (see \cref{fig:decomp})
\begin{enumerate}
\item $N$ is non-empty and separates $D$ and $R$, \label{it:bipdecomp-separator}
\item $|D\cap X| \geq 2$ or $|D \cap Y| \geq 2$, \label{it:bipdecomp-geq2}
\item $N$ induces a biclique in $H$, \label{it:bipdecomp-biclique}
\item $(D \cap X) \cup (N \cap Y)$ and $(D \cap Y) \cup (N \cap X)$ induce bicliques in $H$. \label{it:bipdecomp-complete}
\end{enumerate}
If $H$ does not admit a bipartite decomposition, then $H$ is \emph{undecomposable}.
\end{definition}

Let us point out that if $H$ is not the complement of a circular-arc graph, then it contains an induced subgraph, which is not the complement of a circular-arc graph and is undecomposable, see e.g.~\cite[Theorem~46]{LhomoTreewidthFull}. 
Now we are ready to define the following.

\begin{definition}[$i^*(H)$ and $mim^*(H)$ for bipartite $H$]\label{def:i_star}
Let $H$ be a connected bipartite graph, whose complement is not a circular-arc graph. Define
\begin{align*}
i^*(H) := \max \{ & i(H') \colon H' \text{ is an undecomposable, connected, induced}\\
& \text{ subgraph of }H, \text{ whose complement is not a circular-arc graph}\},\\
mim^*(H) := \max \{ & mim(H') \colon H' \text{ is an undecomposable, connected, induced}\\
& \text{ subgraph of }H, \text{ whose complement is not a circular-arc graph}\}.
\end{align*}
\end{definition}

It remains to extend the definitions of $i^*(H)$ and $mim^*(H)$ to general graphs $H$.
Recall that if $H$ is bipartite, then $H^*$ consists of two disjoint copies of $H$. Thus for bipartite $H$ it holds that $i^*(H^*)=i^*(H)$ and $mim^*(H^*)=mim^*(H)$.
In the other case, if $H$ is non-bipartite and additionally connected, then $H^*$ is connected.
This motivates the following extension of the definition of $i^*$ and $mim^*$ to non-bipartite $H$.
\begin{definition}[$i^*(H)$ and $mim^*(H)$] Let $H$ be a non-bi-arc graph. Define:
\begin{align*}
i^*(H):= & i^*(H^*),\\
mim^*(H):= & mim^*(H^*).
\end{align*}
\end{definition}

\newpage
\section{Parameter: the size of a minimum feedback vertex set}
\subsection{Lower bound for bipartite target graphs}
First, we will prove the bipartite version of \cref{thm:main}~b), i.e., the case that $H$ is bipartite.

\begin{theorem}\label{thm:bipartite}
Let $H$ be a connected, bipartite graph, whose complement is not a circular-arc graph.
Even if $H$ is fixed, there is no algorithm that solves every instance $(G,L)$ of  $\lhomo{H}$ in time $(i^*(H)-\epsilon)^{\fvs{G}} \cdot |V(G)|^{\Oh(1)}$ for any $\epsilon > 0$, unless the SETH fails.
\end{theorem}

In order to prove \cref{thm:bipartite} it is sufficient to prove the following.

\begin{theorem}\label{thm:bipartite-undecomp}
Let $H$ be a connected, bipartite, undecomposable graph, whose complement is not a circular-arc graph.
Even if $H$ is fixed, there is no algorithm that solves every instance $(G,L)$ of  $\lhomo{H}$ in time $(i(H)-\epsilon)^{\fvs{G}} \cdot |V(G)|^{\Oh(1)}$ for any $\epsilon > 0$, unless the SETH fails.
\end{theorem}

Let us show that \cref{thm:bipartite} and \cref{thm:bipartite-undecomp} are equivalent. 

%
%
%

\medskip
\noindent\textbf{(\cref{thm:bipartite} $\rightarrow$  \cref{thm:bipartite-undecomp})} Assume the SETH and suppose that \cref{thm:bipartite} holds and \cref{thm:bipartite-undecomp} fails.
Then there exists a connected, bipartite, udecomposable graph $H$, whose complement is not a circular-arc graph and an algorithm that solves $\lhomo H$ for every instance $(G,L)$ in time $(i(H)-\epsilon)^{\fvs{G}} \cdot |V(G)|^{\Oh(1)}$ for some $\epsilon > 0$.
The properties of $H$ imply that $i^*(H)=i(H)$.
Thus $\lhomo H$ can be solved for every instance $(G,L)$ in time $(i^*(H)-\epsilon)^{\fvs{G}} \cdot |V(G)|^{\Oh(1)}$.
Observe that $H$ satisfies the assumptions of \cref{thm:bipartite}. Therefore, by \cref{thm:bipartite}, we get a contradiction with the SETH.

\medskip
\noindent\textbf{(\cref{thm:bipartite-undecomp} $\rightarrow$ \cref{thm:bipartite})} Assume the SETH and suppose that \cref{thm:bipartite-undecomp} holds and \cref{thm:bipartite} fails.
Then there exist a connected, bipartite graph $H$, whose complement is not a circular-arc graph, and an algorithm that solves $\lhomo H$ for every instance $(G,L)$ in time $(i^*(H)-\epsilon)^{\fvs{G}} \cdot |V(G)|^{\Oh(1)}$ for some $\epsilon > 0$. Let $H'$ be an induced subgraph of $H$ such that $H'$ is connected, undecomposable, is not a complement of a circular-arc graph, and $i(H')=i^*(H)$. Observe that any instance $(G,L)$ of $\lhomo {H'}$ can be seen as an instance of $\lhomo H$ such that only vertices of $H'$ appear on lists $L$. Thus we can solve any instance $(G,L)$ of $\lhomo {H'}$ in time $(i^*(H)-\epsilon)^{\fvs{G}} \cdot |V(G)|^{\Oh(1)}=(i(H')-\epsilon)^{\fvs{G}} \cdot |V(G)|^{\Oh(1)}$, which by \cref{thm:bipartite-undecomp} contradicts the SETH.

\paragraph{Gadgets needed for hardness reduction.} From now on we assume that $H$ is a bipartite graph, whose complement is not a circular-arc graph and $(\alpha,\beta,\gamma)$ is the triple given by \cref{obs:walks-between-corners}. In order to prove \cref{thm:bipartite-undecomp} we will need two gadgets. The first one is a graph called an \emph{assignment gadget} and has two special vertices. Its main goal is to ensure that a certain coloring of one special vertex forces a certain coloring of the other special vertex.

\begin{restatable}[Assignment gadget]{definition}{defAssign}\label{def:assign-gadget}
Let $S$ be an  incomparable set in $H$ contained in the same bipartition class as $\alpha,\beta,\gamma$ and let $v \in S$. An \emph{assignment gadget} is a graph $A_v$ with $H$-lists $L$ and with special vertices $x,y$, such that:
\begin{enumerate}[({A}1.)]
\item $L(x)=S$ and $L(y)=\{\alpha,\beta,\gamma\}$,

\item for every $u \in S$ and for every $a \in \{\alpha,\beta\}$ there exists a list homomorphism $\vphi: (A_v,L) \to H$ such that $\vphi(x)=u$ and $\vphi(y)=a$,

\item there exists a list homomorphism $\vphi: (A_v,L) \to H$ such that $\vphi(x)=v$ and $\vphi(y)=\gamma$,

\item for every list homomorphism $\vphi: (A_v,L) \to H$ it holds that if $\vphi(y)=\gamma$, then $\vphi(x)=v$,

\item $A_v - \{x\}$ is a tree,

\item $\deg(x)=(|S|-1)^2$ and $\deg(y)=|S|-1$,

\item the degree of every vertex of $A_v$, possibly except $x$ and $y$, is at most $3$.
\end{enumerate}
\end{restatable}

The second gadget is called a \emph{switching gadget}. It is a path $T$ with a special internal vertex $q$, whose list is $\{\alpha,\beta,\gamma\}$, and endvertices with the same list $\{\alpha,\beta\}$.
Coloring both endvertices of $T$ with the same color, i.e., coloring both with $\alpha$ or both with $\beta$, allows us to color $q$ with one of $\alpha,\beta$, but ``switching sides'' from $\alpha$ to $\beta$ forces coloring $q$ with $\gamma$.

\begin{restatable}[Switching gadget]{definition}{defSwitch}\label{def:switch-gadget}
A \emph{switching gadget} is a path $T$ of even length with $H$-lists $L$, endvertices $p,r$, called respectively the \emph{input} and the \emph{output} vertex, and one special internal vertex $q$, called a \emph{$q$-vertex}, in the same bipartition class as $p,r$, such that:
\begin{enumerate}[(S1.)]
\item $L(p)=L(r)=\{\alpha,\beta\}$ and $L(q)=\{\alpha,\beta,\gamma\}$,
\item for every $a \in \{\alpha,\beta\}$ there exists a list homomorphism $\vphi: (T,L) \to H$, such that $\vphi(p)=\vphi(r)=a$ and $\vphi(q) \neq \gamma$,
\item there exists a list homomorphism $\vphi: (T,L) \to H$, such that $\vphi(p)=\alpha$, $\vphi(r)=\beta$, and $\vphi(q)=\gamma$,
\item for every list homomorphism $\vphi: (T,L) \to H$, if $\vphi(p)=\alpha$ and $\vphi(r)=\beta$, then $\vphi(q)=\gamma$.
\end{enumerate}
\end{restatable}
Note that in a switching gadget we do not care about homomorphisms that map $p$ to $\beta$ and $r$ to $\alpha$.

Later, when discussing assignment and switching gadgets, we will use the notions of $x$-, $y$-, $p$-, $q$-, and $r$-vertices to refer to the appropriate vertices introduced in the definitions of the gadgets.

The following two lemmas show that both, the assignment gadget and the switching gadget, even with some additional restrictions, can be constructed. Since their proofs are quite technical, they are postponed to \cref{sec:constr-of-gadgets}.

\begin{restatable}[Construction of the assignment gadget]{lemma}{lemAssignGadget}\label{lem:assign-gadget}
Let $H$ be an undecomposable, connected, bipartite graph, whose complement is not a circular-arc graph. Let $(\alpha,\beta,\gamma)$ be the triple from \cref{obs:walks-between-corners}. Let $S$ be an  incomparable set in $H$ contained in the same bipartition class as $\alpha,\beta,\gamma$, such that $k:=|S| \geq 2$. Let $g \in \N$. Then for every $v \in S$ there exists an assignment gadget $A_v$ such that $\girth{A_v} \geq g$ and for any distinct vertices $a,b$ in $A_v$ of degree at least $3$ it holds that $\dist(a,b) \geq g$, $\dist(a,x)\geq g$, $\dist(a,y)\geq g$, and $\dist(x,y)\geq g$.
\end{restatable}

\begin{restatable}[Construction of the switching gadget]{lemma}{lemSwitchGadget}\label{lem:switching-gadget}
Let $H$ be an undecomposable, connected, bipartite graph, whose complement is not a circular-arc graph. Let $(\alpha,\beta,\gamma)$ be the triple from \cref{obs:walks-between-corners} and let $g \in \N$. Then there exists a switching gadget $T$ with special vertices $p,q,r$ such that $\dist(p,q)  \geq \frac{g}{2}$ and $\dist(r,q)  \geq \frac{g}{2}$.
\end{restatable}
\paragraph{Reduction.}Suppose that we can construct both, the assignment gadget and the switching gadget. Let us show that this is sufficent to prove \cref{thm:bipartite-undecomp}. The proof is an extension of the construction of Lokshtanov, Marx, and Saurabh for the special case if $H$ is a complete graph, i.e., the $\textsc{List}~k \coloring$ problem~\cite{DBLP:conf/soda/LokshtanovMS11a}.

\begin{proof}[Proof of \cref{thm:bipartite-undecomp}]
Let $\phi$ be an instance of \textsc{CNF-Sat} with $n$ variables and $m$ clauses. Let $\eps>0$ and $k=i(H)$. Let $S$ be a maximum incomparable set contained in one bipartition class of $H$, i.e., $|S|=k$. Let $\alpha, \beta, \gamma$ be the vertices of $H$, in the same bipartition class as $S$, given by \cref{obs:walks-between-corners}. Let $\alpha',\beta'$ be the vertices such that edges $\alpha\alpha',\beta\beta'$ induce a matching in $H$, they exist by \cref{obs:walks-between-corners}.
Observe that $k \ge 3$, since vertices $\alpha,\beta,\gamma$ are pairwise incomparable. Moreover, we define $\lambda:=\log_k(k-\eps)$. Observe that $\lambda<1$. We choose an integer $p$ sufficiently large so that $\lambda \frac{p}{p-1} <1$ and define $t:=\Big{\lceil} \frac{n}{\lfloor \log k^p \rfloor} \Big{\rceil}=\Big{\lceil} \frac{n}{\lfloor p \cdot \log k \rfloor} \Big{\rceil}$.

We will construct a graph $G$ with $H$-lists $L$ such that:
\begin{enumerate}
\item there exists a list homomorphism $\vphi: (G,L) \to H$ if and only if $\phi$ is satisfiable,
\item the size of a minimum feedback vertex set in $G$ is at most $t \cdot p$,
\item $|V(G)|=(n+m)^{\Oh(1)}$.
\end{enumerate}

We partition the variables of $\phi$ into $t$ sets $F_1,\ldots,F_t$ called \emph{groups}, such that $|F_i| \leq \lfloor \log k^p \rfloor$. For each $i \in [t]$ we introduce $p$ vertices $x_1^i,\ldots,x_p^i$ and for every $s \in [p]$ we set $L(x_s^i):=S$. Each coloring of these vertices will be interpreted as a truth assignment of variables in $F_i$. Note that there are at most $2^{\lfloor \log k^p \rfloor} \leq k^p$ possible truth assigments of variables in $F_i$ and there are $k^p$ possible colorings of $x_1^i,\ldots,x_p^i$, respecting lists $L$.
Thus we can define an injective mapping that assigns a distinct coloring of vertices $x_1^i,\ldots,x_p^i$ to each truth assignment of the variables in $F_i$, note that some colorings may remain unassigned.

For every clause $C$ of $\phi$ we introduce a path $P_C$ constructed as follows.
Consider a group $F_i$  that contains at least one variable from $C$, and a truth assignment of $F_i$ that satisfies $C$.
Recall that this assignment corresponds to a coloring $f$ of vertices $x_1^i,\ldots,x_p^i$.
We introduce a switching gadget $T_C^{i,f}$, whose $q$-vertex is denoted by $q_C^{i,f}$.
We fix an arbitrary ordering of all switching gadgets introduced for the clause $C$. For every switching gadget but the last one, we identify its output vertex with the input vertex of the succesor. We add vertices $x_C$ with $L(x_C)=\{\alpha'\}$ and $y_C$ with $L(y_C)=\{\beta'\}$. We add an edge between $x_C$ and the input of the first switching gadget, and between $y_C$ and the output of the last switching gadget. This completes the construction of $P_C$.

Now consider a switching gadget $T_C^{i,f}$ introduced in the previous step. Recall that $C$ is a clause of $\phi$,
and $f$ is a coloring of $x_1^i,\ldots,x_p^i$ corresponding to a truth assignment of variables in $F_i$, which satisfies $C$.
Let us define $v_s:=f(x_s^i)$ for $s \in [p]$.
For every $s \in [p]$, we call \cref{lem:assign-gadget} to construct the assignment gadget $A_{v_s}$. Here we do not care about the girth of this gadget, so $g$ can be chosen arbitrarily.
We identify the $x$-vertex of $A_{v_s}$ with $x_s^i$ and the $y$-vertex with $q_C^{i,f}$.
This completes the construction of $(G,L)$ (see \cref{fig:reduction}), note that the construction is performed in time $(n+m)^{\Oh(1)}$.

\begin{figure}
\centering{\begin{tikzpicture}[every node/.style={draw,circle,fill=white,inner sep=0pt,minimum size=8pt},every loop/.style={}]
\node[label=left:\footnotesize{$x_C$}] (xc) at (0,0) {};
\foreach \k in {1,3,5,8,10,12}
{
\node (a\k) at (\k,0) {};
}
\node[label=right:\footnotesize{$y_C$}] (yc) at (13,0) {};
\draw (xc)--(a1);
\draw (yc)--(a12);

\foreach \k in {1,3,8,10}
{
\draw (a\k)--++(0.3,0)--++(0,0.3)--++(1.4,0)--++(0,-0.6)--++(-1.4,0)--++(0,0.3);
}
\foreach \i in {3,5,10,12}
{
\draw (a\i)--++(-0.3,0);
}
\foreach \j in {2,4,9,11}
{
\node (q\j) at (\j,0.5) {};
\draw (q\j)--++(0,-0.2);
}
\foreach \k in {0,1,2}
{
\node (b\k) at (0.5+5*\k,3.5) {};
}
\foreach \k in {0,1,2}
{
\node (c\k) at (2.5+5*\k,3.5) {};
}
\foreach \k in {0,1,2}
{
\node (d\k) at (1+5*\k,3.5) {};
\node[draw=none,fill=none,label=above:\footnotesize{$x_1^1$}] (e) at (0.5,3.5) {};
\node[draw=none,fill=none,label=above:\footnotesize{$x_2^1$}] (e) at (1,3.5) {};
\node[draw=none,fill=none,label=above:\footnotesize{$x_p^1$}] (e) at (2.5,3.5) {};
\node[draw=none,fill=none,label=above:\footnotesize{$x_1^i$}] (e) at (5.5,3.5) {};
\node[draw=none,fill=none,label=above:\footnotesize{$x_2^i$}] (e) at (6,3.5) {};
\node[draw=none,fill=none,label=above:\footnotesize{$x_p^i$}] (e) at (7.5,3.5) {};
\node[draw=none,fill=none,label=above:\footnotesize{$x_1^t$}] (e) at (10.5,3.5) {};
\node[draw=none,fill=none,label=above:\footnotesize{$x_2^t$}] (e) at (11,3.5) {};
\node[draw=none,fill=none,label=above:\footnotesize{$x_p^t$}] (e) at (12.5,3.5) {};
\draw[fill=black] (1.5+5*\k,3.5) circle (0.02);
\draw[fill=black] (1.75+5*\k,3.5) circle (0.02);
\draw[fill=black] (2+5*\k,3.5) circle (0.02);
\draw[fill=black] (1.5+5*\k,1.8) circle (0.02);
\draw[fill=black] (1.75+5*\k,1.8) circle (0.02);
\draw[fill=black] (2+5*\k,1.8) circle (0.02);
\draw (d\k)--++(0,-1)--++(0.2,0)--++(0,-1.3)--++(-0.4,0)--++(0,1.3)--++(0.2,0);
\draw (c\k)--++(0,-1)--++(0.2,0)--++(0,-1.3)--++(-0.4,0)--++(0,1.3)--++(0.2,0);
\draw (b\k)--++(0,-1)--++(0.2,0)--++(0,-1.3)--++(-0.4,0)--++(0,1.3)--++(0.2,0);
}

\draw (3.3,2.5)--++(0.2,0)--++(0,-1.3)--++(-0.4,0)--++(0,1.3)--++(0.2,0);
\draw (3.8,2.5)--++(0.2,0)--++(0,-1.3)--++(-0.4,0)--++(0,1.3)--++(0.2,0);
\draw (4.7,2.5)--++(0.2,0)--++(0,-1.3)--++(-0.4,0)--++(0,1.3)--++(0.2,0);
\draw (b0)--(3.3,2.5);
\draw (d0)--(3.8,2.5);
\draw (c0)--(4.7,2.5);
\draw (q4)--(3.3,1.2);
\draw (q4)--(3.8,1.2);
\draw (q4)--(4.7,1.2);
\draw (q2)--(0.5,1.2);
\draw (q2)--(1,1.2);
\draw (q2)--(2.5,1.2);
\draw (q9)--(5.5,1.2);
\draw (q9)--(6,1.2);
\draw (q9)--(7.5,1.2);
\draw (q11)--(10.5,1.2);
\draw (q11)--(11,1.2);
\draw (q11)--(12.5,1.2);
\foreach \i in {3.8,4,4.2,8.8,9,9.2}
{
\draw[fill=black] (\i,3.5) circle (0.02);
}
\draw (a5)--++(0.5,0);
\draw (a8)--++(-0.5,0);
\draw[fill=black] (6.2,0) circle (0.02);
\draw[fill=black] (6.5,0) circle (0.02);
\draw[fill=black] (6.7,0) circle (0.02);
\node[draw=none,fill=none,label=below:\footnotesize{$T_C^{1,f_1}$}] (t) at (2,0.5) {};
\node[draw=none,fill=none,label=below:\footnotesize{$T_C^{1,f_2}$}] (t) at (4,0.5) {};
\node[draw=none,fill=none,label=below:\footnotesize{$T_C^{i,f_j}$}] (t) at (9,0.5) {};
\node[draw=none,fill=none,label=below:\footnotesize{$T_C^{t,f_\ell}$}] (t) at (11,0.5) {};
\node[draw=none,fill=none,label=left:\footnotesize{$A_{f_1(x_1^1)}$}] (A) at (0.3,1.8) {};
\node[draw=none,fill=none,label=right:\footnotesize{$A_{f_j(x_p^i)}$}] (B) at (7.7,1.8) {};
\node[draw=none,fill=none,label=left:\footnotesize{$A_{f_\ell(x_1^t)}$}] (C) at (10.3,1.8) {};
\node[draw=none,fill=none,label=right:\footnotesize{$A_{f_\ell(x_p^t)}$}] (D) at (12.7,1.8) {};

\end{tikzpicture}}
\caption{The path $P_C$ for a clause $C$ and vertices $x_s^i$ for $i \in [t], s \in [p]$.}\label{fig:reduction}
\end{figure}
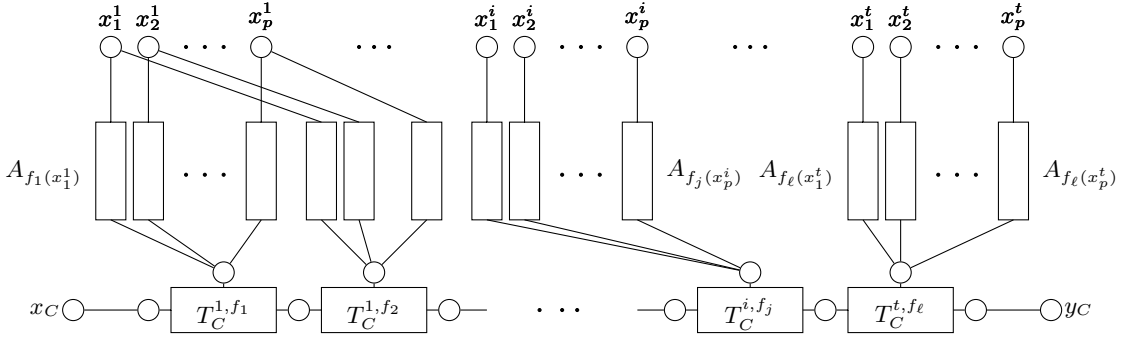
Let us show that $(G,L)$ satisfies the desired properties.

\setcounter{theorem}{18}
\begin{claim}\label{claim:equiv}
$\phi$ is satisfiable if and only if $(G,L) \to H$.
\end{claim}
\begin{claimproof}
First assume that there exists a list homomorphism $\vphi :(G,L) \to H$.
For each $i \in [t]$ consider the coloring $\vphi$ restricted to vertices $x_1^i,\ldots,x_p^i$.
If this coloring does not correspond to any truth assignment of the variables in $F_i$, we set all these variables to false.
Otherwise, we set the values of variables in $F_i$ according to the truth assignment corresponding to $\vphi|_{\{x_1^i,\ldots,x_p^i\}}$.
Let us prove that the obtained truth assignment satisfies $\phi$.

Consider any clause $C$ of $\phi$. Since $\vphi$ is a list homomorphism, we know that $\vphi(x_C)=\alpha'$ and $\vphi(y_C)=\beta'$.
Thus $\vphi$ maps the input vertex of the first switching gadget on $P_C$, i.e., the unique vertex adjacent to $x_C$, to $\alpha$.
Indeed, the list of this vertex is $\{\alpha,\beta\}$, and $\beta$ is non-adjacent to $\alpha'$ in $H$.
Similarly, the output vertex of the last switching gadget on $P_C$, i.e., the only vertex that is adjacent to $y_C$, must be mapped to $\beta$.
That implies that on $P_C$ there is at least one switching gadget $T_C^{i,f}$, whose input is mapped to $\alpha$ and the output is mapped to $\beta$.
By the property (S4.) in \cref{def:switch-gadget} (of a switching gadget), it holds that $\vphi(q_C^{i,f})=\gamma$.
Recall that in the construction of $(G,L)$ we added $T_C^{i,f}$ for a truth assignment of variables in $F_i$, which satisfies $C$,
and $f$ is a coloring of $x_1^i,\ldots,x_p^i$ corresponding to that assignment.
Moreover, the $q$-vertex of $T_C^{i,f}$, i.e., $q_C^{i,f}$, was identified with $y$-vertices of assignment gadgets $A_{v_s}$, where $v_s=f(x_s^i)$ for $s \in [p]$. On the other hand, the $x$-vertex of $A_{v_s}$ is $x_s^i$.
By the definition of an assignment gadget (property (A4.) in \cref{def:assign-gadget}), if the $y$-vertex of $A_{v_s}$ is mapped to $\gamma$, then the $x$-vertex must be mapped to $v_s$.
Thus for every $s \in [p]$ it holds that $\vphi(x_s^i)=v_s=f(x_s^i)$.
Since the values of the variables in $F_i$ were assigned according to the homomorphism $\vphi$ and $f$ corresponds to an assignment that satisfies $C$, the clause $C$ is satisfied.

Now assume that there exists a satisfying assignment $w$ of $\phi$.
Recall that for every $i \in [t]$, the assignment $w$ restricted to $F_i$ corresponds to some coloring $f_i$ of the vertices $x_1^i,\ldots,x_p^i$.
So we can define $\vphi(x_s^i):=f_i(x_s^i)$ for every $s \in [p]$ and $i \in [t]$.
Now for every clause $C$ we choose one group $F_i$ that contains a variable satisfying $\phi$ in the assignment $w$,
it exists since every clause is satisfied by $w$.
Since $w$ restricted to $F_i$ satisfies $C$, we observe that there is a switching gadget $T_C^{i,f_i}$ introduced for the triple $(C,i,f_i)$ and the $q$-vertex of this gadget is $q_C^{i,f_i}$.
We set $\vphi(q_C^{i,f_i}):=\gamma$ and extend $\vphi$ on $P_C$ in a way that:
\begin{enumerate}[a)]
\item $\vphi(x_C):=\alpha'$,
\item all input and output vertices of switching gadgets between $x_C$ and $T_C^{i,f_i}$ on $P_C$ are mapped to $\alpha$,
\item the input vertex of $T_C^{i,f_i}$ is mapped to $\alpha$ and the output vertex is mapped to $\beta$ (recall that $q_C^{i,f_i}$ is already mapped to $\gamma$)
\item inputs and outputs of all remaining switching gadgets on $P_C$ are mapped to $\beta$,
\item $\vphi(y_C):=\beta'$.
\end{enumerate}
We extend $\vphi$ to all remaining vertices of switching gadgets  mentioned in b) and d). We do it in a way, that every $q$-vertex is mapped to $\alpha$ or $\beta$ (it is possible by property (S2.) in \cref{def:switch-gadget}).
Next, we can extend $\vphi$ to remaining vertices of $T_C^{i,f_i}$. Note that this is possible by property (S3.) in \cref{def:switch-gadget}.

It only remains to extend $\vphi$ to the vertices from assignment gadgets.
Observe that if for some assignment gadget $A_{v_s}$ for $s \in [p]$ its $y$-vertex is mapped to one of $\{\alpha,\beta\}$ and $x$-vertex is mapped to $u \in S$, then we can always extend $\vphi$ to all remaining vertices of $A_{v_s}$ by property (A2.) in \cref{def:assign-gadget}.
So assume that the $y$-vertex of $A_{v_s}$ is mapped to $\gamma$.
Recall that the only $q$-vertex on $P_C$ that was mapped to $\gamma$ is $q_C^{i,f_i}$, where $f_i$ corresponds to the assignment $w$ restricted to $F_i$.
Thus if the $y$-vertex of the assignment gadget $A_{v_s}$ is mapped to $\gamma$, then that $y$-vertex is exactly $q_C^{i,f_i}$ and by the construction of $(G,L)$ it holds that $v_s=f_i(x_s^i)$.
Moreover, for every $s \in [p]$, the $y$-vertex of the assignment gadget $A_{f_i(x_s^i)}$ is $q_C^{i,f_i}$, again it follows from the construction of $(G,L)$. Thus all assignment gadgets whose $y$-vertex is mapped by $\vphi$ to $\gamma$ are exactly $A_{f_i(x_s^i)}$ for $s \in [p]$. Since $x$-vertices of these gadgets, i.e., $x_1^i,\ldots,x_p^i$, are already mapped according to $f_i$, we can also extend $\vphi$ to all vertices of the gadgets by the property (A3.) in \cref{def:assign-gadget}.
\end{claimproof}

\begin{claim}\label{claim:fvs}
The set $\bigcup_{i=1}^t \{x^i_1,\ldots,x^i_p\}$ is a feedback vertex set in $G$.
\end{claim}

\begin{claimproof}
First observe that $G$ consists of: pairwise disjoint paths $P_C$,  the independent set $\bigcup_{i=1}^t \{x^i_1,\ldots,x^i_p\}$, and assignment gadgets $A_v$.
Moreover, every assignment gadget $A_v$ has exactly one common vertex with exactly one path $P_C$,
and $x$-vertices of all assignment gadgets are contained in $\bigcup_{i=1}^t \{x^i_1,\ldots,x^i_p\}$.
By the definition of an assignment gadget, every cycle entirely contained in $A_v$ contains the $x$-vertex of $A_v$ (property (A5.) in \cref{def:assign-gadget}). Thus, every cycle in $G$ contains a vertex from $\bigcup_{i=1}^t \{x^i_1,\ldots,x^i_p\}$.
\end{claimproof}

Suppose that the instance $(G,L)$ of $\lhomo H$ can be solved in time $(k-\eps)^{\fvs{G}} \cdot |V(G)|^{\Oh(1)}$.
By \cref{claim:fvs} there is a feedback vertex set in $G$ of size $t \cdot p$ and thus $\fvs{G} \leq t \cdot p$. Moreover, $|V(G)|=(n+m)^{\Oh(1)}$.
Hence, $(k-\eps)^{\fvs{G}} \cdot |V(G)|^{\Oh(1)}\leq (k-\eps)^{t \cdot p} \cdot (n+m)^{\Oh(1)}$. By \cref{claim:equiv} solving the instance $\phi$ of \textsc{CNF-Sat} is equivalent to solving the instance $(G,L)$ of $\lhomo H$ and thus \textsc{CNF-Sat} can be solved in time:
\begin{equation}\label{eq:reduction1}
(k-\eps)^{t \cdot p}\cdot (n+m)^{\Oh(1)}=k^{\log_k(k-\eps) \cdot p \cdot t}\cdot (n+m)^{\Oh(1)}=k^{\lambda \cdot p \cdot t}\cdot (n+m)^{\Oh(1)}=k^{\lambda \cdot p \cdot \big{\lceil} \frac{n}{\lfloor p \cdot \log k \rfloor} \big{\rceil}}\cdot (n+m)^{\Oh(1)}.
\end{equation}

Let us analyze the exponent more carefully:

\begin{equation}\label{eq:reduction2}
\lambda \cdot p \cdot \Big{\lceil} \frac{n}{\lfloor p \cdot \log k \rfloor} \Big{\rceil} \le \lambda \cdot p \cdot \Bigg{(} \frac{n}{\lfloor p \cdot \log k \rfloor} + 1 \Bigg{)} \le \lambda \cdot p \cdot \Bigg{(} \frac{n}{ p \cdot \log k - 1} + 1 \Bigg{)} \le \lambda \cdot p \cdot \Bigg{(} \frac{n}{ (p-1) \cdot \log k} + 1 \Bigg{)}.
\end{equation}

By the choice of $p$ it holds that:

\begin{equation}\label{eq:reduction3}
\lambda \cdot p \cdot \Bigg{(} \frac{n}{ (p-1) \cdot \log k} + 1 \Bigg{)} = \lambda \cdot \frac{p}{p-1} \cdot \frac{n}{ \log k} + \lambda \cdot p  \le \delta' \cdot \frac{n}{\log k} + \lambda \cdot p,
\end{equation}

where $\delta'<1$. Recall that $p$ and $\lambda$ do not depend on $n$ and $m$. Thus the instance $\phi$ of \textsc{CNF-Sat} can be solved in time:

\begin{equation}\label{eq:reduction4}
k^{\delta' \cdot \frac{n}{\log k} + \lambda \cdot p} \cdot (n+m)^{\Oh(1)}=
k^{\delta' \cdot \frac{n}{\log k}} \cdot (n+m)^{\Oh(1)}= 2^{\delta' \cdot n} \cdot (n+m)^{\Oh(1)}=(2-\delta)^n \cdot (n+m)^{\Oh(1)}
\end{equation}
for some $\delta>0$, which contradicts the SETH.
\end{proof}

\setcounter{theorem}{22}
Let us point out that the pathwidth of the graph constructed in the proof of \cref{thm:bipartite-undecomp} is bounded by $t \cdot p + f(H)$, for some function $f$ of $H$ (see also \cite{DBLP:conf/soda/LokshtanovMS11a}).
Furthermore, we note that the constructed instance satisfies conditions 1. and 2. of \cref{def:consist-inst}.
Moreover, in any instance of $\lhomo {H}$, if a list $L(v)$ contains vertices $x,y$ such that $N(x)\subseteq N(y)$, then we can safely remove $x$ from the list. Thus we actually proved the following.

\begin{remark}\label{cor:consist-hard}
\cref{thm:bipartite} and  \cref{thm:bipartite-undecomp} hold, even if we assume that the instance $(G,L)$ is consistent.
\end{remark}

\subsection{Lower bound for general target graphs}
In this section we extend our results from \cref{thm:bipartite} to the general case, i.e., we do not assume that the graph $H$ is bipartite. In particular, we allow loops in $H$. Recall that by $H^*$ we denote a graph whose vertex set is $V(H^*)=\{v',v'' \ | \ v \in V(H)\}$ and there is an edge $v'u''$ in $H^*$ if and only if $vu\in E(H)$.

\begin{cthm}{2~b)}
Let $H$ be a connected non-bi-arc graph.
Even if $H$ is fixed, there is no algorithm that solves every instance $(G,L)$ of  $\lhomo{H}$ in time $(i^*(H)-\epsilon)^{\fvs{G}} \cdot |V(G)|^{\Oh(1)}$ for any $\epsilon > 0$, unless the SETH fails.
\end{cthm}

\begin{proof}
For contradiction, suppose that there exists a connected non-bi-arc graph $H$, a constant $\eps>0$, and an algorithm $A$ that solves $\lhomo H$ for every instance $(G,L)$ in time $(i^*(H)-\eps)^{\fvs{G}} \cdot |V(G)|^{\Oh(1)}$. We can assume that $H$ is non-bipartite, otherwise by \cref{thm:bipartite} we get a contradiction with the SETH. 

Consider a consistent instance $(G,L)$ of $\lhomo {H^*}$. Let $(G,L')$ be an instance of $\lhomo H$ obtained as in \cref{prop:bipartite-associted}, clearly it can be constructed in polynomial time. Recall that if $H$ is a non-bipartite, connected, non-bi-arc graph, then $H^*$ is connected bipartite graph, whose complement is not a circular-arc graph and thus $H^*$ satifies the assumptions of \cref{thm:bipartite}. Moreover, $i^*(H)=i^*(H^*)$.

We can use the algorithm $A$ to solve the instance $(G,L')$ of $\lhomo H$ in time $(i^*(H)-\eps)^{\fvs{G}} \cdot |V(G)|^{\Oh(1)}$, which, by \cref{prop:bipartite-associted}, is equivalent to solving the instance $(G,L)$ of $\lhomo {H^*}$ in time $(i^*(H^*)-\eps)^{\fvs{G}} \cdot |V(G)|^{\Oh(1)}$. By \cref{thm:bipartite} and \cref{cor:consist-hard} it is a contradiction with the SETH.
\end{proof}

%

\newpage
\section{Parameter: cutwidth}
\subsection{Lower bounds for \lhomo{H}, bipartite target graphs}
Similarly as for feedback vertex set let us first prove \cref{thm:main-ctw-list-hard} in bipartite case.
Recall that for a positive integer $g$, by $\Cg$ we denote the class of all graphs $G$ that are bipartite, with maximum degree 3, girth at least $g$, and all vertices of degree $3$ in $G$  are pairwise distance at least $g$.

\begin{theorem}\label{thm:ctw-bipartite}
Let $\mathcal{H}$ be the class of connected, bipartite graphs, whose complement is not a circular-arc graph, and let $g \in \N$.
\begin{enumerate}[a)]
\item For every $H \in \mathcal{H}$, there is no algorithm that solves every instance $(G,L)$ of  $\lhomo H$, where $G \in \mathcal{C}_g$, in time $(mim^*(H)-\epsilon)^{\ctw{G}} \cdot |V(G)|^{\Oh(1)}$ for any $\epsilon > 0$, unless the SETH fails.
\item There exists a universal constant $0<\delta<1$, such that for every $H \in \mathcal{H}$, there is no algorithm that solves every instance $(G,L)$ of $\lhomo H$, where $G \in \mathcal{C}_g$, in time $mim^*(H)^{\delta \cdot \ctw{G}} \cdot |V(G)|^{\Oh(1)}$, unless the ETH fails.
\end{enumerate}
\end{theorem}

Similarly to the case of \cref{thm:bipartite}, it is sufficient to show the following.

\begin{theorem}\label{thm:ctw-bipartite-undecomp}
Let $\mathcal{H}'$ be the class of connected, undecomposable, bipartite graphs, whose complement is not a circular-arc graph, and let $g \in \N$.
\begin{enumerate}[a)]
\item For every $H \in \mathcal{H}'$, there is no algorithm that solves every instance $(G,L)$ of  $\lhomo H$, where $G \in \mathcal{C}_g$, in time $(mim^*(H)-\epsilon)^{\ctw{G}} \cdot |V(G)|^{\Oh(1)}$ for any $\epsilon > 0$, unless the SETH fails.
\item There exists a universal constant $0<\delta<1$, such that for every $H \in \mathcal{H}'$, there is no algorithm that solves every instance $(G,L)$ of $\lhomo H$, where $G \in \mathcal{C}_g$, in time $mim^*(H)^{\delta \cdot \ctw{G}} \cdot |V(G)|^{\Oh(1)}$, unless the ETH fails.
\end{enumerate}
\end{theorem}

It can be shown that \cref{thm:ctw-bipartite}~a) is equivalent to \cref{thm:ctw-bipartite-undecomp}~a) and that \cref{thm:ctw-bipartite}~b) is equivalent to \cref{thm:ctw-bipartite-undecomp}~b). Since the proofs of the equivalence are analogous, let us show only one of them.

\medskip
\noindent\textbf{(\cref{thm:ctw-bipartite}~a) $\rightarrow$  \cref{thm:ctw-bipartite-undecomp}~a))} Assume the SETH and suppose that \cref{thm:ctw-bipartite}~a) holds and \cref{thm:ctw-bipartite-undecomp}~a) fails. Then there exist a connected, bipartite, udecomposable graph $H$, whose complement is not a circular-arc graph, $g \in \N$, $\eps>0$, and an algorithm $A$ that solves every instance $(G,L)$ of  $\lhomo H$  such that $G\in \Cg$ in time $(mim(H)-\epsilon)^{\ctw{G}} \cdot |V(G)|^{\Oh(1)}$. The properties of $H$ imply that $mim^*(H)=mim(H)$. Thus every instance $(G,L)$ of $\lhomo H$ such that $G\in \Cg$ can be solved in time $(mim^*(H)-\epsilon)^{\ctw{G}} \cdot |V(G)|^{\Oh(1)}$. By \cref{thm:ctw-bipartite}~a) we get a contradiction with the SETH.

\medskip
\noindent\textbf{(\cref{thm:ctw-bipartite-undecomp}~a) $\rightarrow$ \cref{thm:ctw-bipartite}~a))} Assume the SETH and suppose that \cref{thm:ctw-bipartite-undecomp}~a) holds and \cref{thm:ctw-bipartite}~a) fails. Then there exist a connected, bipartite graph $H$, whose complement is not a circular-arc graph, $g \in \N$, $\eps>0$, and an algorithm $A$ that solves every instance $(G,L)$ of $\lhomo H$ such that $G\in \Cg$ in time $(mim^*(H)-\epsilon)^{\ctw{G}} \cdot |V(G)|^{\Oh(1)}$. Let $H'$ be an induced subgraph of $H$ such that $H'$ is connected, undecomposable, is not a complement of a circular-arc graph, and $mim(H')=mim^*(H)$. Observe that any instance $(G,L)$ of $\lhomo {H'}$ can be seen as an instance of $\lhomo H$ such that only vertices of $H'$ appear on lists $L$. Thus we can use the algorithm $A$ to solve any instance $(G,L)$ of $\lhomo {H'}$ with $G\in \Cg$ in time $(mim^*(H)-\epsilon)^{\ctw{G}} \cdot |V(G)|^{\Oh(1)}=(mim(H')-\epsilon)^{\ctw{G}} \cdot |V(G)|^{\Oh(1)}$, which, by \cref{thm:ctw-bipartite-undecomp}~a), contradicts the SETH.

\medskip

Now we will show how to modify the reduction from \cref{thm:bipartite-undecomp} to prove \cref{thm:ctw-bipartite-undecomp}.
To get an intuition about what needs to be done, recall that in order to obtain a bound on the cutwidth, we need to bound the pathwidth and the maximum degree. Also, as we already observed, the pathwidth of the instance constructed in \cref{thm:bipartite-undecomp} is upper-bounded by the correct value, so we need to take care of vertices of large degree.

\begin{lemma}\label{lem:instance-lhom-ctw}
Let $g \in \N$ and let $H$ be a connected, bipartite, undecomposable graph, whose complement is not a circular-arc graph.
Let $k:=mim(H)$.
Let $\phi$ be an instance of \textsc{CNF-Sat} with $n$ variables and $m$ clauses.
Let $p$ be a positive integer and let $t:=\Big{\lceil} \frac{n}{\lfloor p \cdot \log k \rfloor} \Big{\rceil}$.
Then there exists an instance $(\tG,\tL)$ of $\lhomo{H}$ which satisfies the following properties.
\begin{enumerate}[(1.)]
\item $(\tG,\tL) \to H$ if and only if $\phi$ is satisfiable,

\item $\ctw{\tG} \leq t \cdot p + f(g,H)$, where $f$ is some function of $g$ and $H$,

\item $\tG \in \Cg$,

\item $|V(\tG)|=(n+m)
^{\Oh(1)}$.
\end{enumerate}
Moreover, $(\tG,\tL)$ can be constructed in polynomial time in $(n+m)$.
\end{lemma}

\begin{proof}
Let $S$ be a strongly incomparable set in $H$ of size $k=mim(H)$, contained in one bipartition class.
Let $S'$ be a set such that $S\cup S'$ induces a matching of size $k$ in $H$, and let $(\alpha,\beta,\gamma)$ be the triple given by \cref{obs:walks-between-corners}, such that $\alpha,\beta,\gamma$ are in the same bipartition class as $S$.
We repeat the construction of the instance $(G,L)$ of $\lhomo{H}$, such that $(G,L)\to H$ if and only if $\phi$ is satisfiable, from the proof of \cref{thm:bipartite-undecomp}.
This is possible since $S$ is in particular incomparable.
Furthermore, in the construction of $(G,L)$ we did not use the fact that $S$ was maximum, we only needed that $|S|\geq 2$, which is the case as $\{\alpha,\beta\}$ is strongly incomparable.
Although in the proof of \cref{thm:bipartite-undecomp} we did not care about girth of the gadgets and distances between vertices of degree $3$, now in the construction of $(G,L)$, while introducing gadgets from \cref{lem:assign-gadget} and \cref{lem:switching-gadget}, we introduce such assignment gadgets in which:
\begin{itemize}
\item the girth is at least $g$,
\item the distance of vertices of degree at least $3$ is at least $g$,
\item the distance between the $x$-vertex and the $y$-vertex is at least $g$,
\item every vertex of degree at least $3$ is at distance at least $g$ from the $x$-vertex and the $y$-vertex.
\end{itemize}
Similarly, we introduce  such switching gadgets with special vertices $p,q,r$, in which $\dist(p,q)\geq\frac{g}{2}$ and $\dist(r,q)\geq\frac{g}{2}$.
We are going to modify the instance $(G,L)$ into the instance $(\tG,\tL)$ with the properties listed in the statement of the lemma.

However, before we do that, let us fix an arbitrary ordering of clauses $C_1,\ldots,C_m$ in $\phi$, which implies the ordering of paths $P_C$ in $G$. 
Then we can fix an ordering of all $q$-vertices in $G$, so that a $q$-vertex $q_1$ precedes a $q$-vertex $q_2$, if:
\begin{itemize}
\item $q_1$ belongs to the path $P_{C_i}$ and $q_2$ belongs to the path $P_{C_j}$, such that $i < j$, or
\item $q_1$ and $q_2$ belong to the same path $P_C$, and $q_1$ precedes $q_2$ on $P_C$ (the order of the vertices of each path $P_C$ is such that $x_C$ is the first vertex and $y_C$ is the last vertex). 
\end{itemize}

Finally, let us fix an ordering of the assignment gadgets in $G$ as follows. Recall that every $q$-vertex $q_C^{i,f}$ is a $y$-vertex of $p$ assignment gadgets whose $x$-vertices are, respectively, $x^i_1,\ldots,x^i_p$. We fix an ordering of the assignment gadgets so that the assignment gadget $A_1$ precedes the assignment gadget $A_2$ if:

\begin{itemize}
\item the $y$-vertex of $A_1$ precedes the $y$-vertex of $A_2$ in the fixed order of the $q$-vertices, or

\item $A_1$ and $A_2$ have the same $y$-vertex $q_C^{i,f}$ and $x$-vertices of $A_1$ and $A_2$ are, respectively, $x^i_j$ and $x^i_s$, with $j<s$.
\end{itemize}

Now we are ready to modify the instance $(G,L)$. It turns out that we only need to take care of $q$-vertices and $x$-vertices, as their large degree forces large cutwidth. The construction of $(\tG,\tL)$ will be thus performed in two steps.

\paragraph{Step 1. Splitting $q$-vertices.}
Recall that every $q$-vertex of a switching gadget is a $y$-vertex of $p$ assignment gadgets and the degree of each $y$-vertex in the assignment gadget is $k-1$. For every $q$-vertex $q$, in order to reduce its degree, we will split $q$ into $p\cdot (k-1)$ vertices $q_1,\ldots,q_{p\cdot(k-1)}$. In this step, the construction depends on the structure of $H$. Let us consider two cases.

\smallskip
\textbf{Case I. The set $\{\alpha,\beta,\gamma\}$ is strongly incomparable.}  Let $\oalpha,\obeta,\ogamma$ be vertices such that edges $\alpha\oalpha,\beta\obeta,\gamma\ogamma$ induce a matching in $H$. We split every $q$-vertex $q$ from a path $P_C$ into $p\cdot (k-1)$ vertices $q_1,\ldots,q_{p \cdot (k-1)}$, each for every neighbor of $q$ inside assignment gadgets. After this operation we remove $q$ from the graph. For every $j\in [p\cdot(k-1)-1]$ we introduce a path $Q_j$ of even length which is at least $g$, with lists of consecutive vertices $\{\alpha,\beta,\gamma\},\{\oalpha,\obeta,\ogamma\},\ldots,\{\alpha,\beta,\gamma\}$, and we identify its endvertices with $q_j$ and $q_{j+1}$. In the same way, we introduce paths $Q_{0}$ and $Q_{p\cdot(k-1)}$ and we identify endvertices of $Q_{0}$ with $q_{1}$ and the vertex preceding $q$ on $P_C$, and we identify endvertices of $Q_{p\cdot(k-1)}$ with $q_{p\cdot(k-1)}$ and the vertex following $q$ on $P_C$ (see \cref{fig:spliting-q}). Finally, let us fix an ordering $a_1,a_2,\ldots,a_{p\cdot (k-1)}$ of neighbors of $q$ in assignment gadgets such that for $j \in [p-1]$ vertices of the assignment gadget with the $x$-vertex $x_j^i$ precede vertices of the assignment gadget with the $x$-vertex $x_{j+1}^i$. The order of the neighbors from the same assignment gadget is arbitrary. For every $j \in [p\cdot(k-1)]$ we add an edge between $q_{j}$ and $a_j$ (see \cref{fig:spliting-q}). This completes the step of splitting $q$-vertices in this case.

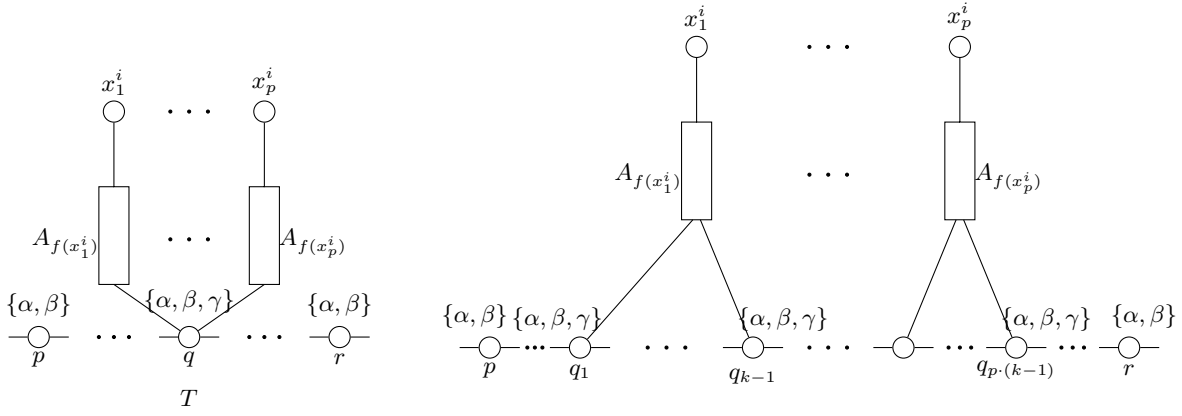
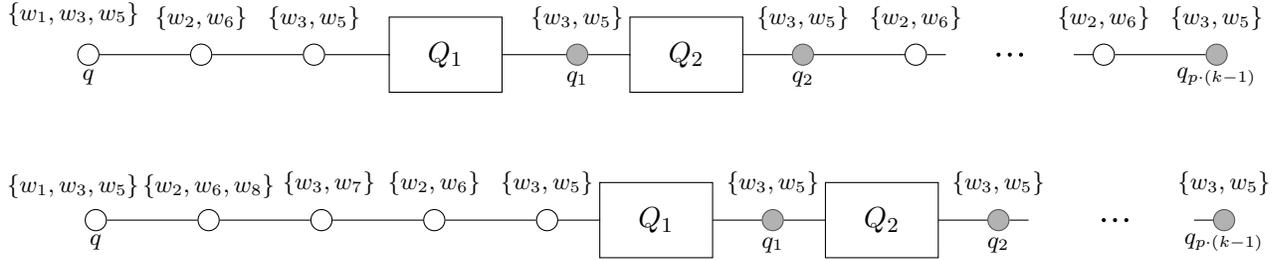
\begin{figure}
\centering{\begin{tikzpicture}[every node/.style={draw,circle,fill=white,inner sep=0pt,minimum size=8pt},every loop/.style={}]

\node[label=below:\footnotesize{$q$}] (q2) at (2,0.5) {};
\node[draw=none,fill=none,label=above:\footnotesize{$\{\alpha,\beta,\gamma\}$}] (q2l) at (2,0.2) {};
\node[label=below:\footnotesize{$p$}] (p) at (0,0.5) {};
\node[draw=none,fill=none,label=above:\footnotesize{$\{\alpha,\beta\}$}] (pl) at (0,0.3) {};
\node[label=below:\footnotesize{$r$}] (r) at (4,0.5) {};
\node[draw=none,fill=none,label=above:\footnotesize{$\{\alpha,\beta\}$}] (rl) at (4,0.3) {};
\draw (p)--++(0.4,0);
\draw (p)--++(-0.4,0);
\draw (r)--++(0.4,0);
\draw (r)--++(-0.4,0);
\draw (q2)--++(0.4,0);
\draw (q2)--++(-0.4,0);
\foreach \k in {0.8,1,1.2,2.8,3,3.2}
{
\draw[fill=black] (\k,0.5) circle (0.02);
}

\node (b0) at (1,3.5) {};
\node (c0) at (3,3.5) {};

\node[draw=none,fill=none,label=above:\footnotesize{$x_1^i$}] (e) at (1,3.5) {};

\node[draw=none,fill=none,label=above:\footnotesize{$x_p^i$}] (e) at (3,3.5) {};

\draw[fill=black] (2.25,3.5) circle (0.02);
\draw[fill=black] (1.75,3.5) circle (0.02);
\draw[fill=black] (2,3.5) circle (0.02);
\draw[fill=black] (2.25,1.8) circle (0.02);
\draw[fill=black] (1.75,1.8) circle (0.02);
\draw[fill=black] (2,1.8) circle (0.02);

\draw (c0)--++(0,-1)--++(0.2,0)--++(0,-1.3)--++(-0.4,0)--++(0,1.3)--++(0.2,0);
\draw (b0)--++(0,-1)--++(0.2,0)--++(0,-1.3)--++(-0.4,0)--++(0,1.3)--++(0.2,0);

\draw (q2)--(1,1.2);
\draw (q2)--(3,1.2);

\node[draw=none,fill=none,label=below:\footnotesize{$T$}] (t) at (2,0) {};
\node[draw=none,fill=none,label=left:\footnotesize{$A_{f(x_1^i)}$}] (af1) at (1,1.75) {};
\node[draw=none,fill=none,label=right:\footnotesize{$A_{f(x_p^i)}$}] (afp) at (3,1.75) {};

\end{tikzpicture} \qquad
\begin{tikzpicture}[every node/.style={draw,circle,fill=white,inner sep=0pt,minimum size=8pt},every loop/.style={}]

\node[label=below:\footnotesize{$p$}] (p) at (-2.5,0.5) {};
\node[draw=none,fill=none,label=above:\footnotesize{$\{\alpha,\beta\}$}] (pa) at (-2.7,0.3) {};
\node[label=below:\footnotesize{$r$}] (r) at (6,0.5) {};
\node[draw=none,fill=none,label=above:\footnotesize{$\{\alpha,\beta\}$}] (ra) at (6.2,0.3) {};
\node (q0) at (-1.3,0.5) {};
\node[draw=none,fill=none,label=below:\footnotesize{$q_1$}] (q0p) at (-1.3,0.5) {};
\node[draw=none,fill=none,label=above:\footnotesize{$\{\alpha,\beta,\gamma\}$}] (q0a) at (-1.6,0.1) {};
\node (q1) at (1,0.5) {};
\node[draw=none,fill=none,label=below:\footnotesize{$q_{k-1}$}] (q1p) at (1,0.6) {};
\node[draw=none,fill=none,label=above:\footnotesize{$\{\alpha,\beta,\gamma\}$}] (q1a) at (1.4,0.1) {};
\node (q2) at (3,0.5) {};
\node (q3) at (4.5,0.5) {};
\node[draw=none, fill=none,label=below:\footnotesize{$q_{p\cdot (k-1)}$}] (q3p) at (4.5,0.9) {};
\node[draw=none,fill=none,label=above:\footnotesize{$\{\alpha,\beta,\gamma\}$}] (q3a) at (4.9,0.1) {};
\draw (p)--++(0.4,0);
\draw (p)--++(-0.4,0);
\draw (r)--++(0.4,0);
\draw (r)--++(-0.4,0);
\draw (q2)--++(0.4,0);
\draw (q2)--++(-0.4,0);
\draw (q0)--++(0.4,0);
\draw (q0)--++(-0.4,0);
\draw (q1)--++(0.4,0);
\draw (q1)--++(-0.4,0);
\draw (q3)--++(0.4,0);
\draw (q3)--++(-0.4,0);
\foreach \k in {-2,-1.9,-1.8, -0.4,-0.15,0.1, 1.75,2,2.25, 3.6,3.75,3.9, 5.1,5.25,5.4}
{
\draw[fill=black] (\k,0.5) circle (0.02);
}

\node (b0) at (0.25,4.5) {};
\node (c0) at (3.75,4.5) {};

\node[draw=none,fill=none,label=above:\footnotesize{$x_1^i$}] (e) at (0.25,4.5) {};

\node[draw=none,fill=none,label=above:\footnotesize{$x_p^i$}] (e) at (3.75,4.5) {};

\draw[fill=black] (1.75,4.5) circle (0.02);
\draw[fill=black] (2,4.5) circle (0.02);
\draw[fill=black] (2.25,4.5) circle (0.02);
\draw[fill=black] (1.75,2.8) circle (0.02);
\draw[fill=black] (2.25,2.8) circle (0.02);
\draw[fill=black] (2,2.8) circle (0.02);
\draw (c0)--++(0,-1)--++(0.2,0)--++(0,-1.3)--++(-0.4,0)--++(0,1.3)--++(0.2,0);
\draw (b0)--++(0,-1)--++(0.2,0)--++(0,-1.3)--++(-0.4,0)--++(0,1.3)--++(0.2,0);

\draw (q0)--(0.2,2.2);
\draw (q1)--(0.3,2.2);
\draw (q2)--(3.7,2.2);
\draw (q3)--(3.8,2.2);

\node[draw=none,fill=none,label=left:\footnotesize{$A_{f(x_1^i)}$}] (af1) at (0.25,2.75) {};
\node[draw=none,fill=none,label=right:\footnotesize{$A_{f(x_p^i)}$}] (afp) at (3.75,2.75) {};
\end{tikzpicture}
}
\caption{The switching gadget $T$ and the group of vertices $x_s^i$ for $s \in [p]$ before the step of splitting $q$-vertices (left) and after the step in the case that $\{\alpha,\beta,\gamma\}$ is a strongly incomparable set (right).}\label{fig:spliting-q}
\end{figure}

\begin{figure}
\centering{
\begin{subfigure}{.99\textwidth}
\centering{
\begin{tikzpicture}[every node/.style={draw,circle,fill=white,inner sep=0pt,minimum size=8pt},every loop/.style={}]
\node (q) at (-3.5,0) {};
\node[draw=none,fill=none,label=below:\footnotesize{$q$}] (qn) at (-3.5,0) {};
\node[draw=none,fill=none,label=above:\footnotesize{$\{w_1,w_3,w_5\}$}] (ql) at (-3.7,-0.5) {};
\node (q') at (-2,0) {};
\node[draw=none,fill=none,label=above:\footnotesize{$\{w_2,w_6\}$}] (q'l) at (-2,-0.3) {};
\node (q'') at (-0.5,0) {};
\node[draw=none,fill=none,label=above:\footnotesize{$\{w_3,w_5\}$}] (q''l) at (-0.5,-0.3) {};

\node[fill=gray,opacity=0.6] (q1) at (3,0) {};
\node[draw=none,fill=none,label=below:\footnotesize{$q_{1}$}] (q1n) at (3,0) {};
\node[draw=none,fill=none,label=above:\footnotesize{$\{w_3,w_5\}$}] (q1l) at (3,-0.3) {};
\node[fill=gray,opacity=0.6] (q2) at (6,0) {};
\node[draw=none,fill=none,label=below:\footnotesize{$q_{2}$}] (q2n) at (6,0) {};
\node[draw=none,fill=none,label=above:\footnotesize{$\{w_3,w_5\}$}] (q2l) at (6,-0.3) {};
\node (q2') at (7.5,0) {};
\node[draw=none,fill=none,label=above:\footnotesize{$\{w_2,w_6\}$}] (q2'l) at (7.5,-0.3) {};
\node (q3') at (10,0) {};
\node[draw=none,fill=none,label=above:\footnotesize{$\{w_2,w_6\}$}] (q3'l) at (10,-0.3) {};
\node[fill=gray,opacity=0.6] (q3) at (11.5,0) {};
\node[draw=none,fill=none,label=below:\footnotesize{$q_{p\cdot (k-1)}$}] (q3n) at (11.5,0.4) {};
\node[draw=none,fill=none,label=above:\footnotesize{$\{w_3,w_5\}$}] (q3l) at (11.5,-0.3) {};
\draw (q)--(q')--(q'')--++(1,0)--++(0,0.5)--++(1.5,0)--++(0,-1)--++(-1.5,0)--++(0,0.5);
\draw (2,0)--(q1);
\draw (q1)--++(0.7,0)--++(0,0.5)--++(1.5,0)--++(0,-1)--++(-1.5,0)--++(0,0.5);
\draw (5.2,0)--(q2);
\draw (q2)--(q2');
\node[draw=none, fill=none] (Q0) at (1.25,0) {$Q_1$};
\node[draw=none, fill=none] (Q1) at (4.45,0) {$Q_2$};
\draw (q3)--(q3');
\draw (q2')--++(0.4,0);
\draw (q3')--++(-0.4,0);
\draw[fill=black] (8.6,0) circle (0.02);
\draw[fill=black] (8.75,0) circle (0.02);
\draw[fill=black] (8.9,0) circle (0.02);
\end{tikzpicture}

\begin{tikzpicture}[every node/.style={draw,circle,fill=white,inner sep=0pt,minimum size=8pt},every loop/.style={}]
\node (q) at (-6,0) {};
\node[draw=none,fill=none,label=below:\footnotesize{$q$}] (qn) at (-6,0) {};
\node[draw=none,fill=none,label=above:\footnotesize{$\{w_1,w_3,w_5\}$}] (ql) at (-6.3,-0.6) {};
\node (q') at (-4.5,0) {};
\node[draw=none,fill=none,label=above:\footnotesize{$\{w_2,w_6,w_8\}$}] (q'l) at (-4.5,-0.6) {};
\node (q0) at (-3,0) {};
\node[draw=none,fill=none,label=above:\footnotesize{$\{w_3,w_7\}$}] (q0l) at (-2.9,-0.3) {};
\node (q0') at (-1.5,0) {};
\node[draw=none,fill=none,label=above:\footnotesize{$\{w_2,w_6\}$}] (q0'l) at (-1.5,-0.3) {};
\node (q0'') at (0,0) {};
\node[draw=none,fill=none,label=above:\footnotesize{$\{w_3,w_5\}$}] (q0''l) at (0,-0.3) {};
\node[fill=gray,opacity=0.6] (q1) at (3,0) {};
\node[draw=none,fill=none,label=below:\footnotesize{$q_{1}$}] (q1n) at (3,0) {};
\node[draw=none,fill=none,label=above:\footnotesize{$\{w_3,w_5\}$}] (q1l) at (3,-0.3) {};

\node[fill=gray,opacity=0.6] (q2) at (6,0) {};
\node[draw=none,fill=none,label=below:\footnotesize{$q_{2}$}] (q2n) at (6,0) {};
\node[draw=none,fill=none,label=above:\footnotesize{$\{w_3,w_5\}$}] (q2l) at (6,-0.3) {};
\node[fill=gray,opacity=0.6] (q3) at (9,0) {};
\node[draw=none,fill=none,label=below:\footnotesize{$q_{p\cdot (k-1)}$}] (q3n) at (9,0.4) {};
\node[draw=none,fill=none,label=above:\footnotesize{$\{w_3,w_5\}$}] (q3l) at (9,-0.3) {};
\draw (q)--(q')--(q0)--(q0')--(q0'');
\draw (q2)--++(0.4,0);
\draw (q3)--++(-0.4,0);
\draw (q0'')--++(0.7,0)--++(0,0.5)--++(1.5,0)--++(0,-1)--++(-1.5,0)--++(0,0.5);
\draw (q1)--++(0.7,0)--++(0,0.5)--++(1.5,0)--++(0,-1)--++(-1.5,0)--++(0,0.5);
\draw (2.2,0)--(q1);
\draw (5.2,0)--(q2);
\node[draw=none,fill=none] (Q1) at (1.45,0) {$Q_1$};
\node[draw=none,fill=none] (Q2) at (4.45,0) {$Q_2$};

\draw[fill=black] (7.4,0) circle (0.02);
\draw[fill=black] (7.55,0) circle (0.02);
\draw[fill=black] (7.7,0) circle (0.02);
\end{tikzpicture}}
\caption{The construction of the path $Q$ in case that $H$ contains an induced $C_6$ with consecutive vertices $w_1,\ldots,w_6$ and $\alpha=w_1,\beta=w_5,\gamma=w_3$ (above) and in case that $H$ contains an induced $C_8$ with consecutive vertices $w_1,\ldots,w_8$ and $\alpha=w_1,\beta=w_5,\gamma=w_3$ (below).
}
\end{subfigure}

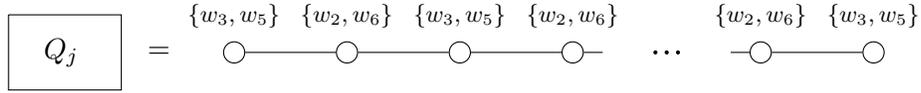
\begin{subfigure}{.99\textwidth}
\centering
\begin{tikzpicture}[every node/.style={draw,circle,fill=white,inner sep=0pt,minimum size=8pt},every loop/.style={}]
\node (q1) at (3,0) {};
\node[draw=none,fill=none,label=above:\footnotesize{$\{w_3,w_5\}$}] (q1l) at (3,-0.3) {};
\node (q1') at (4.5,0) {};
\node[draw=none,fill=none,label=above:\footnotesize{$\{w_2,w_6\}$}] (q1'l) at (4.5,-0.3) {};
\node (q2) at (6,0) {};
\node[draw=none,fill=none,label=above:\footnotesize{$\{w_3,w_5\}$}] (q2l) at (6,-0.3) {};
\node (q2') at (7.5,0) {};
\node[draw=none,fill=none,label=above:\footnotesize{$\{w_2,w_6\}$}] (q2'l) at (7.5,-0.3) {};
\node (q3') at (10,0) {};
\node[draw=none,fill=none,label=above:\footnotesize{$\{w_2,w_6\}$}] (q3'l) at (10,-0.3) {};
\node (q3) at (11.5,0) {};
\node[draw=none,fill=none,label=above:\footnotesize{$\{w_3,w_5\}$}] (q3l) at (11.5,-0.3) {};
\draw (q1)--(q1')--(q2)--(q2');
\draw (q3)--(q3');
\draw (q2')--++(0.4,0);
\draw (q3')--++(-0.4,0);
\draw[fill=black] (8.6,0) circle (0.02);
\draw[fill=black] (8.75,0) circle (0.02);
\draw[fill=black] (8.9,0) circle (0.02);
\draw (0,0.5)--++(1.5,0)--++(0,-1)--++(-1.5,0)--++(0,1);
\node[draw=none,fill=none] (Qj) at (0.7,0) {$Q_j$};
\node[draw=none,fill=none] (eq) at (2,0) {$=$};
\end{tikzpicture}
\caption{The construction of each subpath $Q_j$, for $j=0,\ldots p\cdot (k-1)$, used in the construction of $Q$.
The length of each $Q_j$ is at least $g$.
}
\end{subfigure}
}
\caption{The construction of the path $Q$ from the step of splitting $q$-vertices, case 2. The new $q$-vertices are marked gray.
}\label{fig:spliting-q2}
\end{figure}

\begin{figure}
\centering{\begin{tikzpicture}[every node/.style={draw,circle,fill=white,inner sep=0pt,minimum size=8pt},every loop/.style={}]

\node[label=below:\footnotesize{$q$}] (q2) at (2,0.5) {};
\node[draw=none,fill=none,label=above:\footnotesize{$\{\alpha,\beta,\gamma\}$}] (q2l) at (2,0.2) {};
\node[label=below:\footnotesize{$p$}] (p) at (0,0.5) {};
\node[draw=none,fill=none,label=above:\footnotesize{$\{\alpha,\beta\}$}] (pl) at (0,0.3) {};
\node[label=below:\footnotesize{$r$}] (r) at (4,0.5) {};
\node[draw=none,fill=none,label=above:\footnotesize{$\{\alpha,\beta\}$}] (rl) at (4,0.3) {};
\draw (p)--++(0.4,0);
\draw (p)--++(-0.4,0);
\draw (r)--++(0.4,0);
\draw (r)--++(-0.4,0);
\draw (q2)--++(0.4,0);
\draw (q2)--++(-0.4,0);
\foreach \k in {0.8,1,1.2,2.8,3,3.2}
{
\draw[fill=black] (\k,0.5) circle (0.02);
}

\node (b0) at (1,3.5) {};
\node (c0) at (3,3.5) {};

\node[draw=none,fill=none,label=above:\footnotesize{$x_1^i$}] (e) at (1,3.5) {};

\node[draw=none,fill=none,label=above:\footnotesize{$x_p^i$}] (e) at (3,3.5) {};

\draw[fill=black] (2.25,3.5) circle (0.02);
\draw[fill=black] (1.75,3.5) circle (0.02);
\draw[fill=black] (2,3.5) circle (0.02);
\draw[fill=black] (2.25,1.8) circle (0.02);
\draw[fill=black] (1.75,1.8) circle (0.02);
\draw[fill=black] (2,1.8) circle (0.02);

\draw (c0)--++(0,-1)--++(0.2,0)--++(0,-1.3)--++(-0.4,0)--++(0,1.3)--++(0.2,0);
\draw (b0)--++(0,-1)--++(0.2,0)--++(0,-1.3)--++(-0.4,0)--++(0,1.3)--++(0.2,0);

\draw (q2)--(1,1.2);
\draw (q2)--(3,1.2);

\node[draw=none,fill=none,label=below:\footnotesize{$T$}] (t) at (2,0) {};
\node[draw=none,fill=none,label=left:\footnotesize{$A_{f(x_1^i)}$}] (af1) at (1,1.75) {};
\node[draw=none,fill=none,label=right:\footnotesize{$A_{f(x_p^i)}$}] (afp) at (3,1.75) {};

\end{tikzpicture} \qquad
\begin{tikzpicture}[every node/.style={draw,circle,fill=white,inner sep=0pt,minimum size=8pt},every loop/.style={}]

\node[label=below:\footnotesize{$p$}] (p) at (-3,-1) {};
\node[draw=none,fill=none,label=above:\footnotesize{$\{\alpha,\beta\}$}] (pa) at (-3.2,-1.2) {};
\node[label=below:\footnotesize{$r$}] (r) at (0,-1) {};
\node[draw=none,fill=none,label=above:\footnotesize{$\{\alpha,\beta\}$}] (ra) at (0.5,-1.4) {};
\node[label=below:\footnotesize{$q$}] (q) at (-1.5,-1) {};
\node (q0) at (-1.3,0.5) {};
\node[draw=none,fill=none,label=right:\footnotesize{$q_1$}] (q0p) at (-1.3,0.5) {};
\node[draw=none,fill=none,label=above:\footnotesize{$\{\beta,\gamma\}$}] (q0a) at (-1.6,0.1) {};
\node (q1) at (1,0.5) {};
\node[draw=none,fill=none,label=right:\footnotesize{$q_{k-1}$}] (q1p) at (1,0.5) {};
\node[draw=none,fill=none,label=above:\footnotesize{$\{\beta,\gamma\}$}] (q1a) at (1.4,0.1) {};
\node (q2) at (3,0.5) {};
\node (q3) at (4.5,0.5) {};
\node[draw=none, fill=none,label=right:\footnotesize{$q_{p\cdot (k-1)}$}] (q3p) at (4.5,0.5) {};
\node[draw=none,fill=none,label=above:\footnotesize{$\{\beta,\gamma\}$}] (q3a) at (4.9,0.1) {};
\node[draw=none,fill=none] (Q) at (1.5,-0.25) {$Q$};

\draw (p)--++(0.4,0);
\draw (p)--++(-0.4,0);
\draw (r)--++(0.4,0);
\draw (r)--++(-0.4,0);
\draw (q)--++(0.4,0);
\draw (q)--++(-0.4,0);
\draw (q2)--++(0,-0.4);
\draw (q0)--++(0,-0.4);
\draw (q1)--++(0,-0.4);
\draw (q3)--++(0,-0.4);
\draw (q)--++(0,0.4)--++(-0.4,0)--++(0,0.7)--++(7,0)--++(0,-0.7)--++(-6.6,0);

\foreach \k in {-0.4,-0.15,0.1, 1.75,2,2.25, 3.6,3.75,3.9}
{
\draw[fill=black] (\k,0.3) circle (0.02);
}

\foreach \k in {-2.35,-2.25,-2.15, -0.85,-0.75,-0.65}
{
\draw[fill=black] (\k,-1) circle (0.02);
}

\node (b0) at (0.25,4.5) {};
\node (c0) at (3.75,4.5) {};

\node[draw=none,fill=none,label=above:\footnotesize{$x_1^i$}] (e) at (0.25,4.5) {};

\node[draw=none,fill=none,label=above:\footnotesize{$x_p^i$}] (e) at (3.75,4.5) {};

\draw[fill=black] (1.75,4.5) circle (0.02);
\draw[fill=black] (2,4.5) circle (0.02);
\draw[fill=black] (2.25,4.5) circle (0.02);
\draw[fill=black] (1.75,2.8) circle (0.02);
\draw[fill=black] (2.25,2.8) circle (0.02);
\draw[fill=black] (2,2.8) circle (0.02);
\draw (c0)--++(0,-1)--++(0.2,0)--++(0,-1.3)--++(-0.4,0)--++(0,1.3)--++(0.2,0);
\draw (b0)--++(0,-1)--++(0.2,0)--++(0,-1.3)--++(-0.4,0)--++(0,1.3)--++(0.2,0);

\draw (q0)--(0.2,2.2);
\draw (q1)--(0.3,2.2);
\draw (q2)--(3.7,2.2);
\draw (q3)--(3.8,2.2);
\node[draw=none,fill=none,label=left:\footnotesize{$A_{f(x_1^i)}$}] (af1) at (0.25,2.75) {};
\node[draw=none,fill=none,label=right:\footnotesize{$A_{f(x_p^i)}$}] (afp) at (3.75,2.75) {};
\end{tikzpicture}
}
\caption{The switching gadget $T$ and the group of vertices $x_s^i$ for $s \in [p]$ before the step of splitting $q$-vertices (left) and after introducing the path $Q$ in the case that $\{\alpha,\beta,\gamma\}$ is not strongly incomparable (right).}\label{fig:spliting-q3}
\end{figure}
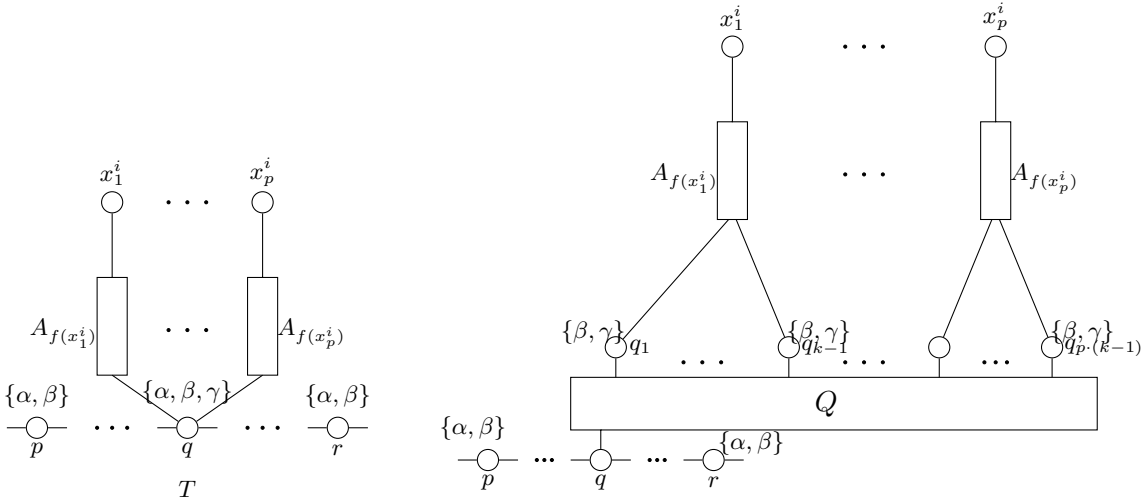

\smallskip
\textbf{Case II: The set $\{\alpha, \beta,\gamma\}$ is not strongly incomparable.} By \cref{obs:walks-between-corners} this means that $H$ contains an induced $C_6$ with consecutive vertices $w_1,\ldots,w_6$ and $\alpha=w_1$, $\beta=w_5$, $\gamma=w_3$, or an induced $C_8$ with consecutive vertices $w_1,\ldots,w_8$ and $\alpha=w_1$, $\beta=w_5$, $\gamma=w_3$.

In this case we leave each $q$-vertex $q$ in the graph, but we introduce a path $Q$ with $H$-lists $L$, with $q$ as one of endvertices, $p \cdot (k-1)$ special vertices $q_{j}$ for $j \in [p\cdot (k-1)]$, with list $L(q_{j})=\{\beta,\gamma\}$ and such that:
\begin{itemize}
\item for every list homomorphism $\vphi: (Q,L) \to H$, if $q$ is mapped to $\gamma$, then for every $j \in [p\cdot (k-1)]$ the vertex $q_{j}$ is mapped to $\gamma$.
\item there exists a list homomorphism $\vphi: (Q,L) \to H$ such that $\vphi(q)=\gamma$ and $\vphi(q_j)=\gamma$ for every $j \in [p\cdot (k-1)]$.
\item for every $c \in \{\alpha,\beta\}$ there exists a list homomorphism $\vphi: (Q,L) \to H$ such that $q$ is mapped to $c$ and for every $j \in [p\cdot (k-1)]$ the vertex $q_{j}$ is mapped to $\beta$,
\item for every distinct $j,s \in [p\cdot(k-1)]$, the distance between the vertices $q_{j}$ and $q_{s}$ is at least $g$. Moreover, the distance between $q_{j}$ and $q$ is at least $g$.
\end{itemize}
The construction of the path $Q$ is shown on \cref{fig:spliting-q2}. Again, for each neighbor $a_j$ of $q$ (the neighbors of $q_C^{i,f}$ are ordered as in the previous case) we add an edge $q_{j} a_j$ and remove the edge $q a_j$ (see \cref{fig:spliting-q3}). 

\smallskip
This completes the Step 1. In both cases we will refer to the newly introduced vertices $q_j$ as $q$-vertices.

\paragraph{Step 2. Splitting $x$-vertices.}  The only vertices that might still have degree larger than $3$ are vertices from  $\{x^i_j \ | \ i \in [t],j\in[p] \}$. More precisely, the degree of the $x$-vertex in an assignment gadget is $(k-1)^2$, and thus the degree of an $x$-vertex $x$ is $d=d(x)\cdot (k-1)^2$, where $d(x)$ is the number of the assignment gadgets, whose $x$-vertex is $x$. We split the vertex $x$ into $d$ vertices $x_1,\ldots,x_d$, each with list $S$. We remove $x$ from the graph. For every $s \in [d-1]$ we introduce a path $X_s$ of even length at least $g$, lists of consecutive vertices $S, S',\ldots, S$, and we identify its endvertices with $x_s$ and $x_{s+1}$, respectively.
We fix an ordering $b_1,\ldots,b_d$ of neighbors of $x$, such that if $b_i$ and $b_j$ belong, respectively, to assignment gadgets $A_i$ and $A_j$, and $A_i$ precedes $A_j$ in the fixed order of the assignment gadgets, then $b_i$ precedes $b_j$. The order of the neighbors from the same assignment gadget is arbitrary. For every $s \in [d]$ we add an edge $b_s x_s$. We will refer to the new vertices $x_j$ introduced in this step  also as $x$-vertices. This completes the construction of $(\tG,\tL)$.

Now let us verify that $(\widetilde{G},\widetilde{L})$ satisfies desired properties. First, let us show the property (1.).

\begin{claim}\label{claim:equiv-ctwproof}
$(\widetilde{G},\widetilde{L})\to H$ if and only if $\phi$ is satisfiable.
\end{claim}
\begin{claimproof}
Recall that $(G,L) \to H$ if and only if $\phi$ is satisfiable. Thus it is sufficient to show that $(\tG,\tL) \to H$ if and only if $(G,L) \to H$. So suppose first that there exists a list homomorphism $\vphi: (G,L) \to H$. We consider two cases, depending on which case in Step 1. was applied.

If the first case in Step 1. was applied, we define $\tphi: (\tG,\tL) \to H$ as follows.
For every vertex $v$ of $G$ that is not a $q$-vertex or an $x$-vertex (note that these vertices are also vertices of $\tG$),
we set $\tphi(v)=\vphi(v)$.
For every $x$-vertex $x_j$ that was introduced in Step 2. for some $x$-vertex $x$ in $G$, we set $\tphi(x_j):=\vphi(x)$ and we extend $\tphi$ on the path $X_j$: we map odd vertices from $X_j$ to $\vphi(x)$ and even vertices to the private neighbor of $\vphi(x)$ from $S'$.
Similarly, we extend $\tphi$ on the new $q$-vertices: for every $q$-vertex $q_j$ that was introduced for a $q$-vertex $q$ in $G$, we set $\tphi(q_{j}):=\vphi(q)$ and we extend $\tphi$ to the other vertices of $Q_j$. This completes the definition of $\tphi$ in this case.

If the second case in Step 1. was applied, we define $\tphi(v):=\vphi(v)$ for every vertex $v$ of $G$ that is a vertex of a path $P_C$, again, these vertices are also vertices of $\tG$.
Then for every $q$-vertex $q$, if $\vphi(q)\in\{\alpha,\beta\}$, we extend $\tphi$ to the vertices of the path $Q$ introduced for $q$ in Step 1., so that every $q$-vertex $q_j$ on $Q$ is mapped to $\beta$.
Otherwise, if $q$ is mapped to $\gamma$, then we extend $\tphi$ to $Q$ so that every $q$-vertex $q_j$ is mapped to $\gamma$.
Then we define $\tphi(x_j):=\vphi(x)$ for every vertex $x_j$ that was introduced for an $x$-vertex $x$ in Step 2. and we extend $\tphi$ to the vertices of the paths $X_j$.
Finally, we extend $\tphi$ to the remaining vertices of the assignment gadgets. Note that this is possible by the way of defining $\tphi$ on the $x$-vertices and $q$-vertices, and by the properties (A2.) and (A3.) from \cref{def:assign-gadget}. This completes the definition of $\tphi$. It is straightforward to verify that in both cases $\tphi$ is a list homomorphism from $(\tG,\tL)$ to $H$.

Suppose now that there exists a list homomorphism $\tphi: (\tG,\tL) \to H$. Again we consider two cases.

If the first case in Step 1. was applied, we construct a list homomorphism $\vphi: (G,L) \to H$ as follows. For every $v$ from $G$ that is not a $q$-vertex, or an $x$-vertex, we set $\vphi(v):=\tphi(v)$. Observe now that for every $x$-vertex $x$ from the graph $G$, all the new $x$-vertices $x_j$ that were introduced for $x$ must be mapped by $\tphi$ to the same vertex of $H$.
That follows from the construction of the paths $X_j$, as lists of their vertices are $S$ and $S'$, and for every $s \in S$ there is exactly one $s'\in S'$ adjacent to $s$.
Similarly, the new $q$-vertices $q_j$ that were introduced for the same $q$-vertex $q$ must be mapped by $\tphi$ to the same vertex of $H$. Thus for every $x$-vertex $x$ in $G$ we set $\vphi(x):=\tphi(x_1)$, where $x_1$ is the first $x$-vertex introduced for $x$ in the Step 2., and for every $q$-vertex $q$ we set $\vphi(q):=\tphi(q_1)$, where $q_1$ is the $q$-vertex introduced for $q$ in the Step 1.
This completes the definition of $\vphi$ in this case.

If the second case in Step 1. was applied, we define $\vphi(v):=\tphi(v)$ for every vertex $v$ of a path $P_C$.
Then for every $x$-vertex $x$ from the graph $G$ we set $\vphi(x):=\tphi(x_1)$, where $x_1$ is the first $x$-vertex introduced for $x$ in Step 2. Finally, we extend $\vphi$ to all remaining vertices of the assignment gadgets. Note that in this case all $q$-vertices $q_j$ introduced for the same $q$-vertex $q$ are mapped by $\tphi$ to the same vertex. Moreover, if $\tphi(q)=\gamma$, then $\tphi(q_j)=\gamma$ for every $q_j$, and if $\tphi(q)\in \{\alpha,\beta\}$, then $\tphi(q_j)=\beta$ for every $q_j$, and $\vphi(x)=\tphi(x_j)$ for every $x_j$ introduced for $x$. Therefore, extending $\vphi$ to the remaining vertices of the assignment gadgets is possible. It is straightforward to verify that in both cases $\vphi$ is a list homomorphism from $(G,L)$ to $H$.
\end{claimproof}

Now let us verify the property (2.).

\begin{claim}\label{claim:ctw}
The cutwidth of $\tG$ is at most $t\cdot p + f(g,H)$, where $f$ is some function of $g$ and $H$.
\end{claim}
\begin{claimproof}
We construct a linear layout of $\tG$ as follows. First we order the original paths $P_C$ (those from graph $G$) according to the fixed order of clauses. Then we order vertices of those paths in a way that vertices from $P_{C_j}$ precede vertices of $P_{C_{j+1}}$, and vertices on each path $P_C$ are ordered in a natural way (the vertex $x_C$ is the first one and the vertex $y_C$ is the last one). Then, depending on the case applied in Step 1., we either replace each $q$-vertex $q$ with vertices $q_j$ and vertices of paths $Q_j$ in the following order: $Q_0,q_1,Q_1,\ldots, q_{p\cdot (k-1)},Q_{p\cdot (k-1)}$ (if Case 1. was applied), or we insert the vertices from the path $Q$ just after $q$, in the natural order with $q$ being the first one.

It only remains to place the new $x$-vertices (i.e., vertices introduced in Step 2.), and the vertices from the assignment gadgets. Consider an arbitrary assignment gadget, whose $x$-vertex before Step 2. was $x$ and the $y$-vertex before Step 1. was $q$.
Observe that in $\tG$ the vertices adjacent to the vertices of the gadget are $x_\ell,\ldots,x_{\ell'}$ and $q_s,\ldots,q_{s'}$, for some $\ell$ and $s$, and where $\ell'=\ell+(k-1)^2-1$ and $s'=s+k-2$. The vertices $q_s,\ldots,q_{s'}$ are consecutive $q$-vertices introduced for $q$ in Step 1. and the vertices $x_\ell,\ldots,x_{\ell'}$ are consecutive $x$-vertices introduced for $x$ in Step 2. We insert the vertices from the gadget, $x$-vertices $x_\ell,\ldots,x_{\ell'}$ and the paths $X_\ell,\ldots,X_{\ell'}$ (also introduced in Step 2. for $x$), just before the group of $q$-vertices $q_s,\ldots,q_{s'}$. The order of the vertices from the gadget is abitrary. The $x$-vertices are ordered $x_\ell,\ldots,x_{\ell'}$ and the vertices from the paths $X_\ell,\ldots,X_{\ell'}$ are in the order of appearing on the path, with $x_j$ being the first one from $X_j$. This completes the construction of the linear layout of $\tG$.

Consider an arbitrary cut of the layout. The edges that can possibly cross this cut are:
\begin{itemize}
\item at most one edge from a path $P_C$,
\item at most one edge from a path $Q$,
\item edges from at most one assignment gadget (together with $(k-1)^2$ edges from the new $x$-vertices and $k-1$ edges from the new $q$-vertices to that gadget)
\item for each $i \in [t], j \in [p]$ at most one edge from the path $X_\ell$ that was introduced for the $x$-vertex $x^i_j$ . 
\end{itemize}
Thus for each cut the number of edges crossing this cut is at most a constant depending on $g$ and $H$, let us denote it by $f(g,H)$, and at most $t \cdot p$ edges from the paths $X_\ell$.
\end{claimproof}

Now let us verify that $\tG \in \Cg$. Note that the fact that $\tG$ is bipartite follows directly from the construction. It remains to show the remaining conditions of the class $\Cg$.

\begin{claim}\label{claim:degree}
The maximum degree of $\tG$ is $3$.
\end{claim}
\begin{claimproof}
Observe that the only vertices of $G$ that might have degree larger than $3$ are $q$-vertices and vertices $x^i_j$ for $i \in [t],j\in[p]$ -- that follows from the definitions of the gadgets and the construction of $G$. In the first step we reduce the degrees of $q$-vertices by splitting them into vertices of degree at most $3$. In the second step we repeat this for vertices $x^i_j$ for $i \in [t],j\in[p]$. Thus the maximum degree of $\tG$ is at most $3$.
\end{claimproof}

\begin{claim}\label{claim:dist}
Vertices of degree $3$ in $\tG$ are pairwise at distance at least $g$.
\end{claim}
\begin{claimproof}
First observe that in all introduced gadgets the distance between the vertices of degree at least $3$ is at least $g$, and the applied steps did not change it. Moreover, the $q$-vertices in $\tG$ are at distance at least $g$, and the $x$-vertices are at distance at least $g$. Finally, internal vertices of degree $3$ from distinct gadgets are at distance at least $g$, since they are at distance at least $g$ from the special vertices of those gadgets and only special vertices of the gadgets are adjacent to any vertices outside the gadgets.
\end{claimproof}

\begin{claim}\label{claim:girth}
The girth of $\tG$ is at least $g$.
\end{claim}
\begin{claimproof}
Recall that  a switching gadget is a path, so it cannot contain a cycle.
Furthermore, every assignment gadget that we introduced in the construction of $\tG$ has girth at least $g$.
Every other cycle in $\tG$, which is not fully contained in an assignment gadget, must contain at least two distinct $x$-vertices, or at least two distinct $q$-vertices, or at least one $q$-vertex and at least one $x$-vertex.
All $x$-vertices and $q$-vertices are pairwise at distance at least $g$, so the claim follows. 
\end{claimproof}

By previous observations and \cref{claim:degree}, \cref{claim:dist}, and \cref{claim:girth}, the property (3.) is satisfied. The property (4.) follows directly from the construction of $(\tG,\tL)$. Obviously, $(\tG,\tL)$ can be constructed in polynomial time.
\end{proof}

Now we will use the construction introduced in \cref{lem:instance-lhom-ctw} to prove \cref{thm:ctw-bipartite-undecomp}. The proof is split into two parts -- first we will prove the statement a), and then the statement b).

\begin{proof}[Proof of \cref{thm:ctw-bipartite-undecomp}~a)]
Let $\eps>0$, let $H$ be an undecomposable, bipartite graph, whose complement is not a circular-arc graph, and let $g \in \N$. Suppose now that $\lhomo H$ can be solved in time $(mim(H)-\eps)^{\ctw{G}} \cdot |V(G)|^{\Oh(1)}$ for every instance $(G,L)$ such that $G\in \Cg$. Let $\phi$ be an instance of \textsc{CNF-Sat} with $n$ variables and $m$ clauses. Let $k:=mim(H)$, observe that since $H$ is not a complement of a circular-arc graph, by \cref{obs:walks-between-corners}~1., it holds that $k\geq 2$. Let $\lambda:=\log_k(k-\eps)$, note that $\lambda<1$, and let $p$ be sufficently large so that $\lambda \cdot \frac{p}{p-1}< 1$. Let $t:=\Big{\lceil} \frac{n}{\lfloor p \cdot \log k \rfloor} \Big{\rceil}$.

We call \cref{lem:instance-lhom-ctw} for $H$, $g$, $\phi$, and $p$ to construct in polynomial time the instance $(\tG,\tL)$ of $\lhomo H$. By \cref{lem:instance-lhom-ctw} (1.) solving the instance $(\tG,\tL)$ of $\lhomo H$ is equivalent to solving the instance $\phi$ of \textsc{CNF-Sat}. Moreover, by \cref{lem:instance-lhom-ctw} (4.) it holds that $|V(\tG)|=(n+m)^{\Oh(1)}$. Thus we can solve the instance $\phi$ in time:
\begin{align*}
& (k-\eps)^{\ctw{\tG}} \cdot |V(\tG)|^{\Oh(1)}=(k-\eps)^{\ctw{\tG}} \cdot (n+m)^{\Oh(1)} \leq (k-\eps)^{t\cdot p + f(g,H)} \cdot (n+m)^{\Oh(1)} = \\
& (k-\eps)^{t \cdot p} \cdot (k-\eps)^{f(g,H)} \cdot (n+m)^{\Oh(1)} =(k-\eps)^{t \cdot p} \cdot (n+m)^{\Oh(1)},
\end{align*}

where the first inequality follows from \cref{lem:instance-lhom-ctw} (2.) and the last equality follows from the fact that $g$ and $|H|$ are constant.
Similarly to equations \cref{eq:reduction1}, \cref{eq:reduction2}, \cref{eq:reduction3}, and \cref{eq:reduction4} we can deduce that the above implies that \textsc{CNF-Sat} can be solved in time $(2-\delta)^n\cdot (n+m)^{\Oh(1)}$ for some $\delta>0$, which contradicts the SETH.
\end{proof}

\begin{proof}[Proof of \cref{thm:ctw-bipartite-undecomp}~b)]
Assume the ETH and let $0<\delta'<1$ be such that 3-\textsc{Sat} cannot be solved in time $2^{\delta'\cdot n}\cdot n^{\Oh(1)}$ for every instance $\phi$ with $n$ variables and $m$ clauses.
Define $\delta:=\frac{\delta'}{2}$. Suppose that there is an undecomposable, bipartite graph $H$, whose complement is not a circular-arc graph, and an algorithm $A$ that solves $\lhomo{H}$ for every instance $(G,L)$ with $G \in \Cg$ in time $mim(H)^{\delta\cdot\ctw{G}}\cdot |V(G)|^{\Oh(1)}$. Let $\phi$ be an instance of 3-\textsc{Sat} with $n$ variables and $m$ clauses.
Let $k:=mim(H)$, again note that $k\geq 2$. Define $\eps'>0$ such that $\eps'<k^{\delta'}-k^{\delta}$ and $\eps:=k-k^{\delta'}+\eps'$. Observe that $\lambda:=\log_k(k-\eps)<\delta'$. We choose $p$ sufficently large so that $\lambda\cdot\frac{p}{p-1}\leq \delta'$ and define $t:=\Big{\lceil} \frac{n}{\lfloor p \cdot \log k \rfloor} \Big{\rceil}$.
As $3$-\textsc{Sat} is a special case of \textsc{CNF-Sat}, we can call \cref{lem:instance-lhom-ctw} for $H$, $g$, $p$, and $\phi$ to obtain an instance $(\tG,\tL)$ of $\lhomo{H}$.
Moreover, for any instance of $3$-\textsc{Sat} with $n$ variables and $m$ clauses, we can assume that $m=n^{\Oh(1)}$. By \cref{lem:instance-lhom-ctw}, solving the instance $(\tG,\tL)$ is equivalent to solving the instance $\phi$ of $3$-\textsc{Sat}. Thus $\phi$ can be solved in time: 
\[
k^{\delta\cdot\ctw{\tG}}\cdot |V(\tG)|^{\Oh(1)}=k^{\delta\cdot\ctw{\tG}}\cdot n^{\Oh(1)}<(k^{\delta'}-\eps')^{\ctw{\tG}}\cdot n^{\Oh(1)}=(k-\eps)^{\ctw{\tG}}\cdot n^{\Oh(1)}. 
\]

By \cref{lem:instance-lhom-ctw}~(2.) the running time is at most $(k-\eps)^{t\cdot p}\cdot n^{\Oh(1)}$. We can provide similar computations as in equations \cref{eq:reduction1}, \cref{eq:reduction2} and \cref{eq:reduction3}; recall that $\lambda\cdot\frac{p}{p-1}\leq \delta'$ so \cref{eq:reduction3} applies. So we can solve any instance of $3$-\textsc{Sat} with $n$ variables in time:

\[
k^{\delta' \cdot \frac{n}{\log k} + \lambda \cdot p} \cdot n^{\Oh(1)}=k^{\delta' \cdot \frac{n}{\log k}} \cdot k^{\lambda \cdot p} \cdot n^{\Oh(1)}=k^{\delta' \cdot \frac{n}{\log k}} \cdot n^{\Oh(1)}= 2^{\delta' \cdot n} \cdot n^{\Oh(1)},
\]

which, by the choice of $\delta'$, contradicts the ETH.
\end{proof}

Observe that the instance $(\tG,\tL)$ constructed in \cref{lem:instance-lhom-ctw} satisfies conditions 1. and 2. of \cref{def:consist-inst}. Similarly as in \cref{cor:consist-hard}, we conclude the following.

\begin{remark}\label{cor:ctw-consist-hard}
\cref{thm:ctw-bipartite} and \cref{thm:ctw-bipartite-undecomp} hold, even if we assume that the instance $(G,L)$ is consistent.
\end{remark}

\subsection{Lower bounds for \lhomo{H}, general target graphs}
Now, similarly as in the proof of \cref{thm:main}~b), we will show that \cref{thm:ctw-bipartite} implies \cref{thm:main-ctw-list-hard}.

\mainctwlhomo*

\begin{proof}
First, assume the SETH and suppose that \cref{thm:main-ctw-list-hard}~a) fails.
Then there exists a connected, non-bi-arc graph $H$, a constant $\eps>0$, $g\in \N$, and an algorithm $A$ that solves $\lhomo H$ for every instance $(G,L)$, such that $G\in \Cg$, in time $(mim^*(H)-\eps)^{\ctw{G}} \cdot |V(G)|^{\Oh(1)}$. 
We can assume that $H$ is non-bipartite, as otherwise the result follows immediately from \cref{thm:ctw-bipartite}.
Consider a consistent instance $(G,L)$ of $\lhomo {H^*}$ such that $G\in \Cg$. Let $(G,L')$ be an instance of $\lhomo H$ obtained as in \cref{prop:bipartite-associted}, clearly it can be constructed in polynomial time. Recall that if $H$ is a non-bipartite, connected, non-bi-arc graph, then $H^*$ is connected, bipartite graph, whose complement is not a circular-arc graph and thus $H^*$ satifies the assumptions of \cref{thm:ctw-bipartite}. Moreover, $mim^*(H)=mim^*(H^*)$.

We can use $A$ to solve the instance $(G,L')$ of $\lhomo H$ in time $(mim^*(H)-\eps)^{\ctw{G}} \cdot |V(G)|^{\Oh(1)}$, which by \cref{prop:bipartite-associted} is equivalent to solving the instance $(G,L)$ of $\lhomo {H^*}$ in time $(mim^*(H^*)-\eps)^{\ctw{G}} \cdot |V(G)|^{\Oh(1)}$, which, by \cref{thm:ctw-bipartite} and \cref{cor:ctw-consist-hard}, contradicts the SETH. That completes the proof of the statement~a).
\medskip

Now assume the ETH and let $\delta$ be the constant from \cref{thm:ctw-bipartite}~b). Suppose there exists a connected non-bi-arc graph $H$ and an algorithm $A$ that solves $\lhomo{H}$ for any instance $(G,L)$, such that $G \in \Cg$, in time $(mim^*(H))^{\delta\cdot \ctw{G}}\cdot |V(G)|^{\Oh(1)}$. Again, we can assume that $H$ is non-bipartite.
 Let $(G,L)$ be a consistent instance of $\lhomo{H^*}$ such that $G\in \Cg$ and let $(G,L')$ be the instance of $\lhomo{H}$ given by  \cref{prop:bipartite-associted}. As $mim^*(H)=mim^*(H^*)$, we can use the algorithm $A$ to solve the instance $(G,L')$ of $\lhomo{H}$, and thus the instance $(G,L)$ of $\lhomo{H^*}$ in time $mim^*(H^*)^{\delta\cdot \ctw{G}}\cdot |V(G)|^{\Oh(1)}$.
By \cref{thm:ctw-bipartite} and \cref{cor:ctw-consist-hard}, this contradicts the ETH. That completes the proof of the statement~b).
\end{proof}

Finally, let us point that that in the proofs of \cref{thm:ctw-bipartite-undecomp} and \cref{thm:main-ctw-list-hard}, we did not claim that the linear layout we constructed is an optimal one, we only cared that its width is upper-bounded by a correct value.
Thus the proof actually yields the following, slightly stronger statement.

\begin{theorem}\label{thm:instance-with-layout}
Let $\mathcal{H}$ be the class of connected non-bi-arc graphs.
For $g \in \N$, let $\mathcal{C}_g$ be the class of subcubic bipartite graphs $G$ with girth at least $g$, such that vertices of degree $3$ in $G$ are at distance at least $g$.
\begin{enumerate}[a)]
\item For every $H \in \mathcal{H}$, there is no algorithm that solves every instance $(G,L)$ of  $\lhomo H$, where $G \in \mathcal{C}_g$, given with a linear layout of width $w$, in time $(mim^*(H)-\epsilon)^{w} \cdot |V(G)|^{\Oh(1)}$ for any $\epsilon > 0$, unless the SETH fails.
\item There exists a constant $0<\delta<1$, such that for every $H \in \mathcal{H}$, there is no algorithm that solves every instance $(G,L)$ of $\lhomo H$, where $G \in \mathcal{C}_g$, given with a linear layout of width $w$, in time $mim^*(H)^{\delta \cdot w} \cdot |V(G)|^{\Oh(1)}$, unless the ETH fails.
\end{enumerate}
\end{theorem}
\subsection{Lower bounds for \homo{H}} \label{sec:nonlist}
In this section we extend \cref{thm:main-ctw-list-hard} to the non-list case, i.e., we prove \cref{thm:main-ctw-nonlist-hard}.

\mainctwhomo*

First, let us define the graph class mentioned in the statement.
Recall that by the result of Hell and Ne\v{s}et\v{r}il~\cite{DBLP:journals/jct/HellN90}, the \homo{H} problem is NP-hard if $H$ nonbipartite and has no loops. In particular, this implies that $H$ has at least three vertices.
We say that a graph $H$ is a \emph{core} if every homomorphism $\vphi: H \to H$ is an automorphism, i.e., is injective and surjective. It is well-known that in order to understand the complexity of \homo{H}, it is sufficient to consider the case if $H$ is a core~\cite{DBLP:journals/dm/HellN92,hell2004graphs}.
The following observation is straightforward.
\begin{observation}\label{obs:core-incomp}
Let $H$ be a core. Then every two distinct vertices of $H$ are incomparable.
\end{observation}

The class considered in \cref{thm:main-ctw-nonlist-hard} are the so-called \emph{projective graphs}.
The definition of these graphs is technical, so we will skip it (see e.g. Hell, Ne\v{s}et\v{r}il~\cite[Section 2.7]{hell2004graphs}).
The following equivalent characterization is much more relevant to our paper, see also \cite{OkrasaSODA}.

\begin{theorem}[Larose, Tardif~\cite{larose2001strongly}]\label{lem:list-gadget}
Let $H$ be graph with at least three vertices. The following are equivalent:
\begin{enumerate}
\item $H$ is projective.
\item For every  $L\subseteq V(H)$ there exist a tuple $(x_1,\ldots,x_\ell)$ of vertices in $H$ and a graph $F_L$ with a tuple of its vertices $(y_0,y_1,\ldots,y_{\ell})$ such that $L = \{ \vphi(y_0)  \ | \ \vphi : F_L \to H, \text{such that}\ \vphi(y_1)=x_1, \ldots , \vphi(y_{\ell})=x_\ell \}$.
\end{enumerate}
\end{theorem}

Now we are ready to prove \cref{thm:main-ctw-nonlist-hard}.

\begin{proof}[Proof of \cref{thm:main-ctw-nonlist-hard}]
Let $H$ be a non-bipartite projective core without loops. Note that $H$ contains an induced odd cycle $C_s$ for $s\geq 3$ and thus $H^*$ contains an induced cycle of length $2s\geq 6$, so $H^*$ is not a complement of a circular-arc graph, and therefore $H$ is not a bi-arc graph.

We will reduce from \lhomo{H}, whose hardness was proven in \cref{thm:main-ctw-list-hard}. Actually, we will use the stronger version stated in \cref{thm:instance-with-layout}.
Let $(G,L)$ be an instance of $\lhomo H$, and let $\pi=(v_1,\ldots,v_{|V(G)|})$ be a linear layout of $G$ of width $w$.
Consider an instance $\tG$ of $\homo H$ constructed as follows.
For every $v_i \in V(G)$ we call \cref{lem:list-gadget} to obtain the tuple $(x_1^{(i)},\ldots,x_{\ell_i}^{(i)})$ of vertices in $H$ and a graph $F_{L(v_i)}$ with special vertices $y_0^{(i)},\ldots,y_{\ell_i}^{(i)}$.
For every $v_i$ we introduce a copy $H^{(i)}$ of the graph $H$ and identify vertices $y_1^{(i)},\ldots, y_{\ell_i}^{(i)}$, respectively with $x_1^{(i)},\ldots,x_{\ell_i}^{(i)}$ in the copy $H^{(i)}$.
Moreover, we identify $y_0^{(i)}$ with $v_i$.
Finally, for every $i \in [|V(G)|-1]$ we add edges between the copies $H^{(i)}$ and $H^{(i+1)}$ as follows.
For every vertex $z^{(i)}$ in $H^{(i)}$ and its corresponding copy $z^{(i+1)}$ in $H^{(i+1)}$ we add all edges between $z^{(i)}$ and $N_{H^{(i+1)}}(z^{(i+1)})$. This completes the construction of $\tG$.

Let us show the following properties of $\tG$:
\begin{enumerate}[(1)]
\item $(G,L) \to H$ if and only if $\tG \to H$,
\item $\ctw{\tG} \leq w + g(H)$, where $g$ is some function of $H$,
\item $|V(\tG)|=|V(G)|^{\Oh(1)}$,
\item $\tG$ can be constructed in polynomial time.
\end{enumerate}

To see property (1), consider a list homomorphism $\vphi: (G,L) \to H$. We define $\tphi: \tG \to H$ for every $v_i \in V(G)$ as $\tphi(v_i):=\vphi(v_i)$. Since $\vphi$ is a list homomorphism, for every $v_i \in V(G)$ it holds $\vphi(v_i)\in L(v_i)$. Observe that by \cref{lem:list-gadget} we can extend $\tphi$ to the rest vertices of $F_{L(v_i)}$ so that $\tphi(y_1^{(i)})=x_1^{(i)}, \ldots, \tphi(y_{\ell_i}^{(i)})=x_{\ell_i}^{(i)}$. Then we can extend $\tphi$ to every copy $H^{(i)}$ as the identity function. Observe that the edges between distinct copies of $H$ are preserved by $\tphi$ and thus $\tphi$ is a homomorphism from $\tG$ to $H$. 

So suppose now that there exists a homomorphism $\tphi: \tG \to H$. We define $\vphi:=\tphi|_{V(G)}$. Recall that since $H$ is a core, the mapping $\tphi$ restricted to each copy of $H$ is an automorphism.
By \cref{lem:list-gadget}, in order to show that $\vphi$ respects lists $L$, it is sufficent to show that every copy of $H$ is colored by $\tphi$ according to the same automorphism.
Consider copies $H^{(i)}$ and $H^{(i+1)}$ for some $i \in [|V(G)|-1]$.
Let $z^{(i)}$ be an arbitrary vertex of  $H^{(i)}$ and let $z^{(i+1)}$ be its corresponding vertex in $H^{(i+1)}$,
and suppose that for $s:= \tphi(z^{(i)})$ and $u:= \tphi(z^{(i+1)})$ we have $s \neq u$. 
Since $H$ is a core, $\tphi$ is an automorphism on $H^{(i+1)}$ and thus the image of $N_{H^{(i+1)}}(z^{(i+1)})$ is precisely the set $N_H(u)$.
Moreover, as $z^{(i)}$ is adjacent to every vertex in $N_{H^{(i+1)}}(z^{(i+1)})$, we observe that all vertices of the image of $N_{H^{(i+1)}}(z^{(i+1)})$ must be adjacent to $s$ in $H$.
This means $N_H(u) \subseteq N_H(s)$, which, by \cref{obs:core-incomp}, contradicts that $H$ is a core.

Now let us show the property (2). We modify the linear layout $\pi$ of $G$ to obtain a linear layout $\tpi$ of $\tG$ as follows.
For every $v_i \in V(G)$ we insert the vertices of $F_{L(v_i)}$ and the copy $H^{(i)}$ just after $v$ (the order of these vertices is arbitrary). Consider an arbitrary cut of $\tpi$. The edges crossing the cut might be:
\begin{itemize}
\item at most $w$ edges from $G$,
\item edges from at most one gadget $F_{L(v)}$ and at most two copies of $H$ (including edges between copies and between the gadget and a copy). 
\end{itemize}
Recall that the size of $F_{L(v)}$ for $v \in V(G)$ depends only on $H$ and there might be at most $2^{|H|}$ different lists $L(v)$.
Thus there is some function of $H$ bounding the size of each gadget $F_{L(v)}$.
Thus we conclude that $\ctw{\tG} \leq w + g(H)$, where $g$ is some function of $H$.

Properties (3) and (4) follow directly from the construction of $\tG$.

Now suppose that $\homo H$ can be solved for every instance $G'$ in time $(mim^*(H)-\eps)^{\ctw{G'}}\cdot |V(G')|^{\Oh(1)}$ for some $\eps>0$.
Then, for an instance $(G,L)$ of $\lhomo H$ with a linear layout $\pi$ of width $w$, we can construct in polynomial time the instance $\tG$ of $\homo H$ as above. We solve the instance $\tG$ in time:
\begin{align*}
& (mim^*(H)-\eps)^{\ctw{\tG}}\cdot |V(\tG)|^{\Oh(1)} \leq (mim^*(H)-\eps)^{w+g(H)}\cdot |V(G)|^{\Oh(1)}= \\
& (mim^*(H)-\eps)^{w}\cdot (mim^*(H)-\eps)^{g(H)}\cdot |V(G)|^{\Oh(1)}=(mim^*(H)-\eps)^{w}\cdot |V(G)|^{\Oh(1)},
\end{align*}
which is equivalent to solving the instance $(G,L)$ with a linear layout $\pi$ of width $w$ in time $(mim^*(H)-\eps)^{w}\cdot |V(G)|^{\Oh(1)}$. By \cref{thm:instance-with-layout}~a) it contradicts the SETH, so the statement~a) follows.
\medskip

Now let $\delta$ be the constant from \cref{thm:main-ctw-list-hard} (note that $\delta$ does not depend on $H$) and suppose that $\homo H$ can be solved for every instance $G'$ in time $(mim^*(H))^{\delta \cdot \ctw{G'}}\cdot |V(G')|^{\Oh(1)}$.
Again, in order to solve an instance $(G,L)$ of $\lhomo H$ with a linear layout $\pi$ of width $w$, we can solve the constructed instance $\tG$ of $\homo H$ in time:
\begin{align*}
& (mim^*(H))^{\delta \cdot \ctw{\tG}}\cdot |V(\tG)|^{\Oh(1)} \leq (mim^*(H))^{\delta \cdot (w+g(H))}\cdot |V(G)|^{\Oh(1)}= \\
& (mim^*(H))^{\delta\cdot w}\cdot (mim^*(H))^{\delta\cdot g(H)}\cdot |V(G)|^{\Oh(1)}=(mim^*(H))^{\delta\cdot w}\cdot |V(G)|^{\Oh(1)},
\end{align*}

which is equivalent to solving the instance $(G,L)$ with $\pi$ in time $(mim^*(H))^{\delta\cdot w}\cdot |V(G)|^{\Oh(1)}$. By \cref{thm:instance-with-layout}~b) it contradicts the ETH, so the statement~b) follows.
\end{proof}

%

\newpage
\section{Construction of the gadgets}\label{sec:constr-of-gadgets}
In this section we show how to construct the gadgets that we used in our hardness reductions. First, let us show how we use walks for building gadgets. For a set $\DD = \{\cD_i\}_{i=1}^k$ of walks of equal length $\ell \geq 1$, let $\PP(\DD)$ be the path with $\ell+1$ vertices $p_1,\ldots,p_{\ell+1}$, equipped with $H$-lists, such that the list of $p_i$ is the set of $i$-th vertices of walks in $\DD$.
The vertex $p_1$ will be called the \emph{input vertex} and $p_{\ell+1}$ will be called the \emph{output vertex}.

\begin{lemma}[\cite{LhomoTreewidth,LhomoTreewidthFull}]\label{lem:gadget-of-walks}
Let $\DD = \{\cD_i\}_{i=1}^k$ be a set of walks $\cD_i:s_i\to t_i$ of equal length $\ell \geq 1$.
Let $A,B$ be a partition of $\DD$ into two non-empty sets. Moreover, for $C \in \{A,B\}$, define $S(C) = \{s_i \colon \cD_i \in C\}$ and $T(C) = \{t_i \colon \cD_i \in C\}$. Suppose that $S(A) \cap S(B) = \emptyset$ and $T(A) \cap T(B) = \emptyset$, and every walk in $A$ avoids every walk in $B$. Then $\PP(\DD)$ with the input vertex $x$, the output vertex $y$, and lists $L$, has the following properties:
\begin{enumerate}[(a)]
\item $L(x) = S(A) \cup S(B)$ and $L(y) = T(A) \cup T(B)$,
\item for every $i \in [k]$ there is a list homomorphism $f_i \colon \PP(\DD) \to H$, such that $f_i(x) = s_i$ and $f_i(y)=t_i$,
\item for every list homomorphism $f \colon \PP(\DD) \to H$, if $f(x)\in S(A)$, then $f(y) \notin T(B)$.
\end{enumerate}
Furthermore, if every walk in $B$ avoids every walk in $A$, we additionally have 
\begin{enumerate}[(d)]
\item for every list homomorphism $f \colon \PP(\DD) \to H$, if $f(x)\in S(B)$, then $f(y) \notin T(A)$.
\end{enumerate}
\end{lemma}

Let $\PP_1,\PP_2$ be gadgets defined as above, such that lists of the output vertex of $\PP_1$ and the input vertex of $\PP_2$ are equal.
We define a \emph{composition} of $\PP_1$ and $\PP_2$ as the graph $\PP$ obtained by identifying the output vertex of $\PP_1$ and the input vertex of $\PP_2$. The input and output of $\PP$ are respectively the input of $\PP_1$ and the output of $\PP_2$.
\medskip

Let $H$ be a graph and let $R\subseteq V(H)^k$ be a $k$-ary relation on $V(H)$. We define an $R$-gadget as a graph $F$ with $H$-lists $L$ and $k$ special vertices $x_1,\ldots,x_k$, called \emph{interface} vertices, such that:
\[
R=\{(f(x_1),\ldots,f(x_k)) \  | \ \ f: (F,L) \to H \}.
\] 

Let $H$ be a bipartite graph, whose complement is not a circular-arc graph and let $(\alpha,\beta,\gamma)$ be the triple of vertices in $H$ given by \cref{obs:walks-between-corners}. We will consider two special $k$-ary relations: $\ork k = \{\alpha,\beta\}^k \setminus \{\alpha\}^k$ and $\nand k = \{\alpha,\beta\}^k \setminus \{\beta\}^k$. The intuition behind the names is that we think of $\alpha$ as false and of $\beta$ as true. We will use the $\ork k$- and the $\nand 2$-gadgets to construct the assignment gadget.

The existence of the $\ork k$- and $\nand 2$-gadgets was proved in \cite{LhomoTreewidth,LhomoTreewidthFull}. However, we will show a sligthly different proof to obtain certain structure of the gadgets.

\begin{lemma}[\cite{LhomoTreewidth,LhomoTreewidthFull}]\label{lem:relationgadgets}
Let $H$ be a bipartite graph, whose complement is not a circular-arc graph, let $(\alpha,\beta,\gamma)$ be defined as in \cref{obs:walks-between-corners}, and let $g \in \N$. For every $k \geq 2$ there exists an $\ork{k}$-gadget, and an $\nand{2}$-gadget such that:
\begin{enumerate}[(1.)]
\item The $\nand{2}$-gadget is a path of length at least $g$ with endvertices as interface vertices.
\item The $\ork{k}$-gadget is a tree and every interface vertex of the $\ork{k}$-gadget is a leaf.
\item The maximum degree of the $\ork k$-gadget is at most $3$.
\item For any distinct vertices $a,b$ in the $\ork k$-gadget such that $\deg(a)=\deg(b)=3$ it holds $\dist(a,b)\geq g$.
\end{enumerate}
\end{lemma}

\begin{proof}
First, note that an $\ork 2$-gadget with interface vertices $x_1,x_2$ can be obtained from an $\ork 3$-gadget with interface vertices $x_1,x_2,x_3$ by removing $\beta$ from the list of $x_3$. Moreover, observe that if we are able to construct an $\ork {k-1}$-, an $\ork 3$-, and a $\nand 2$-gadgets for some $k>3$, then we can obtain an $\ork k$-gadget. Indeed, consider an $\ork {k-1}$-gadget with interface vertices $x_1,x_2,\ldots,x_{k-2},y$, a $\nand 2$-gadget with interface vertices $y',z'$, and an $\ork 3$-gadget with interface vertices $z,x_{k-1},x_k$. We obtain the $\ork k$-gadget by identifying vertices: $y$ with $y'$ and $z$ with $z'$. The interface vertices of the constructed $\ork 3$-gadget are $x_1,\ldots,x_k$ (see \cref{fig:ork-from-ors-and-nand}).

Moreover, if the $\ork {k-1}$-gadget and the $\ork 3$-gadget are both trees and the $\nand 2$-gadget is a path, then the $\ork k$-gadget is a tree, since we only join two trees with a path (see \cref{fig:ork-from-ors-and-nand}). The interface vertices $x_1,\ldots,x_k$ of the $\ork {k}$-gadget are leaves as they were in the $\ork {k-1}$- and $\ork 3$-gadgets. The maximum degree of the $\ork{k}$ is at most $3$ since the only vertices whose degree increased by $1$ are $y,z$ which were leaves and thus in the $\ork k$-gadget their degree is $2$. Finally, if (4.) is satisfied for both, the $\ork 3$-gadget and the $\ork{k-1}$-gadget, and the $\nand 2$-gadget is a path of length at least $g$, then for any distinct vertices $a,b$ in the $\ork k$-gadget with degree $3$ it holds $\dist(a,b)\geq g$. So it is sufficent to construct a $\nand 2$- and an $\ork 3$-gadget.

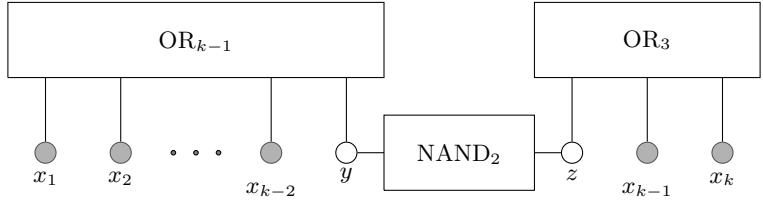
\begin{figure}
\centering{\begin{tikzpicture}[every node/.style={draw,circle,fill=white,inner sep=0pt,minimum size=8pt},every loop/.style={}]
\node[fill=gray,opacity=0.6,label=below:\footnotesize{$x_1$}] (x1) at (0,0) {};
\node[fill=gray,opacity=0.6,label=below:\footnotesize{$x_2$}] (x2) at (1,0) {};
\node[fill=gray,opacity=0.6,label=below:\footnotesize{$x_{k-2}$}] (x3) at (3,0) {};
\node[label=below:\footnotesize{$y$}] (y) at (4,0) {};
\node[label=below:\footnotesize{$z$}] (z) at (7,0) {};
\node[fill=gray,opacity=0.6,label=below:\footnotesize{$x_{k-1}$}] (x4) at (8,0) {};
\node[fill=gray,opacity=0.6,label=below:\footnotesize{$x_k$}] (x5) at (9,0) {};
\node[draw=none,fill=none,label=above:\footnotesize{$\nand 2$}] (a) at (5.5,-0.75) {};
\node[draw=none,fill=none,label=above:\footnotesize{$\ork {k-1}$}] (a) at (2,0.8) {};
\node[draw=none,fill=none,label=above:\footnotesize{$\ork 3$}] (a) at (8,1) {};
\draw (x1)++(-0.5,1)--++(5,0)--++(0,1)--++(-5,0)--++(0,-1);
\draw (z)++(-0.5,1)--++(3,0)--++(0,1)--++(-3,0)--++(0,-1);
\draw (y)++(0.5,-0.5)--++(2,0)--++(0,1)--++(-2,0)--++(0,-1);
\draw (x1)--++(0,1);
\draw (x2)--++(0,1);
\draw (x3)--++(0,1);
\draw (x4)--++(0,1);
\draw (x5)--++(0,1);
\draw (y)--++(0,1);
\draw (z)--++(0,1);
\draw (y)--++(0.5,0);
\draw (z)--++(-0.5,0);
\draw[fill=gray] (x2)++(0.7,0) circle (0.03);
\draw[fill=gray] (x2)++(1,0) circle (0.03);
\draw[fill=gray] (x2)++(1.3,0) circle (0.03);
\end{tikzpicture}}
\caption{Construction of the $\ork k$-gadget with interface vertices $x_1,x_2,\ldots, x_{k-2},x_{k-1},x_k$.}\label{fig:ork-from-ors-and-nand}
\end{figure}

\paragraph{Construction of a $\nand 2$-gadget.} First, we show how to construct the $\nand 2$-gadget. Observe that if we are able to construct the $\nand 2$-gadget as a path of length $\ell<g$, then we can easily obtain the $\nand 2$-gadget of length at least $g$. Indeed, we can introduce a path with $H$-lists $L$, of even length $\ell'>g-\ell$, with consecutive vertices $p_1,\ldots,p_\ell$ and lists $L(p_i)=\{\alpha,\beta\}$ for odd $i$ and $L(p_i)=\{\alpha',\beta'\}$ for even $i$, where $\alpha',\beta'$ are taken from \cref{obs:walks-between-corners}. We identify one endvertex of the $\nand 2$-gadget with one endvertex of the path to obtain another $\nand 2$-gadget of length at least $g$. Since edges $\alpha\alpha'$, $\beta\beta'$ induce a matching in $H$, in any list homomorphism both endvertices must be mapped to the same vertex, either $\alpha$ or $\beta$, and thus the properties of a $\nand{2}$-gadget are preserved by that operation.

Now let show how to construct the $\nand{2}$-gadget. If there is an induced $C_6$ in $H$ with consecutive vertices $w_1,\ldots,w_6$ such that $\alpha=w_1, \beta=w_5, \gamma=w_3$, then the $\nand 2$-gadget is a path of length $4$ with lists of consecutive vertices: $\{w_1,w_5\},\{w_2,w_6\},\{w_1,w_3\},\{w_2,w_4\},\{w_1,w_5\}$.
Similarly, if there is an induced $C_8$ in $H$ with consecutive vertices $w_1,\ldots,w_8$ such that $\alpha=w_1, \beta=w_5, \gamma=w_3$, then the $\nand 2$-gadget is a path of length $4$ with lists of consecutive vertices: $\{w_1,w_5\}$, $\{w_2,w_4,w_8\}$, $\{w_1,w_3,w_7\}$, $\{w_2,w_6\}$, $\{w_1,w_5\}$.
It is straightforward to verify that in both cases the constructed graph is indeed a $\nand 2$-gadget.
So now assume that $H$ does not contain an induced $C_6$ or an induced $C_8$ with $\alpha=w_1$, $\beta=w_5, \gamma=w_3$. Then by \cref{obs:walks-between-corners} we obtain walks: 
\begin{itemize}
\item $\cX_{\gamma}: \alpha \to \alpha, \cY_{\gamma}: \alpha \to \beta, \cZ_{\gamma}: \beta \to \gamma$ such that $\cX_{\gamma}, \cY_{\gamma}$ avoid $\cZ_{\gamma}$ and $\cZ_{\gamma}$ avoids $\cX_{\gamma},\cY_{\gamma}$,
\item $\cX_{\alpha}: \alpha \to \beta, \cY_{\alpha}: \alpha \to \gamma, \cZ_{\alpha}: \beta \to \alpha$ such that $\cX_{\alpha},\cY_{\alpha}$ avoid $\cZ_{\alpha}$ and $\cZ_{\alpha}$ avoids $\cX_{\alpha},\cY_{\alpha}$.
\end{itemize}

As the $\nand{2}$-gadget we take the composition of $\PP(\{\cX_{\gamma},\cY_{\gamma},\cZ_{\gamma}\})$ and $\PP(\{\oX_{\alpha},\oY_{\alpha},\oZ_{\alpha}\})$ (see \cref{fig:nand}). Note that the constructed graph is a path. It follows from properties of the walks and \cref{lem:gadget-of-walks} that it is indeed a $\nand 2$-gadget.

\begin{figure}
\centering{\begin{tikzpicture}[every node/.style={draw,circle,fill=white,inner sep=0pt,minimum size=8pt},every loop/.style={}]
\node[fill=gray,opacity=0.6,label=above:\footnotesize{$\{\alpha,\beta\}$}] (x1) at (0,0) {};
\node[label=above:\footnotesize{$\{\alpha,\beta,\gamma\}$}] (x2) at (5,0) {};
\node[fill=gray,opacity=0.6,label=above:\footnotesize{$\{\alpha,\beta\}$}] (x3) at (10,0) {};

\draw (x1)--++(1,1)--++(3,0)--(x2);
\draw (x2)--++(1,1)--++(3,0)--(x3);
\draw (x1)--++(1,-1)--++(3,0)--(x2);
\draw (x2)--++(1,-1)--++(3,0)--(x3);

\node[color=white,label=left:\footnotesize{$\alpha$}] (x11) at (1.5,0.5) {};
\node[color=white,label=left:\footnotesize{$\beta$}] (x12) at (1.5,-0.5) {};
\node[color=white,label=right:\footnotesize{$\alpha$}] (x13) at (3.5,0.5) {};
\node[color=white,label=right:\footnotesize{$\beta$}] (x14) at (3.5,0) {};
\node[color=white,label=right:\footnotesize{$\gamma$}] (x15) at (3.5,-0.5) {};

\node[color=white,label=right:\footnotesize{$\beta$}] (x24) at (8.5,0.5) {};
\node[color=white,label=right:\footnotesize{$\alpha$}] (x25) at (8.5,-0.5) {};
\node[color=white,label=left:\footnotesize{$\alpha$}] (x21) at (6.5,0.5) {};
\node[color=white,label=left:\footnotesize{$\beta$}] (x22) at (6.5,0) {};
\node[color=white,label=left:\footnotesize{$\gamma$}] (x23) at (6.5,-0.5) {};

\draw[thick, color=blue] (x11)--(x13);
\draw[thick, color=blue] (x11)--(x14);
\draw[thick, color=blue] (x12)--(x15);
\draw[thick, color=blue] (x21)--(x24);
\draw[thick, color=blue] (x22)--(x25);
\draw[thick, color=blue] (x23)--(x25);

\node[color=white,label=below:\footnotesize{$\PP(\{\cX_{\gamma},\cY_{\gamma},\cZ_{\gamma}\})$}] (a) at (2.5,2.7) {};
\node[color=white,label=below:\footnotesize{$\PP(\{\oX_{\alpha},\oY_{\alpha},\oZ_{\alpha}\})$}] (a) at (7.5,2.7) {};

\end{tikzpicture}}
\caption{The $\nand 2$ gadget as the composition of $\PP(\{\cX_{\gamma},\cY_{\gamma},\cZ_{\gamma}\})$ and $\PP(\{\oX_{\alpha},\oY_{\alpha},\oZ_{\alpha}\})$. Blue lines inside the gadgets denote possible mappings of the input and the output vertex. Interface vertices are marked gray.}\label{fig:nand}
\end{figure}
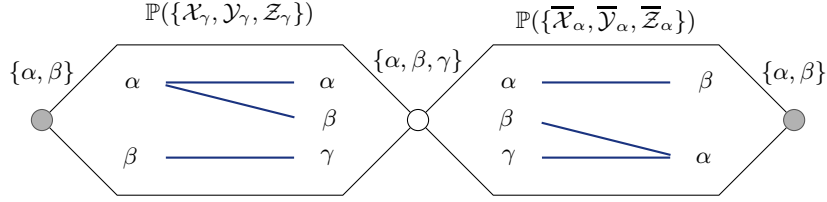

\paragraph{Construction of an $\ork 3$-gadget.} Now it remains to construct the $\ork 3$-gadget. In order to do that, for every $c \in \{\alpha,\beta,\gamma\}$ we construct an $R_c$-gadget, where $R_c=\big(\{\alpha,\beta\} \times \{\alpha,\beta,\gamma\}\big)\setminus\{(\alpha,c)\}$. Moreover, the $R_c$-gadget will be a path with endvertices $x,y$ such that $L(x)=\{\alpha,\beta\}$ and $L(y)=\{\alpha,\beta,\gamma\}$ as interface vertices. Observe that if we identify $y$-vertices of the $R_{\alpha}$-, $R_{\beta}$-, and $R_{\gamma}$-gadget, then we obtain the $\ork 3$-gadget with $x$-vertices of the gadgets as interface vertices. Furthermore, in the $\ork{3}$-gadget there is only one vertex of degree $3$ and all other vertices are of degree smaller than $3$.

\begin{figure}
\centering{\begin{tikzpicture}[every node/.style={draw,circle,fill=white,inner sep=0pt,minimum size=8pt},every loop/.style={}]

\node[fill=black!20,label=right:\footnotesize{$\{w_1,w_5\}$}] (a) at (0,0) {};
\node[label=right:\footnotesize{$\{w_4,w_6\}$}] (b) at (0,1) {};
\node[label=right:\footnotesize{$\{w_1,w_3\}$}] (c) at (0,2) {};
\node[label=right:\footnotesize{$\{w_2,w_4\}$}] (d) at (0,3) {};
\node[fill=black!20,label=right:\footnotesize{$\{w_1,w_5\}$}] (e) at (2,0) {};
\node[label=right:\footnotesize{$\{w_4,w_6\}$}] (f) at (2,1) {};
\node[label=right:\footnotesize{$\{w_1,w_3,w_5\}$}] (g) at (2,4) {};
\node[label=right:\footnotesize{$\{w_2,w_4\}$}] (h) at (4,3) {};
\node[label=right:\footnotesize{$\{w_3,w_5\}$}] (i) at (4,2) {};
\node[label=right:\footnotesize{$\{w_4,w_6\}$}] (j) at (4,1) {};
\node[fill=black!20,label=right:\footnotesize{$\{w_1,w_5\}$}] (k) at (4,0) {};
\draw (a) -- (b) -- (c) -- (d) -- (g) -- (f) -- (e);
\draw (g) -- (h) -- (i) -- (j) -- (k);
\end{tikzpicture}\hskip 2cm
\begin{tikzpicture}[every node/.style={draw,circle,fill=white,inner sep=0pt,minimum size=8pt},every loop/.style={}]

\node[fill=black!20,label=right:\footnotesize{$\{w_1,w_5\}$}] (a) at (4,0) {};
\node[label=right:\footnotesize{$\{w_6,w_8\}$}] (b) at (4,1) {};
\node[label=right:\footnotesize{$\{w_5,w_7\}$}] (c) at (4,2) {};
\node[label=right:\footnotesize{$\{w_4,w_6,w_8\}$}] (d) at (4,3) {};
\node[fill=black!20,label=left:\footnotesize{$\{w_1,w_5\}$}] (e) at (2,0) {};
\node[label=left:\footnotesize{$\{w_4,w_8\}$}] (f) at (2,1) {};
\node[label=left:\footnotesize{$\{w_1,w_3\}$}] (g) at (2,2) {};
\node[label=left:\footnotesize{$\{w_2,w_4\}$}] (h) at (2,3) {};
\node[label=right:\footnotesize{$\{w_1,w_3,w_5\}$}] (i) at (2,6) {};
\node[fill=black!20,label=left:\footnotesize{$\{w_1,w_5\}$}] (j) at (0,0) {};
\node[label=left:\footnotesize{$\{w_4,w_8\}$}] (k) at (0,1) {};
\node[label=left:\footnotesize{$\{w_5,w_7\}$}] (l) at (0,2) {};
\node[label=left:\footnotesize{$\{w_4,w_6\}$}] (m) at (0,3) {};
\node[label=left:\footnotesize{$\{w_3,w_5\}$}] (n) at (0,4) {};
\node[label=left:\footnotesize{$\{w_2,w_4\}$}] (o) at (0,5) {};
\draw (a) -- (b) -- (c) -- (d) -- (i) -- (h) -- (g) -- (f) -- (e);
\draw (i) -- (o) -- (n) -- (m) -- (l) -- (k) -- (j);
\end{tikzpicture}}
\caption{The $\ork 3$-gadget in case that $H$ contains an induced $C_6$ (left) and an induced $C_8$ (right). Recall that the consecutive vertices of those cycles are denoted by $(w_1,w_2,\ldots)$. The sets next to vertices indicate lists. Interface vertices are marked gray. The figure is taken from \cite{LhomoTreewidth,LhomoTreewidthFull}.}\label{fig:cycle-ors}
\end{figure}
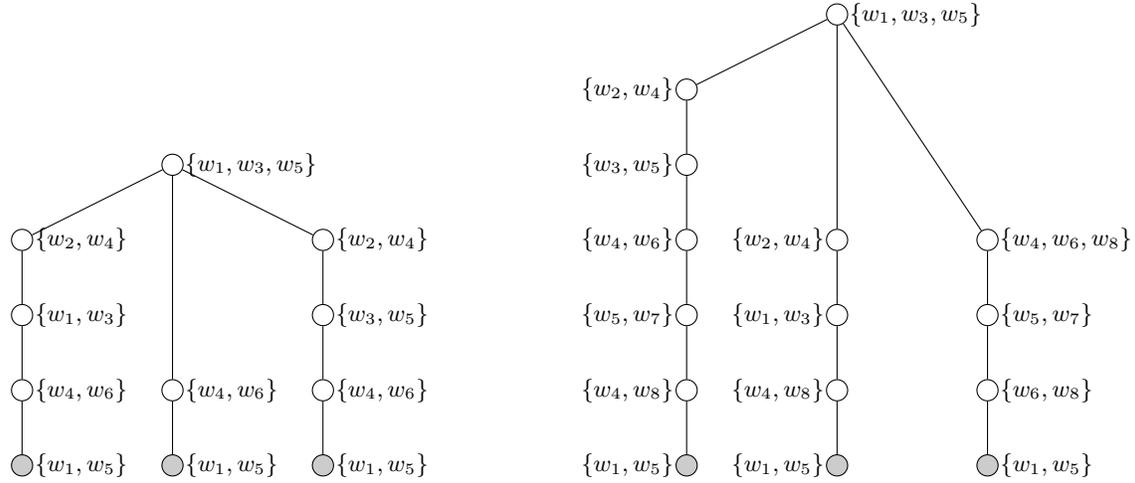

So we only need to construct the $R_c$-gadget for every $c\in\{\alpha,\beta,\gamma\}$. If there is an induced $C_6$ in $H$ with consecutive vertices $w_1,\ldots,w_6$ such that $\alpha=w_1,\beta=w_5,\gamma=w_3$, then for example the $R_{\beta}$-gadget is a path of length $4$ with consecutive lists $\{w_1,w_5\},\{w_2,w_4\},\{w_3,w_5\},\{w_2,w_6\},\{w_1,w_3,w_5\}$. The construction of the other $R_c$-gadgets and the case of $C_8$ is similar (see \cref{fig:cycle-ors} where $y$-vertices of the gadgets are already identified). Now let $\{a,b,c\}=\{\alpha,\beta,\gamma\}$ and assume that $H$ does not contain an induced $C_6$ or an induced $C_8$ with $\alpha=w_1$, $\beta=w_5$ and $\gamma=w_3$. Then by \cref{obs:walks-between-corners} there exist the following walks:
\begin{itemize}
\item $\cX: \alpha \to \beta, \cY: \beta \to \alpha$, such that $\cX$ avoids $\cY$,
\item $\cX_a: \alpha \to b, \cY_a: \alpha \to c, \cZ_a: \beta \to a$, such that $\cX_a, \cY_a$ avoid $\cZ_a$ and $\cZ_a$ avoids $\cX_a, \cY_a$,
\item $\cX_c: \alpha \to a, \cY_c: \alpha \to b, \cZ_c: \beta \to c$, such that $\cX_c,\cY_c$ avoid $\cZ_c$ and $\cZ_c$ avoids $\cX_c,\cY_c$.
\end{itemize}

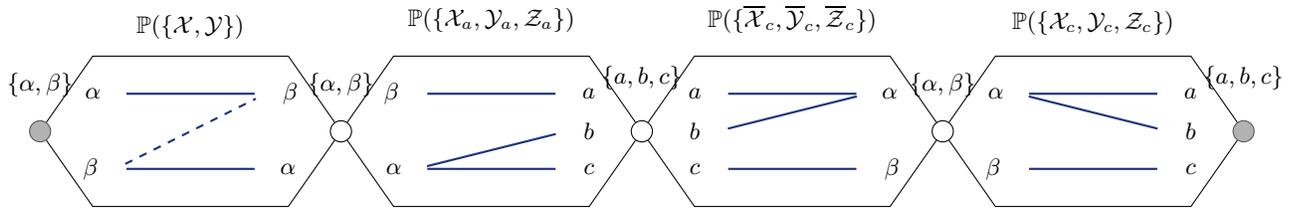
\begin{figure}
\centering{\begin{tikzpicture}[every node/.style={draw,circle,fill=white,inner sep=0pt,minimum size=8pt},every loop/.style={}]
\node[fill=gray,opacity=0.6,label=above:\footnotesize{$\{\alpha,\beta\}$}] (x1) at (0,0) {};
\node[label=above:\footnotesize{$\{\alpha,\beta\}$}] (x2) at (4,0) {};
\node[label=above:\footnotesize{$\{a,b,c\}$}] (x3) at (8,0) {};
\node[label=above:\footnotesize{$\{\alpha,\beta\}$}] (x4) at (12,0) {};
\node[fill=gray,opacity=0.6,label=above:\footnotesize{$\{a,b,c\}$}] (x5) at (16,0) {};

\node[color=white,label=left:\footnotesize{$\alpha$}] (x11) at (1,0.5) {};
\node[color=white,label=left:\footnotesize{$\beta$}] (x12) at (1,-0.5) {};
\node[color=white,label=right:\footnotesize{$\alpha$}] (x14) at (3,-0.5) {};
\node[color=white,label=right:\footnotesize{$\beta$}] (x13) at (3,0.5) {};

\node[color=white,label=left:\footnotesize{$\alpha$}] (x22) at (5,-0.5) {};
\node[color=white,label=left:\footnotesize{$\beta$}] (x21) at (5,0.5) {};
\node[color=white,label=right:\footnotesize{$a$}] (x23) at (7,0.5) {};
\node[color=white,label=right:\footnotesize{$b$}] (x24) at (7,0) {};
\node[color=white,label=right:\footnotesize{$c$}] (x25) at (7,-0.5) {};

\node[color=white,label=left:\footnotesize{$a$}] (x31) at (9,0.5) {};
\node[color=white,label=left:\footnotesize{$b$}] (x32) at (9,0) {};
\node[color=white,label=left:\footnotesize{$c$}] (x33) at (9,-0.5) {};
\node[color=white,label=right:\footnotesize{$\alpha$}] (x34) at (11,0.5) {};
\node[color=white,label=right:\footnotesize{$\beta$}] (x35) at (11,-0.5) {};

\node[color=white,label=left:\footnotesize{$\alpha$}] (x41) at (13,0.5) {};
\node[color=white,label=left:\footnotesize{$\beta$}] (x42) at (13,-0.5) {};
\node[color=white,label=right:\footnotesize{$a$}] (x43) at (15,0.5) {};
\node[color=white,label=right:\footnotesize{$b$}] (x44) at (15,0) {};
\node[color=white,label=right:\footnotesize{$c$}] (x45) at (15,-0.5) {};

\draw (x1)--++(0.7,1)--++(2.6,0)--(x2);
\draw (x2)--++(0.7,1)--++(2.6,0)--(x3);
\draw (x3)--++(0.7,1)--++(2.6,0)--(x4);
\draw (x4)--++(0.7,1)--++(2.6,0)--(x5);
\draw (x1)--++(0.7,-1)--++(2.6,0)--(x2);
\draw (x2)--++(0.7,-1)--++(2.6,0)--(x3);
\draw (x3)--++(0.7,-1)--++(2.6,0)--(x4);
\draw (x4)--++(0.7,-1)--++(2.6,0)--(x5);

\draw[thick, color=blue] (x11)--(x13);
\draw[thick, color=blue] (x12)--(x14);
\draw[thick, dashed, color=blue] (x12)--(x13);
\draw[thick, color=blue] (x21)--(x23);
\draw[thick, color=blue] (x22)--(x24);
\draw[thick, color=blue] (x22)--(x25);
\draw[thick, color=blue] (x31)--(x34);
\draw[thick, color=blue] (x32)--(x34);
\draw[thick, color=blue] (x33)--(x35);
\draw[thick, color=blue] (x41)--(x43);
\draw[thick, color=blue] (x41)--(x44);
\draw[thick, color=blue] (x42)--(x45);

\node[draw=none,fill=none,label=above:\footnotesize{$\PP(\{\cX,\cY\})$}] (a1) at (2,0.5) {};
\node[draw=none,fill=none,label=above:\footnotesize{$\PP(\{\cX_a,\cY_a,\cZ_a\})$}] (a2) at (6,0.2) {};
\node[draw=none,fill=none,label=above:\footnotesize{$\PP(\{\oX_c,\oY_c,\oZ_c\})$}] (a3) at (10,0.2) {};
\node[draw=none,fill=none,label=above:\footnotesize{$\PP(\{\cX_c,\cY_c,\cZ_c\})$}] (a4) at (14,0.2) {};

\end{tikzpicture}}
\caption{Construction of the $R_c$-gadget. Blue lines inside the gadgets denote possible mappings of the input and the output vertex. The dashed line indicates a mapping that might exist but not necessarily.}\label{fig:chooser}
\end{figure}

Now we set the $R_c$-gadget as the composition of the gadgets: $\PP(\{\cX,\cY\})$, $\PP(\{\cX_a,\cY_a,\cZ_a\})$, $\PP(\{\oX_c,\oY_c,\oZ_c\})$, and $\PP(\{\cX_c,\cY_c,\cZ_c\})$ (see \cref{fig:chooser}). It follows from the definition of those walks and \cref{lem:gadget-of-walks} that the constructed graph is indeed an $R_c$-gadget. This completes the proof.
\end{proof}

The next gadget we will use to construct the assignment gadget is called \emph{distiguisher} and was constructed in \cite{LhomoTreewidth,LhomoTreewidthFull}. It is a graph $D_{a/b}$, for distinct vertices $a$, $b$ in $H$, with $H$-lists $L$ and two special vertices $x$, $y$. The lists of the special vertices are $L(x)=S$, for an incomparable set $S$ such that $a,b \in S$, and $L(y)= \{\alpha,\beta\}$. In $D_{a/b}$ we can distingiush $a$ and $b$ by mapping $y$ to $\beta$, i.e., if $x$ is mapped to $a$, then $y$ must be mapped to $\alpha$ and it is possible to map $x$ to $b$ and $y$ to $\beta$.

\begin{lemma}[Construction of the distinguisher gadget \cite{LhomoTreewidth,LhomoTreewidthFull}]\label{lem:distinguisher}
Let $H$ be a connected, bipatrite, undecomposable graph whose complement is not a circular-arc graph. Let $\alpha,\beta$ be defined as in \cref{obs:walks-between-corners}. Let $S$ be an incomparable set in $H$, such that $|S| \ge 2$ and $\{\alpha,\beta\} \cup S$ is contained in one bipartition class of $H$.
Let $a,b$ be distinct vertices of $S$.
Then there exists a \emph{distinguisher} gadget which is a path $D_{a/b}$ with two endvertices $x$ (called \emph{input}) and $y$ (called \emph{output}),
and $H$-lists $L$ such that:
\begin{enumerate}[(D1.)]
\item $L(x)=S$ and $L(y)=\{\alpha,\beta\}$,
\item there is a list homomorphism $\vphi_a: (D_{a/b},L) \to H$, such that $\vphi_a(x)=a$ and $\vphi_a(y)=\alpha$,
\item there is a list homomorphism $\vphi_b: (D_{a/b},L) \to H$, such that $\vphi_b(x)=b$ and $\vphi_b(y)=\beta$,
\item for any $c \in S \setminus \{a,b\}$ there is $\vphi_c: (D_{a/b},L) \to H$, such that $\vphi_c(x)=c$ and $\vphi_c(y) \in \{\alpha,\beta\}$,
\item there is no list homomorphism $\vphi: (D_{a/b},L) \to H$, such that $\vphi(x)=a$ and $\vphi(y)=\beta$.
\end{enumerate} 
\end{lemma}

Now we introduce a \emph{detector} gadget which will be useful in the construction of an assignment gadget. It is a graph $\tF_u$ with $H$-lists $L$ and two special vertices $x_u$, $c_u$. It will be used to detect if $x_u$ is colored with $u$ as then $c_u$ must be colored with $\beta$, and for every other coloring of $x_u$, the vertex $c_u$ can be colored with $\alpha$.

\begin{definition}[Detector gadget]\label{def:detector-gadget}
Let $H$ be a bipartite graph, whose complement is not a circular-arc graph. Let $\alpha,\beta$ be defined as in \cref{obs:walks-between-corners}. Let $S$ be a set of vertices in $H$ contained in one bipartition class, let $k:=|S|\geq2$, and let $u \in S$. A detector gagdet is a graph $\tF_u$ with special vertices $x_u$ and $c_u$, with $H$-lists $L$, such that:
\begin{enumerate}[($\widetilde{F}$1.)]
\item $L(x_u)=S$ and $L(c_u)=\{\alpha,\beta\}$,
\item for every $s \in S$ there exists a list homomorphism $\vphi: (\tF_u,L) \to H$ such that $\vphi(x_u)=s$ and $\vphi(c_u)=\beta$, 
\item for every $s \in S\setminus\{u\}$ there exists a list homomorphism $\vphi: (\tF_u,L) \to H$ such that $\vphi(x_u)=s$ and $\vphi(c_u)=\alpha$,
\item for every list homomorphism $\vphi: (\tF_u,L) \to H$, if $\vphi(x_u)=u$, then $\vphi(c_u)=\beta$, 
\item $\tF_u-\{x_u\}$ is a tree,
\item $\deg(x_u)=k-1$ and $\deg(c_u)=1$,
\item the degree of every vertex in $\tF_u $, possibly except $x_u$, is at most $3$.
\end{enumerate}
\end{definition}

The gadget was first constructed in \cite{LhomoTreewidthFull} (the first step in the proof of Lemma~4), but we repeat the construction to obtain a certain structure of the gadget.

\begin{lemma}[Construction of the detector gadget~\cite{LhomoTreewidthFull}]\label{lem:detector-gadget}
Let $H$ be a connected, bipatrite, undecomposable graph whose complement is not a circular-arc graph. Let $\alpha,\beta$ be defined as in \cref{obs:walks-between-corners}. Let $S$ be an incomparable set in $H$, such that $|S| \ge 2$ and $\{\alpha,\beta\} \cup S$ is contained in one bipartition class of $H$. Let $g \in \N$ and let $u \in S$. Then there exists a detector gadget $\tF_u$, such that $\girth{\tF_u}\geq g$ and for any distinct $a,b$ in $\tF_u$ of degree at least $3$ it holds that $\dist(a,b) \geq g$ and $\dist(a,x_u)\geq g$.
\end{lemma}

\begin{proof}
For every $w \in S\setminus\{u\}$ we call \cref{lem:distinguisher} for $S,u,w$ to obtain a distinguisher gadget $D_{u/w}$ with input $x_{u,w}$ and output $y_{u,w}$. We can assume that every distinguisher $D_{u/w}$ is a path of length at least $g$. Otherwise, if there is any distinguisher $D_{u/w}$ that is a path of length $\ell<g$, then similarly as in the proof of \cref{lem:relationgadgets} (construction of the $\nand{2}$-gadget) we can append to the vertex $y_{u,w}$ a path of even length $\ell'>g-\ell$ with lists of consecutive vertices $\{\alpha,\beta\},\{\alpha',\beta'\},\ldots,\{\alpha,\beta\}$, where $\alpha',\beta'$ are taken from \cref{obs:walks-between-corners}, and the properties of the distinguisher gadget are preserved. We identify all input vertices into one vertex $x_u$. Then we call \cref{lem:relationgadgets} to construct an $\ork k$ gadget and identify $k-1$ of its $k$ interface vertices with the output vertices of distinguishers. Let us call the remaining $k$-th interface vertex of the $\ork k$ gadget $c_u$ and the constructed graph $\tF_u$ (see \cref{fig:tilde-Fu-gadget}).
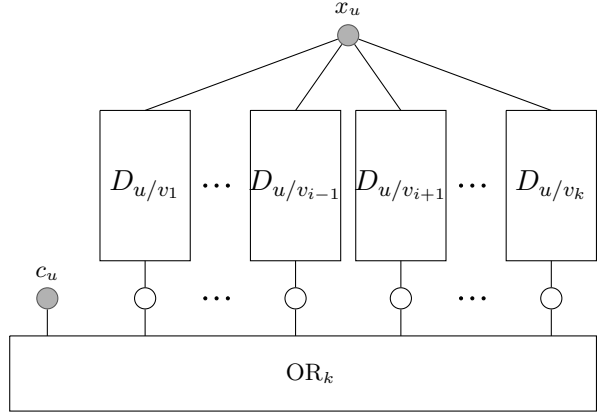
\begin{figure}
\centering{\begin{tikzpicture}[every node/.style={draw,circle,fill=white,inner sep=0pt,minimum size=8pt},every loop/.style={}]
\node[fill=gray,opacity=0.6,label=above:\footnotesize{$x_u$}] (x1) at (0,0) {};
\foreach \k in {-2.7,-0.7,0.7,2.7}
{
\draw (x1)--(\k,-1);
\draw (\k-0.6,-1)--++(1.2,0)--++(0,-2)--++(-1.2,0)--++(0,2);
\node (x) at (\k,-3.5) {};
\draw (x)--++(0,0.5);
\draw (x)--++(0,-0.5);
}
\draw (-4.5,-4)--(3.3,-4)--++(0,-1)--(-4.5,-5)--++(0,1);
\node[fill=gray,opacity=0.6,label=above:\footnotesize{$c_u$}] (c) at (-4,-3.5) {};
\draw (c)--++(0,-0.5);

\foreach \i in {-1.9,-1.75,-1.6}
{
\draw[fill=black] (\i,-2) circle (0.02);
\draw[fill=black] (\i+3.4,-2) circle (0.02);
\draw[fill=black] (\i,-3.5) circle (0.02);
\draw[fill=black] (\i+3.4,-3.5) circle (0.02);
}
\node[draw=none,fill=none,label=above:\footnotesize{$\ork k$}] (a) at (-0.5,-5) {};
\node[draw=none,fill=none] (a) at (-2.7,-2) {$D_{u/v_1}$};
\node[draw=none,fill=none] (a) at (-0.7,-2) {$D_{u/v_{i-1}}$};
\node[draw=none,fill=none] (a) at (0.7,-2) {$D_{u/v_{i+1}}$};
\node[draw=none,fill=none] (a) at (2.7,-2) {$D_{u/v_k}$};
\end{tikzpicture}}
\caption{The graph $\widetilde{F}_u$ for $S=\{v_1,\ldots,v_k\}$ and $u=v_i$.}\label{fig:tilde-Fu-gadget}
\end{figure}

Now let us verify, that the constructed graph is indeed a detector gadget. Property ($\tF$1.) is obviously satisfied by the construction of $\tF_u$. To show ($\widetilde{F}$2.), consider a list homomorphism $\vphi$ such that $\vphi(x_u)=s$ and $\vphi(y_{u,w}) \in \{\alpha,\beta\}$ for every $w \in S\setminus\{u\}$. It exists by the definition of a distinguisher (properties (D2.), (D3.), and (D4.)). We set $\vphi(c_u):=\beta$ and then we can extend $\vphi$ to all vertices of the $\ork k$-gadget. Property ($\widetilde{F}$3.) follows from the fact that there exists a list homomorphism $\vphi: (D_{u/s},L) \to H$ such that $\vphi(y_{u,s})=\beta$ (property (D3.)). By (D4.) we can extend $\vphi$ to $D_{u/w}$ for every $w \in S\setminus\{u,s\}$ so that $\vphi(y_{u,w})\in \{\alpha,\beta\}$. We can set $\vphi(c_u)=\alpha$ and extend $\vphi$ to remaining vertices of the $\ork k$-gadget. To show ($\widetilde{F}$4.), consider a list homomorphism $\vphi: (\tF_u,L) \to H$, such that $\vphi(x_u)=u$. Then by (D5.) it holds that $\vphi(y_{u,w})=\alpha$ for every $w \in S\setminus\{u\}$. Since at least one of interface vertices of the $\ork k$-gadget must be mapped to $\beta$, we conclude that $\vphi(c_u)=\beta$. Property ($\widetilde{F}$5.) follows from the fact that each distinguisher gadget is a path and all distinguishers share one common vertex, which is $x_u$. Moreover, the $\ork k$-gadget is a tree and every distinguisher has one common vertex (endvertex) with the $\ork k$-gadget. Thus if we remove $x_u$ from $\tF_u$, then we obtain a tree (the $\ork k$-gadget) and pairwise disjoint paths added to that tree by identifying single vertices, so $\tF_u - \{x_u\}$ is a tree (see \cref{fig:tilde-Fu-gadget}). Properties ($\widetilde{F}$6.) and ($\widetilde{F}$7.) follow directly from the construction of $\tF_u$. First we join $k-1$ paths by identifying their endvertices into one vertex $x_u$ and thus $\deg(x_u)=k-1$. Then we introduce an $\ork k$-gadget whose maximum degree is at most $3$ and we identify its interface vertices, whose degrees are all $1$, with the other endvertices of the paths, so the degree of each vertex, possibly except $x_u$, is at most $3$. Moreover, we do not modify the interface vertex $c_u$ of the $\ork k$-gadget, so its degree remains $1$. Thus $\tF_u$ is indeed a detector gadget. So it remains to show that the girth of $\tF_u$ is at least $g$, the vertices of degree at least $3$ are pairwise at distance at least $g$ and at distance at least $g$ from $x_u$ (note that we consider $x_u$ separately, since for $|S|\leq 3$ the degree of $x_u$ is less than $3$). First observe that the only vertices of degree at least $3$ might be $x_u$ and some vertices of the $\ork{k}$-gadget, and in the $\ork{k}$-gadget any two distinct vertices of degree at least $3$ are at distance at least $g$. Moreover, $x_u$ is at distance at least $g$ from any vertex of the $\ork{k}$-gadget since every distinguisher gadget has length at least $g$. Therefore, vertices of degree at least $3$ in $\tF_u$ are at distance at least $g$ and at distance at least $g$ from $x_u$. Finally, $\girth{\tF_u}\geq g$ since every cycle in $\tF_u$ contains $x_u$, whose every adjacent edge belongs to one of $k-1$ induced paths, each of length at least $g$. 
\end{proof}

Now we are ready to show how to construct the assignment gadget, but first let us remind the definition.

\defAssign*

\lemAssignGadget*

\begin{proof}
The construction of the assignment gadget $A_v$ will be done in three steps.

\paragraph{Step I.} The first step is the construction of a path $P_u$ with $H$-lists $L$, and endvertices $c,y$ with $L(c)=\{\alpha,\beta\}$, $L(y)=\{\alpha,\beta,\gamma\}$, satisfying the following properties:
\begin{enumerate}[(P1.)]
\item for every $a \in \{\alpha,\beta\}$ and for every $b\in \{\alpha,\beta\}$ there exists a list homomorphism $\vphi: (P_u,L)\to H$ such that $\vphi(c)=a$ and $\vphi(y)=b$,
\item there exists a list homomorphism $\vphi: (P_u,L)\to H$ such that $\vphi(c)=\alpha$ and $\vphi(y)=\gamma$,
\item there is no list homomorphism $\vphi: (P_u,L)\to H$ such that $\vphi(c)=\beta$ and $\vphi(y)=\gamma$.
\item the length of $P_u$ is at least $g$.
\end{enumerate}
First observe that we can assume that the last property holds. Otherwise, as in the previous constructions we can append a path of appropriate length with consecutive lists $\{\alpha,\beta\},\{\alpha',\beta'\},\ldots,\{\alpha,\beta\}$ to the path $P_u$ by identifying one endvertex with $c$. Since $\alpha\alpha'$ and $\beta\beta'$ induce a matching in $H$, adding such a path preserves the other desired properties of the path $P_u$.

So let us show how to construct $P_u$, which satisfies properties (P1.), (P2.), and (P3.). First, consider the case that $H$ contains $C_6$ as an induced subgraph with consecutive vertices $w_1,\ldots,w_6$ and $\alpha=w_1$, $\beta=w_5$, $\gamma=w_3$. Then we set as $P_u$ a path of length $2$ with lists of consecutive vertices $\{w_1,w_5\},\{w_2,w_6\},\{w_1,w_3,w_5\}$. If $H$ contains an induced $C_8$ with consecutive vertices $w_1,\ldots,w_8$ and $\alpha=w_1$, $\beta=w_5$, $\gamma=w_3$, then we set as $P_u$ a path of length $4$ with lists of consecutive vertices $\{w_1,w_5\},\{w_2,w_6\},\{w_1,w_3,w_7\}, \{w_2,w_4,w_6,w_8\},\{w_1,w_3,w_5\}$. It is straightforward to verify that in both cases $P_u$ satisfies the required properties.
If $H$ does not contain $C_6$ or $C_8$ as an induced subgraph with $\alpha=w_1$, $\beta=w_5$, and $\gamma=w_3$, then by \cref{obs:walks-between-corners}, the following walks exist in $H$:
\begin{itemize}
\item $\cX_{\alpha}: \alpha \to \beta, \cY_{\alpha}: \alpha \to \gamma, \cZ_{\alpha}: \beta \to \alpha$, such that $\cX_{\alpha},\cY_{\alpha}$ avoid $\cZ_{\alpha}$ and $\cZ_{\alpha}$ avoids $\cX_{\alpha},\cY_{\alpha}$,
\item $\cX_{\gamma}: \alpha \to \alpha, \cY_{\gamma}: \alpha \to \beta, \cZ_{\gamma}: \beta \to \gamma$, such that $\cX_{\gamma},\cY_{\gamma}$ avoid $\cZ_{\gamma}$ and $\cZ_{\gamma}$ avoids $\cX_{\gamma},\cY_{\gamma}$.
\end{itemize}

As the path $P_u$ we take the composition of gadgets: $\PP(\{\cX_{\alpha},\cY_{\alpha},\cZ_{\alpha}\}), \PP(\{\overline{\cX_{\gamma}},\overline{\cY_{\gamma}},\overline{\cZ_{\gamma}}\}), \PP(\{\cX_{\gamma},\cY_{\gamma},\cZ_{\gamma}\})$ (see \cref{fig:path}).

\paragraph{Step II.} The next step is the construction of a graph $F_u$. For every $u \in S\setminus \{v\}$ we call \cref{lem:detector-gadget} for $H,S,u,g$ to construct a detector gadget $\tF_u$. Then, for every $u \in S\setminus \{v\}$, we join the path $P_u$ with the graph $\tF_u$ by identifying vertices $c$ and $c_u$. Let us call the constructed graph $F_u$ and the vertex $y$ from the path $P_u$ let us call $y_u$.

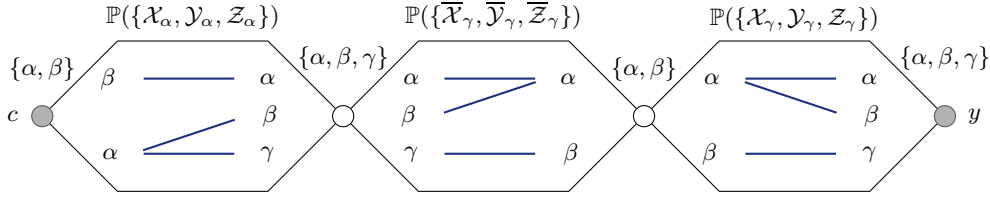
\begin{figure}
\centering{\begin{tikzpicture}[every node/.style={draw,circle,fill=white,inner sep=0pt,minimum size=8pt},every loop/.style={}]
\node[fill=gray,opacity=0.6,label=above:\footnotesize{$\{\alpha,\beta\}$}] (x1) at (0,0) {};
\node[label=above:\footnotesize{$\{\alpha,\beta,\gamma\}$}] (x2) at (4,0) {};
\node[label=above:\footnotesize{$\{\alpha,\beta\}$}] (x3) at (8,0) {};
\node[fill=gray,opacity=0.6,label=above:\footnotesize{$\{\alpha,\beta,\gamma\}$}] (x4) at (12,0) {};

\draw (x1)--++(1,1)--++(2,0)--(x2);
\draw (x2)--++(1,1)--++(2,0)--(x3);
\draw (x3)--++(1,1)--++(2,0)--(x4);
\draw (x1)--++(1,-1)--++(2,0)--(x2);
\draw (x2)--++(1,-1)--++(2,0)--(x3);
\draw (x3)--++(1,-1)--++(2,0)--(x4);

\node[color=white,label=left:\footnotesize{$\alpha$}] (x12) at (1.2,-0.5) {};
\node[color=white,label=left:\footnotesize{$\beta$}] (x11) at (1.2,0.5) {};
\node[color=white,label=right:\footnotesize{$\alpha$}] (x13) at (2.7,0.5) {};
\node[color=white,label=right:\footnotesize{$\beta$}] (x14) at (2.7,0) {};
\node[color=white,label=right:\footnotesize{$\gamma$}] (x15) at (2.7,-0.5) {};

\node[color=white,label=right:\footnotesize{$\alpha$}] (x24) at (6.7,0.5) {};
\node[color=white,label=right:\footnotesize{$\beta$}] (x25) at (6.7,-0.5) {};
\node[color=white,label=left:\footnotesize{$\alpha$}] (x21) at (5.2,0.5) {};
\node[color=white,label=left:\footnotesize{$\beta$}] (x22) at (5.2,0) {};
\node[color=white,label=left:\footnotesize{$\gamma$}] (x23) at (5.2,-0.5) {};

\node[color=white,label=left:\footnotesize{$\beta$}] (x32) at (9.2,-0.5) {};
\node[color=white,label=left:\footnotesize{$\alpha$}] (x31) at (9.2,0.5) {};
\node[color=white,label=right:\footnotesize{$\alpha$}] (x33) at (10.7,0.5) {};
\node[color=white,label=right:\footnotesize{$\beta$}] (x34) at (10.7,0) {};
\node[color=white,label=right:\footnotesize{$\gamma$}] (x35) at (10.7,-0.5) {};

\draw[thick, color=blue] (x11)--(x13);
\draw[thick, color=blue] (x12)--(x14);
\draw[thick, color=blue] (x12)--(x15);
\draw[thick, color=blue] (x21)--(x24);
\draw[thick, color=blue] (x22)--(x24);
\draw[thick, color=blue] (x23)--(x25);
\draw[thick, color=blue] (x31)--(x33);
\draw[thick, color=blue] (x31)--(x34);
\draw[thick, color=blue] (x32)--(x35);

\node[draw=none,fill=none,label=above:\footnotesize{$\PP(\{\cX_{\alpha},\cY_{\alpha},\cZ_{\alpha}\})$}] (a) at (2,0) {};
\node[draw=none,fill=none,label=above:\footnotesize{$\PP(\{\oX_{\gamma},\oY_{\gamma},\oZ_{\gamma}\})$}] (b) at (6,0) {};
\node[draw=none,fill=none,label=above:\footnotesize{$\PP(\{\cX_{\gamma},\cY_{\gamma},\cZ_{\gamma}\})$}] (c) at (10,0) {};
\node[draw=none,fill=none,label=left:\footnotesize{$c$}] (d) at (-0.1,0) {};
\node[draw=none,fill=none,label=right:\footnotesize{$y$}] (e) at (12.1,0) {};

\end{tikzpicture}}
\caption{The path $P_u$. Blue lines inside each gadget denote possible mappings of the input and the output vertex of the gadget.}\label{fig:path}
\end{figure}

Observe that the graph $F_u$ satisfies the following properties:
\begin{enumerate}[(F1.)]
\item $L(x_u)=S$ and $L(y_u)=\{\alpha,\beta,\gamma\}$.
\item For every $s \in S$ and every $a \in \{\alpha,\beta\}$ there exists a list homomorphism $\vphi: (F_u,L) \to H$ such that $\vphi(x_u)=s$ and $\vphi(y_u)=a$.
\item For every $s \in S\setminus\{u\}$ there exists a list homomorphism $\vphi: (F_u,L) \to H$ such that $\vphi(x_u)=s$ and $\vphi(y_u)=\gamma$. 
\item There is no list homomorphism $\vphi: (F_u,L) \to H$ such that $\vphi(x_u)=u$ and $\vphi(y_u)=\gamma$. 
\item $F_u - \{x_u\}$ is a tree. 
\item $\deg(x_u)=k-1$ and $\deg(y_u)=1$.
\item The degree of every vertex except $x_u$ in $F_u$ is at most $3$.
\item Vertices of degree at least $3$ are at distance at least $g$.
\item $\girth{F_u}\geq g$.
\item $y_u$ and $x_u$ are at distance at least $g$ from each other and from any vertex of degree at least $3$.
\end{enumerate}

Property (F1.) follows directly from the construction of $F_u$. Property (F2.) follows from property ($\widetilde{F}$2.) of $\tF_u$ and property (P1.) of $P_u$. Property (F3.) follows from properties ($\widetilde{F}$3.) and (P2.). To show (F4.) consider a list homomorphism $\vphi: (F_u,L) \to H$ such that $\vphi(x_u)=u$. By the property ($\tF$4.) it holds that $\vphi(c_u)=\beta$, which by the property (P3.) implies that $\vphi(y_u)\neq \gamma$. Property (F5.) follows directly from ($\widetilde{F}$5.), since we only joined $\tF_u$ with a path by identyfing their single vertices. Properties (F6.) and (F7.) follow from properties ($\widetilde{F}$6.), ($\widetilde{F}$7.), and the fact that we added a path to the graph $\tF_u$ by identifying an endvertex with a vertex of degree $1$. Properties (F8.) and (F9.) follow from the fact that by \cref{lem:detector-gadget}, every $\tF_u$ has girth at least $g$ and the vertices of degree at least $3$ are at distance at least $g$. Property (F10.) follows from property (P4.) of $P_u$ and the fact that we constructed $\tF_u$ such that $x_u$ is at distance at least $g$ from any vertex of degree at least $3$.

\paragraph{Step III.} Finally, we are able to construct an assignment gadget $A_v$. For every $u \in S\setminus\{v\}$ we introduce the gadget $F_u$ with special vertices $x_u,y_u$ and we identify all $x_u$'s into a single vertex $x$ and all $y_u$'s into a single vertex $y$. That completes the construction of $A_v$. It only remains to show that $A_v$ is indeed an assignment gadget with girth at least $g$ and vertices of degree at least $3$ pairwise at distance at least $g$. Properties (A1.), (A2.), and (A3.) from the \cref{def:assign-gadget} follow directly from, respectively, properties (F1.), (F2.) and (F3.) of $F_u$. To show (A4.) consider a list homomorphism $\vphi: (A_v,L) \to H$ such that $\vphi(y)=\gamma$. By the property (F4.), for every $u \in S\setminus\{v\}$, it holds that $\vphi(x)\neq u$, so the only possible mapping of $x$ is $\vphi(x)=v$. Property (A5.) follows from the property (F5.) of $F_u$ and the fact that we joined $F_u$'s by identifying $x_u$'s into one vertex $x$ and $y_u$'s into one vertex $y$. Properties (A6.) and (A7.) follow from properties (F6.) and (F7.) of $F_u$ since we joined $k-1$ graphs $F_u$ (one for every $u \in S \setminus \{v\}$) by identifying vertices $x_u$, each of degree $k-1$, into one vertex $x$ of degree $(k-1)^2$ and identifying vertices $y_u$, each of degree $1$, into one vertex $y$ of degree $k-1$. Degrees of the other vertices did not change. Finally, let us show that $\girth{A_v} \geq g$, the vertices of degree at least $3$ are pairwise at distance at least $g$, at distance at least $g$ from $x$, $y$, and $\dist(x,y)\geq g$. By property (F9.) we have $\girth{F_u}\geq g$ and by property (F10.) the vertex $y_u$ is at distance at least $g$ from $x_u$. Thus by identifying $x_u$'s and $y_u$'s we obtain a graph with girth at least $g$. Moreover, by the property (F8.), in every $F_u$  any two vertices of degree at least $3$ are at distance at least $g$ and by the property (F10.), we have $\dist(x_u,y_u)\geq g$ and $x_u$, $y_u$ are at distance at least $g$ from any vertex of degree at least $3$. Therefore, for any distinct $a,b$ in $A_v$ of degree at least $3$, it holds that $\dist(a,b)\geq g$, $\dist(a,x)\geq g$, and $\dist(a,y)\geq g$. That completes the proof.
\end{proof}

Now we will show the construction of a switching gadget. Let us remind its definition first.

\defSwitch*

Finally, let us proceed to the proof.

\lemSwitchGadget*

\begin{proof}
First, consider the case that there is an induced $C_6$ with consecutive vertices $w_1,\ldots,w_6$ or an induced $C_8$ with consecutive vertices $w_1,\ldots,w_8$ in $H$ such that $\alpha=w_1$, $\beta=w_5$, and $\gamma=w_3$. Define $T$ as a path of length 4 with consecutive vertices $x_1,x_2,x_3,x_4,x_5$ and lists $\{w_1,w_5\}$, $\{w_2,w_4\}$, $\{w_1,w_3,w_5\}$, $\{w_2,w_4\}$, $\{w_1,w_5\}$.
It is straightforward to verify that in both cases $T$ with $p:=x_1,q:=x_3$, and $r:=x_5$ is a switching gadget.

So now we assume that $H$ does not contain an induced $C_6$ or an induced $C_8$ with $\alpha=w_1$, $\beta=w_5$, and $\gamma=w_3$. Then by \cref{obs:walks-between-corners} there exist the following walks: 
\begin{itemize}
\item $\cX': \alpha \to \beta, \cY':\beta \to \alpha$ such that $\cY'$ avoids $\cX'$,
\item $\cX_{\alpha}: \alpha \to \beta, \cY_{\alpha}: \alpha \to \gamma$ and $\cZ_{\alpha}: \beta \to \alpha$, such that $\cX_{\alpha},\cY_{\alpha}$ avoid $\cZ_{\alpha}$ and $\cZ_{\alpha}$ avoids $\cX_{\alpha},\cY_{\alpha}$,
\item $\cX_{\beta}: \alpha \to \alpha, \cY_{\beta}: \alpha \to \gamma$ and $\cZ_{\beta}: \beta \to \beta$, such that $\cX_{\beta},\cY_{\beta}$ avoid $\cZ_{\beta}$ and $\cZ_{\beta}$ avoids $\cX_{\beta},\cY_{\beta}$.
\end{itemize}

\begin{figure}
\centering{\begin{tikzpicture}[every node/.style={draw,circle,fill=white,inner sep=0pt,minimum size=8pt},every loop/.style={}]

\node[fill=gray, opacity=0.6, label=above:\footnotesize{$\{\alpha,\beta\}$}] (x1) at (0,0) {};
\node[fill=gray, opacity=0.6, label=above:\footnotesize{$\{\alpha,\beta,\gamma\}$}] (x2) at (4,0) {};
\node[label=above:\footnotesize{$\{\alpha,\beta\}$}] (x3) at (8,0) {};
\node[fill=gray, opacity=0.6, label=above:\footnotesize{$\{\alpha,\beta\}$}] (x4) at (12,0) {};
\node[draw=none, fill=none, label=below:\footnotesize{$p$}] (x) at (0,-0.1) {};
\node[draw=none, fill=none, label=below:\footnotesize{$q$}] (x) at (4,-0.1) {};
\node[draw=none, fill=none, label=below:\footnotesize{$r$}] (x) at (12,-0.1) {};

\draw (x1)--++(1,1)--++(2,0)--(x2);
\draw (x2)--++(1,1)--++(2,0)--(x3);
\draw (x3)--++(1,1)--++(2,0)--(x4);
\draw (x1)--++(1,-1)--++(2,0)--(x2);
\draw (x2)--++(1,-1)--++(2,0)--(x3);
\draw (x3)--++(1,-1)--++(2,0)--(x4);

\node[color=white,label=left:\footnotesize{$\beta$}] (x12) at (1.2,-0.5) {};
\node[color=white,label=left:\footnotesize{$\alpha$}] (x11) at (1.2,0.5) {};
\node[color=white,label=right:\footnotesize{$\alpha$}] (x13) at (2.7,0.5) {};
\node[color=white,label=right:\footnotesize{$\gamma$}] (x14) at (2.7,0) {};
\node[color=white,label=right:\footnotesize{$\beta$}] (x15) at (2.7,-0.5) {};

\node[color=white,label=right:\footnotesize{$\beta$}] (x24) at (6.7,0.5) {};
\node[color=white,label=right:\footnotesize{$\alpha$}] (x25) at (6.7,-0.5) {};
\node[color=white,label=left:\footnotesize{$\alpha$}] (x21) at (5.2,0.5) {};
\node[color=white,label=left:\footnotesize{$\gamma$}] (x22) at (5.2,0) {};
\node[color=white,label=left:\footnotesize{$\beta$}] (x23) at (5.2,-0.5) {};

\node[color=white,label=left:\footnotesize{$\alpha$}] (x32) at (9.2,-0.5) {};
\node[color=white,label=left:\footnotesize{$\beta$}] (x31) at (9.2,0.5) {};
\node[color=white,label=right:\footnotesize{$\alpha$}] (x33) at (10.7,0.5) {};
\node[color=white,label=right:\footnotesize{$\beta$}] (x34) at (10.7,-0.5) {};

\draw[thick, color=blue] (x11)--(x13);
\draw[thick, color=blue] (x11)--(x14);
\draw[thick, color=blue] (x12)--(x15);
\draw[thick, color=blue] (x21)--(x24);
\draw[thick, color=blue] (x22)--(x25);
\draw[thick, color=blue] (x23)--(x25);
\draw[thick, color=blue] (x31)--(x33);
\draw[thick, dashed, color=blue] (x32)--(x33);
\draw[thick, color=blue] (x32)--(x34);

\node[draw=none,fill=none,label=above:\footnotesize{$\PP(\{\cX_{\beta},\cY_{\beta},\cZ_{\beta}\})$}] (a) at (2,0.1) {};
\node[draw=none,fill=none,label=above:\footnotesize{$\PP(\{\oX_{\alpha},\oY_{\alpha},\oZ_{\alpha}\})$}] (b) at (6,0.1) {};
\node[draw=none,fill=none,label=above:\footnotesize{$\PP(\{\cX',\cY'\})$}] (c) at (10,0.4) {};

\end{tikzpicture}}
\caption{The switching gadget $T$. Blue lines inside each gadget denote possible mappings of the input and the output vertex of the gadget. The dashed line denotes a mapping that might exists but not necessarily.}\label{fig:switching-gadget}
\end{figure}
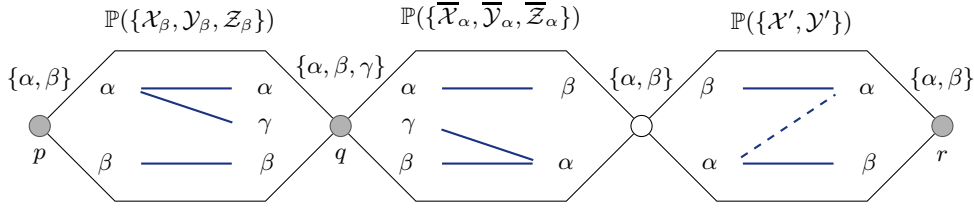

We set as $T$ the composition of gadgets: $\PP(\{\cX_{\beta},\cY_{\beta}, \cZ_{\beta})$, $\PP(\{\oX_{\alpha},\oY_{\alpha},\oZ_{\alpha}\})$, and $\PP(\{\cX',\cY'\})$ (see \cref{fig:switching-gadget}). Let $p_1,\ldots,p_{\ell}$ be consecutive vertices of $T$ and let $p:=p_1,q:=p_s,r:=p_{\ell}$, where $s$ is the number of vertices of $\cX_\beta$. By the definition, $L(p)=L(r)=\{\alpha,\beta\}$ and $L(q)=\{\alpha,\beta,\gamma\}$. Remaining properties of the switching gadget are satisfied by the properties of used walks (see \cref{fig:switching-gadget}).

Finally, observe that we can always assume that in every switching gadget it holds that $\dist(p,q)\geq\frac{g}{2}$ and $\dist(r,q)\geq\frac{g}{2}$. Otherwise, as in the proofs of \cref{lem:relationgadgets} and \cref{lem:assign-gadget} we can introduce paths of appropriate even lengths with lists of consecutive vertices $\{\alpha,\beta\},\{\alpha',\beta'\},\ldots\{\alpha,\beta\}$, where $\alpha',\beta'$ are taken from \cref{obs:walks-between-corners}. We append each path to one of endvertices of the switching gadget. Since the edges $\alpha\alpha'$, $\beta\beta'$ induce a matching in $H$, adding those paths to the gadget preserves its properties. That completes the proof.
\end{proof}

\newpage
\section{Algorithm for $\lhomo H$ parameterized by cutwidth} \label{sec:algo}
In this section we will show how to generalize the algorithm for  $k \coloring$ by Jansen and Nederlof~\cite{DBLP:journals/tcs/JansenN19},
so that it works for all target graphs $H$.
Actually, we will present a more general result, i.e., the algorithm that solves the so-called Binary Constraint Satisfiaction Problem.

Let us start with some definitions.
In the Binary Constraint Satisfaction Problem (BCSP) we are given a triple $I=(V,D,C)$ such that $V$ is the set of variables, $D$ is family of domains $(D_v)_{v \in V}$ and $C$ is the set of binary constraints.
Each constraint $c$ is a tuple $(u,v,S_c)$, where $u,v \in V$ are distinct variables and $S_c \subseteq D_{u} \times D_{v}$.
We will identify the constraints $(u,v,S_c)$ and $(v,u,S'_c)$, where $S'_c = \{ (b,a) ~|~ (a,b) \in S_c\}$, and treat them  as the same constraint.
Thus we can assume that for every pair $u,v$ there is at most one constraint $c \in C$, which contains variables $u,v$.
We have to decide whether there exists a mapping $w$ that assigns a value from $D_v$ to every $v \in V$ so that for each constraint $c=(u,v,S_c) \in C$ it holds that $(w(u),w(v)) \in S_c$.
We will always assume that each domain $D_v$ is a finite subset of $\N_+$, and by $D_{max}$ we will denote the size of a maximum set $D_v \in D$.
 
For an instance $I=(V,D,C)$ of BCSP a \emph{primal graph} $P$ is a graph with vertex set $V$ and there is an edge $uv$ in $P$ if there exists a constraint $c \in C$ containing both variables $u$ and $v$. We also define the \emph{instance graph} $G(I)$ as a graph with vertex set $\{ (v,d) \ | \ v \in V, \ d \in D_v \}$ and there is an edge $(v,d_v)(u,d_u)$ if there is a constraint $c=(u,v,S_c)$ with $(d_u,d_v) \in S_c$. 
Observe that solving an instance $I$ of BCSP is equivalent to finding a copy of the primal graph $P$ in $G(I)$, such that each $u \in V(P)=V$ is chosen from the set $\{(u,d_u) ~|~d_u \in D_u\}$.

Note that an instance $(G,L)$ of $\lhomo{H}$ can be seen as an instance $I=(V,D,C)$ of BCSP, where the set of variables $V=V(G)$, the domain $D_v=L(v)$ for every $v\in V(G)$ and the set of constraints $C$ is the set of all triples $(u,v,S_c)$ such that $uv \in E(G)$ and $(a,b)\in S_c$ if and only if $a\in L(u)$, $b\in L(v)$, and $ab\in E(H)$. The primal graph of this instance is exactly $G$.

For an instance $I=(V,D,C)$ of BCSP we define:
\[
K(I):=\max_{uv \in E(P)} \  \max_{a \in D_v} | \{ b \in D_u \ | \ \exists_{(u,v,S_c)\in C} \ (b,a) \notin S_c\} |.
\]

We aim to show the following theorem, recall that $\omega  < 2.373$ is the matrix multiplication exponent.
\begin{restatable}{theorem}{thmAlgoBcsp}\label{thm:algo-bcsp}
Let $I=(V,D,C)$ be an instance of BCSP and let $P$ be its primal graph, given with the linear ordering $\pi=(v_1,\ldots,v_n)$ of $V$ of width $k$. Then the instance $I$ can be solved in time $2^{\omega \cdot K(I) \cdot k}\cdot (D_{max} \cdot n)^{\Oh(1)}$.
\end{restatable}

\subsection{General setting and auxiliary results}
Let $I=(V,D,C)$ be an instance of BCSP. Let $P$ be its primal graph and let $K:=K(I)$. For every edge $uv \in E(P)$ we define $K$ mappings $\sigma_{uv}^{(i)}: D_v \to D_u\cup\{0\}$, for $i \in [K]$, such that for every $a \in D_v$ it holds that $\big(\sigma_{uv}^{(i)}(a),a\big) \notin S_c$, where $c=(u,v,S_c)\in C$. 

Observe that by the definition of $K(I)$ it is possible that every forbidden pair of values of $u,v$ appears in at least one mapping. Moreover, if for any $uv$ and a value $a \in D_v$ there is no forbidden value for $u$, then we can always set $\sigma^{(i)}_{uv}(a):=0$. Note that for every $uv \in E(P)$, if we assign to $u,v$, respectively the values $x_u \in D_u$ ,$y_v \in D_v$, the constraint $c=(u,v,S_c)\in C$ is satisfied if and only if for every $i \in [K]$ it holds that $x_u \neq \sigma_{uv}^{(i)}(y_v)$.

Let us fix disjoint sets of variables $X,Y \subseteq V$. We define $\X$ to be the the set of all tuples $\x=(x_u)_{u \in X}$, where $x_u \in D_u$ for every $u\in X$. Similarly, $\Y$ is the set of all tuples $\y=(y_v)_{v \in Y}$, where $y_v \in D_v$ for every $v\in Y$.
Let $E\subseteq E(P)$ be the set of edges with one endpoint in $X$ and the other in $Y$. We will treat these edges as directed from $X$ to $Y$, i.e., whenever we write $uv\in E$, we mean that $u \in X$ and $v \in Y$.
Let us define a matrix $M$, whose rows are indexed by tuples $\x \in \X$, columns are indexed by tuples $\y \in \Y$, and the values are:

\begin{equation}
M[\x,\y]:= \prod_{uv \in E} \prod_{i=1}^{K} \big(x_u - \sigma_{uv}^{(i)}(y_v)\big)= \prod_{i=1}^{K} \prod_{uv \in E} \big(x_u - \sigma_{uv}^{(i)}(y_v)\big).\label{eq:matrixM-def}
\end{equation}

Observe that $M[\x,\y] \neq 0$ if and only if assigning the value $x_u$ to every $u \in X$ and $y_v$ to every $v \in Y$ satisfies every constraint that contains one variable from $X$ and one variable from $Y$. We will call such a pair of tuples $(\x,\y)$ a \emph{good} pair.

Fix $i \in [K]$ and let us analyze the factors that appear in equation \cref{eq:matrixM-def}.
To simplify the notation, for a vertex $u$ and a set of edges $Z$, by $\deg_Z(u)$ we denote the number of edges in $Z$ that contain $u$.

\begin{equation}\label{eq:matrixM-factors-i}
\begin{split}
\prod_{uv \in E} (x_u - \sigma_{uv}^{(i)}(y_v))= \sum_{Z \subseteq E}  \Big(\prod_{u \in X} (x_u)^{\deg_{Z}(u)} \Big) \cdot \Big(\prod_{uv \in E \setminus Z} \big(-\sigma_{uv}^{(i)}(y_v)\big) \Big) = \\
= \sum_{(d_u \in \{0,\ldots,\deg_E(u)\})_{u\in X}} \Bigg(\prod_{u \in X} x_u^{d_u} \Bigg)\cdot \Bigg(\sum_{\substack{Z \subseteq E \\ \forall_{u \in X} \deg_{Z}(u)=d_u}} \prod_{uv \in E \setminus Z} \big(-\sigma_{uv}^{(i)}(y_v)\big) \Bigg).
\end{split}
\end{equation}

To simplify the notation, by $\DD$ we will denote the set of all tuples $\dd=(d_u)_{u\in X}$, such that for every $u \in X$ it holds that $d_u \in \{0,\ldots,\deg_E(u)\}$. For $\dd \in \DD$, $\x \in \X$, $\y \in \Y$, and $i \in K$, let us denote:

\begin{equation}\label{eq:matrixM-def-fg}
\begin{split}
 f(\dd, \x) & :=\prod_{u \in X} x_u^{d_u}, \\
 g(i,\dd, \y) &:=\sum_{\substack{Z \subseteq E  \\ \forall_{u \in X} \deg_Z(u)=d_u}} \prod_{uv \in E \setminus Z} \big(-\sigma_{uv}^{(i)}(y_v)\big).
 \end{split}
\end{equation}

For $\od \in \DD^K$, where $\od=(\dd^{(1)},\ldots,\dd^{(K)})$, let us define the matrices $L$ and $R$:
\begin{equation}\label{eq:def-L-R}
\begin{split}
& L[\x,\od]:=\prod_{i=1}^{K} f(\dd^{(i)},\x), \\
& R[\od, \y]:=\prod_{i=1}^{K} g(i,\dd^{(i)},\y).
\end{split}
\end{equation}
The rows of $L$ are indexed by all tuples $\x \in \X$ and the columns are indexed by all tuples $\od \in \DD^K$, and the rows of $R$ are indexed by all tuples $\od \in \DD^K$ and columns by all tuples $\y \in \Y$. Note that the number of columns of $L$ (and thus the number of rows of $R$) is equal to $\big|\DD^K\big|=|\DD|^K=\Big(\prod_{u \in X} (\deg_E(u)+1)\Big)^{K}$.

By applying \cref{eq:matrixM-factors-i}, \cref{eq:matrixM-def-fg}, and \cref{eq:def-L-R} to \cref{eq:matrixM-def}, we can write $M[\x,\y]$ as:
\begin{equation}
\begin{split}
M[\x,\y] & = \prod_{i=1}^{K} \sum_{\dd^{(i)} \in \DD} f(\dd^{(i)},\x) \cdot g(i,\dd^{(i)},\y) = \\
& = \Bigg( \sum_{\dd^{(1)}\in \DD} \big(f(\dd^{(1)},\x) \cdot  g(1,\dd^{(1)},\y) \big) \Bigg) \cdot \ldots \cdot \Bigg( \sum_{\dd^{( K )}\in \DD} f(\dd^{(K)},\x) \cdot  g(K,\dd^{(K)},\y)\Bigg) = \\
& = \sum_{\dd^{(1)}\in \DD}   \ldots \sum_{\dd^{( K )}\in \DD} \Bigg(\prod_{i=1}^{K} f(\dd^{(i)},\x) \cdot g(i,\dd^{(i)},\y)\Bigg) = \\
& = \sum_{\dd^{(1)}\in \DD}   \ldots \sum_{\dd^{( K )}\in \DD} \Bigg(\prod_{i=1}^{K} f(\dd^{(i)},\x)\Bigg) \cdot \Bigg(\prod_{i=1}^K g(i,\dd^{(i)},\y)\Bigg) = \\
& = \sum_{\od\in \DD^K}  \Bigg(\prod_{i=1}^{K} f(\dd^{(i)},\x)\Bigg) \cdot \Bigg(\prod_{i=1}^K g(i,\dd^{(i)},\y)\Bigg) = \sum_{\od\in \DD^K} L[\x,\od]\cdot R[\od,\y].
\end{split}
\end{equation}
Observe that now $M$ is expressed as the product of matrices $L$ and $R$ defined in \cref{eq:def-L-R}.

Now we extend the definition of representing sets, introduced by Jansen and Nederlof~\cite{DBLP:journals/tcs/JansenN19}.

\begin{definition}
Let $S$ be a subset of $\X$ and let $S'\subseteq S$. We say that $S'$ is an \emph{($X$-$Y$)-representative} of S (or \emph{($X$-$Y$)-represents} $S$) if for every tuple $\y \in \Y$ it holds that:

There exists $\x \in S$ such that $(\x,\y)$ is good if and only if there exists $\x' \in S'$ such that $(\x',\y)$ is good.
\end{definition}

In the following lemma we show that for a given $S \subseteq \X$ we can compute a small set $S' \subseteq S$, which ($X$-$Y$)-represents $S$.

\begin{lemma}\label{lem:reduce}
There is an algorithm that for a set $S \subseteq \X$ outputs in time 
\[ \Bigg(\prod_{u \in X} (\deg_E(u)+1)\Bigg)^{K\cdot (\omega-1) } \cdot |S| \cdot \Big(K \cdot \big(|X|+|Y|\big)\Big)^{\Oh(1)}\]
 a set $S' \subseteq S$ such that $|S'|\leq \big(\prod_{u \in X} (\deg_E(u)+1)\big)^{K}$ and $S'$ is an ($X$-$Y$)-representative of $S$.
\end{lemma}

\begin{proof}
We compute $L[S,\bullet]$, i.e., the submatrix of $L$ consisting of the rows of the matrix $L$, whose indices belong to the set $S$. The submatrix $L[S,\bullet]$ can be computed in time $|S|\cdot |\DD|^{K} \cdot \big(K \cdot (|X|+|Y|)\big)^{\Oh(1)}$. We find a row basis of $L[S,\bullet]$, what can be done in time $\Oh(|S| \cdot |\DD|^{K\cdot (\omega-1)})$~\cite{10.2307/2005828,DBLP:journals/jal/IbarraMH82}. Let $S'=\{\x^{(1)},\ldots,\x^{(\ell)}\}$ be the set of indices of the rows from the basis. Since the number of columns of $L$ is $|\DD|^{K}$, the rank of $L$ is at most $|\DD|^{K}$ and hence $|S'| \leq |\DD|^K= \big(\prod_{u \in X} (\deg_E(u)+1)\big)^{K}$.

So let us verify that $S'$ is an ($X$-$Y$)-representative of $S$. Let $\y\in \Y$ and let $\x \in S$. Since $S'$ is the set of indices of the rows from the basis, we can express the row $L[\x,\bullet]$ as a linear combination of the rows indexed by elements of $S'$:

\[
L[\x,\bullet]= \sum_{i=1}^{|S'|} \lambda_i \cdot L[\x^{(i)},\bullet],
\]
for some $\lambda_1,\ldots,\lambda_{|S'|}\in \R$.
Thus $M[\x,\y]$ can be expressed as follows:
\[
M[\x,\y]= \Bigg( \sum_{i=1}^{|S'|} \lambda_i \cdot L[\x^{(i)},\bullet] \Bigg) \cdot R[\bullet,\y]= \sum_{i=1}^{|S'|} \lambda_i \cdot \Big( L[\x^{(i)},\bullet] \cdot R[\bullet,\y] \Big)  = \sum_{i=1}^{|S'|} \lambda_i \cdot M[\x^{(i)},\y].
\]
Now if $(\x,\y)$ is good, then $M[\x,\y]$ is non-zero and thus there must be $\x^{(i)} \in S'$ such that $M[\x^{(i)},\y]$ is non-zero. That means that for every $\y \in \Y$, if there is $\x \in S$ such that $(\x,\y)$ is good, then there exists $\x^{(i)} \in S'$, such that $(\x^{(i)},\y)$ is good. Therefore, $S'$ is an ($X$-$Y$)-representative of $S$.
\end{proof}

\subsection{Algorithm} \label{sec:algo-meat}

Let $I=(V,D,C)$ be an instance of BCSP and let $P$ be its primal graph. Fix an ordering $v_1,\ldots,v_n$ of the variables in $V$. For every $i\in[n]$, let us define:
\begin{equation}
\begin{split}
V_i &:= \{v_1,\ldots,v_i\}, \\
X_i & := \{ u \in V_i \ | \ \exists_{v \in V\setminus V_i} \ uv \in E(P) \}, \\
Y_i & := \{ v \in V\setminus V_i \ | \ \exists_{u \in V_i} \ uv \in E(P) \}.
\end{split}
\end{equation}
We also set $V_0$, $X_0$, and $Y_0$ as empty sets. Note that for every $i\in [n]$ it holds that if $j<i$ and $v_j \in X_i$, then $v_j \in X_{i-1}$, so $X_i \subseteq X_{i-1}\cup\{v_i\}$. Similarly, $Y_{i-1}\subseteq Y_i\cup \{v_i\}$.

Let $V' \subseteq V$. We say that a tuple $(x_u)_{u\in V'}$ is \emph{good} (on $V'$) if $x_u \in D_u$ for every $u\in V'$ and for every constraint $(u_1,u_2,S_c)\in C$, where $u_1,u_2 \in V'$, it holds that $(x_{u_1},x_{u_2})\in S_c$. For $V''\subseteq V'$ and a tuple $\x=(x_u)_{u\in V'}$, by $\x|_{V''}$ we denote the tuple $\x'=(x'_u)_{u\in V''}$ such that $x_u'=x_u$ for every $u\in V''$. By $\X_i$ we denote the set of all tuples $(x_u)_{u\in X_i}$ such that $x_u \in D_u$ for every $u\in X_i$ and by $\Y_i$ we denote the set of all tuples $(y_v)_{v\in Y_i}$ such that $y_v \in D_v$ for every $v\in Y_i$. For every $i \in [n]\cup\{0\}$, by $T[i]$ we denote the set of all tuples $\x=(x_u)_{u \in X_i}$, which can be extended to a good tuple on $V_i$. In particular, $T[0]=\{\emptyset\}$, where $\emptyset$ denotes the $0$-tuple. In the following lemma we show how to construct for every $i \in [n]$ a set $T'[i]\subseteq T[i]$, which is an ($X_i$-$Y_i$)-representative of $T[i]$.

\begin{lemma}\label{lem:representing-set}
Let $i\in [n]$ and let $T'[i-1]\subseteq T[i-1]$ be an ($X_{i-1}$-$Y_{i-1}$)-representative of $T[i-1]$. We define:
\begin{align*}
T'[i]:=\Bigl \{ & \x \in \X_i \ | \ \exists_{\x'\in T'[i-1]} \ \exists_{a \in D_{v_i}} \\
& \Big(\big( \x|_{X_{i-1}\cap X_i}=\x'|_{X_{i-1}\cap X_i}\big) \ 
 \land \big( \forall_{(v_j,v_i,S_c) \in C: \ v_j\in X_{i-1}} \   
(x'_{v_j},a)\in S_c\big) \land \big( v_i \in X_i \ \Rightarrow x_{v_i}=a \big) \Big) \Bigr \}.
\end{align*}
Then $T'[i] \subseteq T[i]$ and it is an ($X_i$-$Y_i$)-representative of $T[i]$. Moreover, $T'[i]$ can be computed in time $|T'[i-1]| \cdot |D_{v_i}|\cdot n^{\Oh(1)}$.
\end{lemma}

\begin{proof}
First, let us show that $T'[i]$ is a subset of $T[i]$.
Let $\x=(x_u)_{u\in X_i}$ be a tuple from $T'[i]$.
Let $\x'=(x'_u)_{u \in X_{i-1}}\in T'[i-1]$ and $a \in D_{v_i}$ be the values that justify that $\x \in T'[i]$, i.e., they satisfy:
\begin{enumerate}[a)]
\item $\x|_{X_{i-1}\cap X_i}=\x'|_{X_{i-1}\cap X_i}$,
\item for every constraint $(v_j,v_i,S_c)\in C$ with $v_j\in X_{i-1}$, it holds that $(x'_{v_j},a)\in S_c$,
\item if $v_i \in X_i$, then $x_{v_i}=a$.
\end{enumerate}
Since $T'[i-1]\subseteq T[i-1]$, the tuple $\x'$ can be extended to a good tuple $\z'$ on $V_{i-1}$.
Recall that $X_i\setminus \{v_i\}\subseteq X_{i-1}$.
Moreover, for every $v_j \in X_{i-1}$ such that $(v_j,v_i,S_c)\in C$ it holds that $(x'_{v_j},a)\in S_c$.
Thus $\x$ can be extended to a good tuple $\z=(z_u)_{u\in V_i}$, such that $z_u:=z'_u$ for every $u\in V_{i-1}$ and $z_{v_i}=a$, and therefore $\x \in T[i]$.

Now it remains to show that $T'[i]$ is an ($X_i$-$Y_i$)-representative of $T[i]$. Let $\y=(y_v)_{v\in Y_i}$ be a tuple in $\Y_i$ and suppose that there exists a tuple $\x=(x_u)_{u\in X_i}\in T[i]$, such that the pair $(\x,\y)$ is good. Let $\z=(z_u)_{u\in V_i}$ be a good tuple that extends $\x$ on $V_i$. Note that the pair $(\z,\y)$ is good. Let us define $\y'=(y'_v)_{v\in Y_{i-1}}$ such that $\y'|_{Y_{i-1}\cap Y_i}=\y|_{Y_{i-1}\cap Y_i}$ and if $v_i \in Y_{i-1}$, then $y'_{v_i}:=z_{v_i}$. Note that $\y'$ is well-defined, as $Y_{i-1}\subseteq Y_i\cup \{v_i\}$. Define $\z':=\z|_{X_{i-1}}$. Now observe that the pair $(\z',\y')$ is good. Moreover, by the definition, the tuple $\z'$ can be extended to a good tuple on $V_{i-1}$, so $\z' \in T[i-1]$. 

Recall that $T'[i-1]$ is an ($X_{i-1}$-$Y_{i-1}$)-representative of $T[i-1]$. Thus there exists a tuple $\x''=(x''_u)_{u \in X_{i-1}}\in T'[i-1]$, such that the pair $(\x'',\y')$ is good. Define $\x'=(x'_u)_{u\in X_i}$ such that $\x'|_{X_{i-1}\cap X_i}:=\x''|_{X_{i-1}\cap X_i}$ and if $v_i \in X_i$, then we set $x'_{v_i}:=z_{v_i}$. By the definition of $T'[i]$ it holds that $\x'\in T'[i]$. Moreover, the pair $(\x',\y)$ is good. Therefore, $T'[i]$ is an ($X_i$-$Y_i$)-representative of $T[i]$, which completes the proof.
\end{proof}

Now we are ready to prove \cref{thm:algo-bcsp}.

\begin{proof}[Proof of \cref{thm:algo-bcsp}]
Let $I=(V,D,C)$ be an instance of BCSP, let $P$ be its primal graph and let $K:=K(I)$. Let $\pi=(v_1,\ldots,v_n)$ be a linear layout of vertices of $P$ of width $k$. For every $i \in [n]$, by $E_i$ we denote the set of edges in $P$ with one endpoint in $X_i$ and the other in $Y_i$.

For every $i\in [n]$, we will construct a set $T'[i]$ that is an ($X_i$-$Y_i$)-representative of $T[i]$. Note that the sets $X_n$ and $Y_n$ are empty. Thus $T[n]$, which is the set of all tuples $(x_u)_{u\in X_n}$ that can be extended to a good tuple $(x_u)_{u\in V_n}$, is either empty or contains a $0$-tuple $\emptyset$. The latter one holds if and only if there exists an assignment of values to all variables in $V$ that satisfies every constraint in $C$. Therefore, the instance $(V,D,C)$ is a yes-instance of BCSP if and only if $T[n]$ is non-empty. Moreover, the set $T[n]$ is non-empty if and only if its representing set $T'[n]$ is non-empty. So in order to solve the instance $(V,D,C)$, it is sufficent to compute a set $T'[n]$ that is an ($X_n$-$Y_n$)-representative of $T[n]$.

Recall that $T[0]=\{\emptyset\}$ and thus we set $T'[0]:=\{\emptyset\}$.
For every $i\in [n]$ we proceed as follows.
Since we have already computed a set $T'[i-1]$, which is an ($X_{i-1}$-$Y_{i-1}$)-representative of $T[i-1]$,
we can call \cref{lem:representing-set} to construct a set $T'[i]$, that is an ($X_i$-$Y_i$)-representative of $T[i]$.
Then we call \cref{lem:reduce} for $S = T'[i]$ to obtain another set, $T''[i]$, that is an ($X_i$-$Y_i$)-representative of $T'[i]$
and its size is at most $\big(\prod_{u \in X_i} (\deg_{E_i}(u)+1)\big)^{K}$.
Since the relation of representing is transitive, $T''[i]$ is an ($X_i$-$Y_i$)-representative of $T[i]$.
We replace $T'[i]$ with $T''[i]$ and proceed to the next value of $i$. 

Let $\Lambda:=\max_{i \in [n]} \big(\prod_{u \in X_i} (\deg_{E_i}(u)+1)\big)$. Observe that if the width of $\pi=(v_1,\ldots,v_n)$ is $k$, then for every $i \in [n]$, it holds that $\prod_{u \in X_i} (\deg_{E_i}(u)+1) \leq 2^k$ by the AM-GM inequality. Therefore, $\Lambda \leq 2^k$.

Computing $K(I)$ can be done in time $(D_{max}\cdot n)^{\Oh(1)}$. Every $T'[i]$ from \cref{lem:representing-set} can be computed in time $|T''[i-1]| \cdot |D_{v_i}|\cdot n^{\Oh(1)}$. Since every set $T''[i]$ was obtained by \cref{lem:reduce}, $|T''[i-1]|\leq \Lambda^{K}$ and the set $T'[i]$ from \cref{lem:representing-set} can be computed in time $\Lambda^K \cdot D_{max} \cdot n^{\Oh(1)}$. The time of applying \cref{lem:reduce} to every $T'[i]$ is at most $\Lambda^{K\cdot (\omega-1) } \cdot \Lambda^K \cdot D_{max} \cdot (K\cdot n)^{\Oh(1)} = \Lambda^{K\cdot \omega} \cdot (D_{max} \cdot n)^{\Oh(1)}$. So the total time is at most $\Lambda^{K\cdot \omega} \cdot (D_{max} \cdot n)^{\Oh(1)} \leq 2^{k \cdot K\cdot \omega} \cdot (D_{max} \cdot n)^{\Oh(1)}$. That completes the proof.
\end{proof}
\subsection{Application: DP-coloring}\label{sec:DP}
Let us start with a formal definition of the DP-coloring problem.
A \emph{cover} a of graph $G$ is a pair $\mathcal{H}=(L,H)$, where $H$ is a graph and $L : V(G) \to 2^{V(H)}$ is a function, satisfying the following:
\begin{enumerate}[(1)]
\item the family $\{ L(v) ~|~v \in V(G)\}$ is a partition of $V(H)$,
\item for every $v \in V(G)$, the graph $H[L(v)]$ is complete,
\item for every $uv \in E(G)$, the set of edges joining the sets $L(u)$ and $L(v)$ in $H$ is a matching, \label{cond:matching}
\item for every $uv \notin E(G)$, there are no edges in $H$ with one endpoint in $L(u)$ and the other in $L(v)$.
\end{enumerate}
We are interested in determining the existence of an \emph{$\mathcal{H}$-coloring} of $G$, which is an independent set in $H$ of size $|V(G)|$.
Note that this independent set corresponds to choosing for each $v \in V(G)$ one vertex (color) in $L(v)$, so that no two adjacent vertices are mapped to the neighbors in $H$.

Let us show that this problem is a special case of BCSP. Indeed, for an instance $(G,\mathcal{H}=(L,H))$ of DP-coloring, 
let us define an instance $I=(V,D,C)$ as follows.
Let $V = V(G)$. For each $v \in V$, we set $D_v := L(v)$.
For every $uv \in E(G)$, we add a constraint $(u,v,S_c)$, where $S_c = \{ (a,b) \in L(u) \times L(v) ~|~ ab \notin E(H)\}$.
It is straightforward to verify that $G$ admits an $\mathcal{H}$-coloring if and only if $I$ is a yes-instance of BCSP.
Furthermore, the primal graph $P$ of the instance $I$ is exactly the graph $G$, and $K(I) = 1$, by condition \cref{cond:matching} in the definition of a cover. 
Thus \cref{thm:algo-bcsp} immediately yields the following corollary.

\begin{corollary}
Let $G$ be a graph given with a linear ordering of vertices of width $k$.
Every instance $(G,\mathcal{H}=(L,H))$  of DP-coloring can be solved in time $2^{\omega \cdot k}\cdot (|V(H)| \cdot |V(G)|)^{\Oh(1)}$.
\end{corollary}

\subsection{Application: $\lhomo{H}$}

Recall that every instance $(G,L)$ of $\lhomo{H}$ can be seen as an instance $I=(V,D,C)$ of BCSP, where $V=V(G)$, for each $v \in V(G)$ we have $D_v = L(v)$, and $C$ consists of all tuples $(u,v,S_c)$, where $uv \in E(G)$, and $S_c = \{(a,b)\in L(u) \times L(v) ~|~ ab \in E(H)\}$.
The primal graph of $I$ is precisely $G$. 

Our algorithm is based on the following result of Okrasa~\emph{et al.}~\cite{LhomoTreewidthFull}.

\begin{theorem}[Okrasa~\emph{et al.}~\cite{LhomoTreewidthFull}]\label{thm:factors}
Let $H$ be a graph.
In time $|V(H)|^{\Oh(1)}$ we can construct a family $\cH$ of $\Oh(|V(H)|)$ connected graphs such that:
\begin{enumerate}[(1)]
\item $H$ is a bi-arc graph if and only if every $H' \in \cH$ is a bi-arc graph, \label{factors:hard}
\item if $H$ is bipartite, then each $H' \in \cH$ is an induced subgraph of $H$, and is either the complement of a circular-arc graphs or is undecomposable, \label{factors:bipartite}
\item otherwise, for each $H' \in \cH$, the graph $H'^*$ is an induced subgraph of $H^*$ and at least one of the following holds: \label{factors:types}
\begin{enumerate}
\item $H'$ is a bi-arc graph, or
\item the vertex set of $H'$ can be partitioned into two sets $P,B$, such that $P$ induces a reflexive clique and $B$ is independent\footnote{The  statement of this condition in \cite{LhomoTreewidthFull} is more involved, but the simpler version we present here is sufficient for our application.}, or \label{factors:types:strongsplit}
\item $(H')^*$ is undecomposable, \label{factors:types:undecomposable}
\end{enumerate}
\item for every instance $(G,L)$ of \lhomo{H} with $n$ vertices, the following implication holds: \label{factors:bottomup}

If there exists a non-decreasing, convex function $f_H \colon \N \to \R$,
such that for every $H' \in \cH$, for every induced subgraph $G'$ of $G$, and for every $H'$-lists $L'$ on $G'$,
we can decide whether $(G',L')\to H'$ in time $f_H(|V(G')|)$, then
we can solve the instance $(G,L)$  in time
\[
\Oh \left (|V(H)| f_H(n) + n^2 \cdot |V(H)|^3 \right).
\]
\end{enumerate}
\end{theorem}

The graphs in the family $\cH$ are called \emph{factors} of $H$.
So in order to solve the \lhomo{H} problem, it is sufficient to give an algorithm for \lhomo{H'} for every factor $H'$ of $H$.

Before we proceed to the proof, let us discuss first a special case when $H$ is bipartite.
Then we will lift this result to all target graphs.

\paragraph{Bipartite target graphs.}
Let $H$ be bipartite and let $\cH$ be the family of its factors.
Let $(G,L)$ be an instance of \lhomo{H}, where $G$ is given with a linear layout of width $k$.
Consider $H' \in \cH$.
If $H'$ is the complement of a circular-arc graph, then we can solve the \lhomo{H'} problem in polynomial time.
Otherwise, by \cref{thm:factors}~\cref{factors:bipartite}, we know that $H'$ is a connected induced subgraph of $H$, and it is undecomposable.

Consider an instance $(G',L')$ of \lhomo{H'}, where $G'$ is an induced subgraph of $G$.
Clearly, a linear layout of $G$ with width $k$ induces a linear layout of $G'$ with width at most $k$.
Let $I$ be the instance of BCSP corresponding to $(G',L')$. By \cref{thm:algo-bcsp} we can solve it in time $2^{K(I) \cdot \omega \cdot k} \cdot (|V(G')| \cdot |V(H')|)^{\Oh(1)}$. 

Let us estimate the value of $K(I)$.
Recall that without loss of generality we can assume that $(G',L')$ is consistent.
Furthermore, for every edge $uv$ of $G'$ and every $a \in L'(u)$, we may assume that $a$ is adjacent in $H'$ to some vertex in $L'(v)$,
as otherwise we can safely remove $a$ from $L'(u)$.
Thus $K(I)$ is upper-bounded by $\gamma(H')$, which is defined as the maximum over pairs of incomparable sets $S_1,S_2 \subseteq V(H')$,
each contained in a different bipartition class of $H'$, such that for every $x \in S_1$ there is $y \in S_2 \cap N_{H'}(x)$,
and for every $y \in S_2$ there is $x \in S_1 \cap N_{H'}(y)$, of the value
\[ \max_{x \in S_1} |\{y ~|~xy \notin E(H')\}|.\]

So we conclude that every instance $(G',L')$ of \lhomo{H'}, where $G'$ is an induced subgraph of $G$, can be solved in time $2^{\gamma(H') \cdot \omega \cdot k} \cdot (|V(G')| \cdot |V(H')|)^{\Oh(1)}$.

Now let us define a new parameter $\gamma^*(H)$.
If the complement of $H$ is not a circular-arc graph, we define $\gamma^*(H)$ as the maximum value of $\gamma(H')$ over all connected undecomposable induced subgraphs $H'$ of $H$, which are not the complement of a circular-arc graph.
If $H$ is the complement of a circular-arc graph, we define $\gamma^*(H)=\gamma(H):=1$.

Now observe that for every $H' \in \cH$, every instance $(G,L')$ of \lhomo{H'}, where $G'$ is an induced subgraph of $G$, can be solved in time $2^{\gamma^*(H) \cdot \omega \cdot k} \cdot (|V(G')| \cdot |V(H')|)^{\Oh(1)}$.
As this function is convex and non-decreasing, \cref{thm:factors} yields the following.

\begin{corollary}\label{cor:algo-bipartite}
Let $H$ be a bipartite graph and let $(G,L)$ be an instance of $\lhomo{H}$, where $G$ is given with a linear layout of width $k$.
Then $(G,L)$ can be solved in time $2^{\gamma^*(H) \cdot \omega \cdot k} \cdot (|V(G)| \cdot |V(H)|)^{\Oh(1)}$.
\end{corollary}

\paragraph{General target graphs.}
Let $H$ be a non-bi-arc graph. As usual, we extend the definition of $\gamma^*$ to all graphs by setting $\gamma^*(H) := \gamma^*(H^*)$. 

Let $\cH$ be the family of factors of $H$ and let $(G,L)$ be an instance of \lhomo{H}, where $G$ is given with a linear layout with width $k$.
Again, we need to solve \lhomo{H'} for every $H' \in \cH$ on instances $(G',L')$, where $G'$ is an induced subgraph of $G$.
Recall that the linear layout for $G$ with width $k$ induces a linear layout $\sigma$ for $G'$ with width at most $k$.
Let us discuss the complexity of solving $(G',L')$.
Consider three cases corresponding to the options in \cref{thm:factors}~\cref{factors:types}.

\begin{enumerate}[(a)]
\item If $H' \in \cH$ is a bi-arc graph, then every instance of $\lhomo{H'}$ can be solved in polynomial time.

\item If the vertex set of $H'$ can be partitioned into a reflexive clique $P$ and an independent set $B$, we follow the argument by Okrasa \emph{et al.}~\cite{LhomoTreewidth}.
Let $\tH$ be the bipartite graph obtained from $H'$ by removing all edges with both endpoints in $P$ (including loops).
We note that $\tH$ is an induced subgraph of $(H')^*$, and thus also an induced subgraph of $H^*$, so $\gamma^*(\tH) \leq \gamma^*(H)$.

Observe that for every $p \in P$ and $b \in B$ it holds that $N_{H'}(b) \subseteq N_{H'}(p)$.
Since we may assume that each list is a non-empty incomparable set, we can partition the vertex set of $G'$ into two subsets: $P' := \{v ~|~ L'(v) \cap P \neq \emptyset\}$ and $Q' := \{v ~|~ L'(v) \cap Q \neq \emptyset\}$.
Observe that if $Q'$ is not independent, then we can immediately report that $(G',L')$ is a no-instance.
So let $G''$ be the graph obtained from $G'$ by removing all edges with both endpoints in $P'$.
Clearly $G''$ is bipartite, and, as observed by Okrasa \emph{et al.}~\cite{LhomoTreewidth}, $(G',L')$ is a yes-instance of $\lhomo{H'}$ if and only if $(G'',L')$ is a yes-instance of $\lhomo{\tH}$.
Finally, as that $G''$ was obtained from $G'$ by deleting edges, $\sigma$ is a linear layout for $G''$ width width at most $k$.
So applying \cref{cor:algo-bipartite} to $(G'',L')$, we conclude that the instance $(G',L')$ can be solved in time
\[2^{\gamma^*(\tH) \cdot \omega \cdot k} \cdot (|V(G'')| \cdot |V(\tH)|)^{\Oh(1)} \leq 2^{\gamma^*(H) \cdot \omega \cdot k} \cdot (|V(G')| \cdot |V(H')|)^{\Oh(1)}.\]

\item Finally, suppose that $H'$ is a connected non-bi-arc graph and $(H')^*$ is an undecomposable induced subgraph of $H^*$. Observe that if $H'$ is bipartite, then we are done by  \cref{cor:algo-bipartite}. So let us assume otherwise, so in particular $H'^*$ is connected.
Then $\gamma^*(H')=\gamma^*((H')^*)=\gamma((H')^*)\leq \gamma^*(H^*)=\gamma^*(H)$.

Consider an instance $(G',L')$ of $\lhomo{H'}$.
Following Feder \emph{et al.}~\cite{DBLP:journals/jgt/FederHH03}, we define an \emph{associated instance} $(G'^*,L'^*)$ of $\lhomo{H'^*}$, so that for $v \in V(G')$ and $x \in V(H')$ it holds that $x \in L'(v)$ if and only if  $x'\in L'^*(v')$ if and only if $x''\in L'^*(v'')$.
A list homomorphism $\vphi:(G'^*,L'^*)\to H'^*$ is \emph{clean} if for every $v\in V(G')$ and $x\in V(H')$,
it holds that $\vphi(v')=x'$ if and only if $\vphi(v'')=x''$.
As observed by  Okrasa \emph{et al.}~\cite{LhomoTreewidth} (although the original idea comes from  Feder \emph{et al.}~\cite{DBLP:journals/jgt/FederHH03}), it holds that $(G',L')\to H'$ if and only if $(G'^*,L'^*)$ admits a clean homomorphism to $H^*$.

So let us solve the instance $(G',L')$ of $\lhomo{H'}$ by looking for a clean homomorphism from $(G'^*,L'^*)$ to $H'^*$.
We need to adapt the algorithm given in \cref{sec:algo-meat}. As the adaptation is rather straightforward and technical,
we will just point out the differences to the version presented above.

First, observe that if $G'$ is given with a linear layout $\sigma=(v_1,\ldots,v_{|G'|})$ of width at most $k$,
then in polynomial time we can construct the linear layout $\sigma^*=(v_1',v_1'',\ldots,v_{|G'|}',v_{|G'|}'')$ of $G'^*$ with width at most $2k$.
We compute the sets $T'[i]$ similarly as in the proof of \cref{thm:algo-bcsp}, but this time by $V_i$ we denote the set $\{v_1',v_1'',\ldots,v_i',v_i''\}$, i.e., we either include both $v_i',v_i''$, or none of them.
The  definitions of the other sets, i.e., $X_i$, $Y_i$, $T[i]$, and $T'[i]$, are updated in an analogous way.

Moreover, as we are looking for clean homomorphisms, we will only consider the colorings of tuples $\x$ and $\y$, such that vertices $v_j',v_j''$ are mapped, respectively, to $x'$ and $x''$ for some $x \in V(H)$.
Thus the size of the set obtained by \cref{lem:representing-set} does not increase.
Finally, when we apply \cref{lem:reduce} to $T'[i]$, and thus to some set of colorings of $X_i$, it is enough to construct a matrix $M[\x,\y]$, so that $\x$ is a coloring of vertices in $X_i \cap \{v' \ | \ v \in V(G'\}$ and $\y$ is a coloring of vertices in $Y_i\cap \{v'' \ | \ v \in V(G')\}$, as they imply the colorings of all vertices in $X_i$ and $Y_i$.
Therefore, although the upper bound for the width of the linear layout $\sigma^*$ is $2k$,
we only consider half of the edges crossing a cut.
We conclude that the time of finding a clean homomorphism from $(G'^*,L'^*)$ to $H'^*$, and thus solving the instance $(G',L')$ of $\lhomo{H'}$, is at most $2^{k\cdot \gamma(H'^*) \cdot \omega} \cdot (|V(H'^*)| \cdot |V(G'^*)|)^{\Oh(1)}\leq 2^{k\cdot \gamma^*(H) \cdot \omega} \cdot (|V(H')| \cdot |V(G')|)^{\Oh(1)}$.
\end{enumerate}

Similarly to the case that $H$ is bipartite, by \cref{thm:factors} we obtain the following corollary.

\algolhomo*

\newpage
\section{Conclusion}
\subsection{Comparison of parameters}
In this section we will compare parameters $mim^*(H)$, $i^*(H)$, and $\gamma^*(H)$. We will only consider connected, bipartite, undecomposable graphs $H$, whose complement is not a circular-arc graph, and thus only the parameters $mim(H)$, $i(H)$, $\gamma(H)$, as any inequalities for $i(H)$, $mim(H)$, and $\gamma(H)$ imply the same inequalities for $i^*$, $mim^*$, and $\gamma^*$ for general target graphs. 

First let us show that $mim(H)-1 \leq \gamma(H) \leq i(H)-1$. To see the first inequality, consider a strongly incomparable set  $S_1 \subseteq V(H)$, contained in one bipartition class, such that $|S_1|=mim(H)$. Let $S_2\subseteq V(H)$ be a set of private neighbors of $S_2$, i.e., the set $S_1\cup S_2$ induces a matching in $H$. Let $s \in S_1$. The number of vertices in $S_2$ non-adjacent to $s$ is $|S_2|-1=mim(H)-1$. Since $S_1$, $S_2$ are both incomparable sets, each contained in different bipartition class of $H$, for every $s_1 \in S_1$ it holds that $N(s_1)\cap S_2 \neq \emptyset$, and for every $s_2 \in S_2$ it holds that $N(s_2)\cap S_1 \neq \emptyset$, we conclude that $\gamma(H)\geq mim(H)-1$. The second inequality follows from the fact that the sets $S_1$, $S_2$ from the definition of $\gamma(H)$ are incomparable and for every $s \in S_1$ (resp. $S_2$) at least one vertex in $S_2$ (resp. $S_1$) is adjacent to $s$.

Now let us show that the differences between $mim(H)$ and $\gamma(H)$, and between $\gamma(H)$ and $i(H)$ can be arbitrarily large. First, consider $H$ which is a biclique $K_{r,r}$ with a perfect matching removed. Note that if $r\geq3$, then $H$ contains an induced $C_6$ and thus $H$ is not a complement of a circular-arc graph. Moreover, $H$ is undecomposable and every bipartition class of $H$ is an incomparable set. Therefore $i(H)=r$. On the other hand, for every $v \in V(H)$ there is only one vertex in the other bipartition class that is non-adjacent to $v$, and hence $\gamma(H)= 1$.

To show that $\gamma(H)$ might be arbitrarily larger than $mim(H)$, we start with a biclique $K_{r+1,r+1}$ and again we remove from the graph a perfect matching. Let $u_1,u_2$ be non-adjacent vertices from different bipartition classes. We add to the graph two new vertices, $v_1$, $v_2$, and we add edges $u_1v_2$, $v_1v_2$, and $v_1u_2$. That completes the construction of $H$. It can be verified that $H$ is undecomposable and if $r+1\geq 3$, then $H$ is not a complement of a circular-arc graph. Moreover, both bipartition classes of $H$ are incomparable sets, and the number of vertices in the other bipartition class than $v_1$ that are non-adjacent to $v_1$ is $r$. Thus $\gamma(H)\geq r$. The size of any induced matching in $H$ is at most three, since there could be at most two edges from the biclique and at most one of three added edges  $u_1v_2$, $v_1v_2$, $v_1u_2$. 

Finally, let us point out that although we have the inequality $\gamma(H)\leq i(H)-1$ and the difference between $\gamma(H)$ and $i(H)$ can be arbitrarily large, for some $H$ it holds that $i(H) \leq 2^{\omega \cdot \gamma(H)}$. Indeed, in the second example of $H$ we have $i(H)=r+2$ and $\gamma(H)\geq r$, so for $r\geq 2$, it holds that $i(H)<2^{\omega \cdot \gamma(H)}$. Therefore, our algorithm solving $\lhomo{H}$ in time $2^{\omega \cdot \gamma(H)\cdot \ctw{G}} \cdot n^{\Oh(1)}$ and the algorithm from \cite{LhomoTreewidth} that solves $\lhomo{H}$ in time $i(H)^{\tw{G}}\cdot n^{\Oh(1)} \leq i(H)^{\ctw{G}}\cdot n^{\Oh(1)}$ are incomparable.

\subsection{Further research directions}

As a main problem of the paper, we were investigating the fine-grained complexity of the \lhomo{H} problem, parameterized by the cutwidth of the instance graph. We provided a lower bound and two upper bounds, incomparable to each other. A natural open question is to close the gap between lower and upper bounds, and provide a full complexity classification.

As a concrete problem, we believe that a good starting point is to understand the complexity of $\lhomo{C_k}$, where $k \geq 5$.
Recall that we have a lower bound $(mim^*(C_k))^{\ctw{G}} \cdot |V(G)|^{\Oh(1)}$ and an upper bound $(i^*(C_k))^{\ctw{G}}  \cdot |V(G)|^{\Oh(1)}$ (the bound from \cref{thm:alglhomo} is worse in this case). The value of $mim^*(C_k)$ is $\lfloor k/3 \rfloor$  if $k$ is even, and $\lfloor 2k/3 \rfloor$ is $k$ is odd.
On the other hand, $i^*(C_k)$ is $k/2$ is $k$ is even, and $k$ if $k$ is odd.
Where does the truth lie?
To be even more specific, what is the complexity of $\lhomo{C_6}$?

Another research direction that we find exciting is to study the complexity of \homo{H} and \lhomo{H}, depending on different parameters of the instance graph. In particular, Lampis~\cite{DBLP:conf/icalp/Lampis18} showed that $k \coloring$ on a graph $G$ can be solved in time $(2^k-2)^{\textrm{cw}(G)} \cdot |V(G)|^{\Oh(1)}$, where $\textrm{cw}(G)$ is the \emph{clique-width} of $G$. Furthermore, an algorithm with a running time  $(2^k-2-\epsilon)^{\textrm{cw}(G)} \cdot |V(G)|^{\Oh(1)}$, for any $\epsilon >0$, would contradict the SETH. We believe it is exciting to investigate how these results generalize to non-complete target graphs $H$.

\newpage
\bibliographystyle{plain}
\bibliography{main}

\end{document}